\documentclass[article]{IEEEtran}
\usepackage{amsmath,amssymb,latexsym,cite}
\usepackage{epsf, pstricks,pgf, xcolor}

\def\etal{\textit{et al.~}}

\def\argmax{\operatornamewithlimits{arg\,max}}

\def\ith{i^{\text{th}}}

\def\ellth{\ell^{\text{th}}}
\def\mth{m^{\text{th}}}

\def\e{\ell}
\def\L{\Lambda}

\def\modl{\hspace{-0.1in}\mod\Lambda}
\def\modbl{\hspace{-0.1in}\mod \beta \Lambda}
\def\modz{\hspace{-0.1in}\mod\Zbb^n}
\def\modp{\hspace{-0.1in}\mod p}
\def\modpz{\hspace{-0.1in}\mod p\Zbb^n}

\def\Re{\mathsf{Re}}
\def\Im{\mathsf{Im}}

\def\rcov{r_{_{\text{COV}}}}
\def\reffec{r_{_{\text{EFFEC}}}}
\def\ammse{\alpha_{_{\text{MMSE}}}}

\def\ab{\mathbf{a}}
\def\Ab{\mathbf{A}}
\def\bb{\mathbf{b}}
\def\Bb{\mathbf{B}}
\def\cb{\mathbf{c}}
\def\db{\mathbf{d}}
\def\Gb{\mathbf{G}}
\def\hb{\mathbf{h}}
\def\Hb{\mathbf{H}}
\def\Lb{\mathbf{L}}

\def\Qb{\mathbf{Q}}
\def\sb{\mathbf{s}}

\def\tb{\mathbf{t}}
\def\ub{\mathbf{u}}

\def\ubh{\mathbf{\hat{u}}}
\def\vb{\mathbf{v}}
\def\vbh{\mathbf{\hat{v}}}
\def\wb{\mathbf{w}}
\def\wbh{\mathbf{\hat{w}}}
\def\xb{\mathbf{x}}

\def\yb{\mathbf{y}}
\def\zb{\mathbf{z}}

\def\Bm{\mathcal{B}}
\def\Vm{\mathcal{V}}
\def\Em{\mathcal{E}}

\def\Dm{\mathcal{D}}
\def\Cm{\mathcal{C}}

\def\Rbb{\mathbb{R}}
\def\Cbb{\mathbb{C}}
\def\Fbb{\mathbb{F}}
\def\Zbb{\mathbb{Z}}
\def\ZC{\{\Zbb + j \Zbb\}}

\def\onen{\frac{1}{n}}
\def\onehalf{\frac{1}{2}}

\def\V{\mathsf{Vol}}
\def\n{^{-1}}

\newcommand{\snr}{\text{$\mathsf{SNR}$}}

\newtheorem{theorem}{Theorem}

\newtheorem{lemma}{Lemma}

\newtheorem{definition}{Definition}
\newtheorem{example}{Example}
\newtheorem{remark}{Remark}
 
\begin{document}

\title{Compute-and-Forward: Harnessing Interference through Structured Codes}
\author{Bobak Nazer, \IEEEmembership{Member, IEEE} and Michael Gastpar, \IEEEmembership{Member, IEEE}
\thanks{This work was supported by the National Science Foundation under grants CCR 0347298, CNS 0627024, and CCF 0830428 as well a Graduate Research Fellowship. M. Gastpar was also supported by the European Research Council under grant ERC StG 259530-ComCom. The material in this paper was presented in part at the IEEE International Symposium on Information Theory, Toronto, Canada, July 2008 and at the 42nd Annual IEEE Asilomar Conference on Signals, Systems, and Computers, Monterey, CA, October 2008.} 
\thanks{B. Nazer is with the Department of Electrical and Computer Engineering, Boston University, Boston, MA 02215, USA (email: bobak@bu.edu). M. Gastpar is with the Department of Electrical Engineering and Computer
Sciences, University of California, Berkeley, CA 94720 USA, and with the School of Computer and Communication Sciences, Ecole Polytechnique F\'ed\'erale (EPFL), 1015 Lausanne, Switzerland (e-mail: gastpar@eecs.berkeley.edu).}}
\markboth{IEEE Trans Info Theory, to appear}{~}

\maketitle

\begin{abstract}
Interference is usually viewed as an obstacle to communication in wireless networks. This paper proposes a new strategy, compute-and-forward, that exploits interference to obtain significantly higher rates between users in a network. The key idea is that relays should decode linear functions of transmitted messages according to their observed channel coefficients rather than ignoring the interference as noise. After decoding these linear equations, the relays simply send them towards the destinations, which given enough equations, can recover their desired messages. The underlying codes are based on nested lattices whose algebraic structure ensures that integer combinations of codewords can be decoded reliably. Encoders map messages from a finite field to a lattice and decoders recover equations of lattice points which are then mapped back to equations over the finite field. This scheme is applicable even if the transmitters lack channel state information. 
\end{abstract}
 
\begin{keywords}
Relaying, cooperative communication, structured codes, nested lattice codes, reliable computation, AWGN networks, interference.
\end{keywords}

\section{Introduction}

In a wireless network, a transmission from a single node is heard not only by the intended receiver, but also by all other nearby nodes; by analogy, any receiver  not only captures the signal from its designated transmitter, but from all other nearby transmitters. The resulting interference is usually viewed as highly undesirable and clever algorithms and protocols have been devised to avoid interference between transmitters. Collectively, these strategies transform the physical layer into a set of \textit{reliable bit pipes}, i.e. each link can accommodate a certain number of bits per time unit.  These bit pipes can then  be used seamlessly by higher layers in the protocol stack. 

Since wireless terminals must compete for the same fixed chunk of spectrum, interference avoidance results in diminishing rates as the network size increases. Recent work on cooperative communication has shown that this penalty can be overcome by adopting new strategies at the physical layer. The key idea is that users should help relay each other's messages by exploiting the broadcast and multiple-access properties of the wireless medium; properties that are usually viewed as a hindrance and are not captured by a bit pipe interface. To date, most proposed cooperative schemes have relied on one of the following three core relaying strategies:

\begin{itemize}
\item \textit{Decode-and-Forward:} The relay decodes at least some part of the transmitted messages. The recovered bits are then re-encoded for collaborative transmission to the next relay. Although this strategy offers significant advantages, the relay is ultimately interference-limited as the number of transmitted messages increases \cite{ce79, ltw04, kgg05,ehm07}.

\item \textit{Compress-and-Forward:} The signal observed at the relay is vector quantized and this information is passed towards the destination. If the destination receives information from multiple relays, it can treat the network as a multiple-input multiple-output (MIMO) channel. Unfortunately, since no decoding is performed at intermediate nodes, noise builds up as messages traverse the network \cite{ce79,kgg05,kim08,ary09,ssps09,lkec11}. 

\item \textit{Amplify-and-Forward:} The relay simply acts as a repeater and transmits a scaled version of its observation. Like compress-and-forward, this strategy converts the network into a large MIMO channel with the added possibility of a beamforming gain. However, noise also builds up with each retransmission. \cite{sg00,ltw04, gv05, bzg07, ehm07,mgm10}.
\end{itemize}

 In this paper, we propose a new strategy, \textit{compute-and-forward}, that enables relays to decode linear equations of the transmitted messages using the noisy linear combinations provided by the channel. A destination, given sufficiently many linear combinations, can solve for its desired messages. Our strategy relies on codes with a linear structure, specifically nested lattice codes. The linearity of the codebook ensures that integer combinations of codewords are themselves codewords. A relay is free to determine which linear equation to recover, but those closer to the channel's fading coefficients are available at higher rates. 

This strategy simultaneously affords protection against noise and the opportunity to exploit interference for cooperative gains. One could interpret compress-and-forward and amplify-and-forward as converting a network into a set of noisy linear equations; in this sense, compute-and-forward converts it into a set of \textit{reliable linear equations}. These equations can in turn be used for a digital implementation of cooperative schemes that could fit into a (slightly revised) network protocol stack. Classical relaying strategies seem to require a cross-layer design that dispenses with bit pipes and gives higher layers in the network stack direct access to the wireless medium. However, this would negate many of the advantages of a modular design \cite{kk05}. Compute-and-forward provides a natural solution to this problem by permitting a slight revision of the interface from bits to \textit{equations of bits}. 

We will develop a general framework for compute-and-forward that can  be used in any relay network with linear channels and additive white Gaussian noise (AWGN). Transmitters send out messages taking values in a prime-sized finite field and relays recover linear equations of the messages over the same field, making this an ideal physical layer interface for network coding. We will compare compute-and-forward to classical relaying strategies in a case study based on distributed MIMO. Classical relaying strategies perform well in either low or high signal-to-noise ratio (SNR) regimes. As we will see, compute-and-forward offers advantages in moderate SNR regimes where both interference and noise are significant factors.

\subsection{Related Work}
 
There is a large body of work on lattice codes and their applications in communications. We cannot do justice to all of this work here and point the interested reader to an excellent survey by Zamir \cite{zamir09}. The basic insight is that, for many AWGN networks of interest, nested lattice codes can approach the performance of standard random coding arguments. One key result by Erez and Zamir showed that nested lattice codes (combined with lattice decoding) can achieve the capacity of the point-to-point AWGN channel \cite{ez04}. More generally, Zamir, Shamai, and Erez demonstrated how to use nested lattice codes for many classical AWGN multi-terminal problems in \cite{zse02}. Subsequent work by El Gamal, Caire, and Damen showed that nested lattice codes achieve the diversity-multiplexing tradeoff of MIMO channels \cite{ecd04}. Note that, in general, structured codes are not sufficient to prove capacity results. For instance, group codes cannot approach the capacity of asymmetric discrete memoryless channels \cite{ahlswedegroup71}. 

It is tempting to assume that requiring codes to have a certain algebraic structure diminishes their usefulness for proving capacity theorems. However, it has become clear that for certain network communication scenarios, structured codes can actually outperform standard random coding arguments \cite{ng08ETT}. The first example of such behavior was found by {K\"orner} and Marton in \cite{km79}. They considered a decoder that wants to reconstruct the parity of two dependent binary sources observed by separate encoders. They found the rate region by using the same linear code at each encoder. More recently, we showed that structured codes offer large gains for reliable computation over multiple-access channels \cite{ng07IT}. Philosof \etal demonstrated that structured codes enable distributed dirty paper coding for multiple-access channels \cite{pz09,pzek09}.

The celebrated paper of Ahlswede \etal on network coding showed that for wired networks, relays must send out functions of received data, rather than just routing it \cite{acly00}. Subsequent work has shown that linear codes \cite{lyc03, km03} and linear codes with random coefficients \cite{hkmesk06} are sufficient for multicasting. There has recently been a great deal of interest in exploiting the physical layer of the wireless medium for network coding. To the best of our knowledge, the idea of using wireless interference for network coding was independently and concurrently proposed by several groups. Zhang, Liew, and Lam developed modulation strategies for bi-directional communication and coined the phrase "physical layer network coding" \cite{zll06}. Popovski and Yomo suggested the use of amplify-and-forward for the two-way relay channel \cite{py06VTC}. For this network, Rankov and Wittneben suggested both amplify-and-forward and compress-and-forward \cite{rw06}. We suggested the use of structured codes for the closely related wireless butterfly network \cite{ng06}. Subsequently, we developed lattice strategies for Gaussian multiple-access networks (without fading) \cite{ng07allerton} and Narayanan, Wilson, and Sprintson developed a nested lattice strategy for the two-way relay channel \cite{nws07,wnps10}.  Nam, Chung, and Lee generalized this strategy to include asymmetric power constraints \cite{ncl08}, found the capacity to within half a bit \cite{ncl10}, and extended their scheme to Gaussian multiple-access networks \cite{ncl09}. Owing to space constraints, we point to surveys by Liew, Zhang, and Lu \cite{lzl11} and ourselves \cite{ng11PIEEE} for a broader view of the rest of the physical layer network coding literature.

Work on interference alignment by Maddah-Ali, Motahari, and Khandani \cite{mmk08} and Cadambe and Jafar \cite{cj08} has shown that large gains are possible for interference channels at high signal-to-noise ratio (SNR). The key is to have users transmit along subspaces chosen such that all interference stacks up in the same dimensions at the receivers. Lattice codes can be used to realize these gains at finite SNR. Bresler, Parekh, and Tse used lattice codes to approximate the capacity of the many-to-one and one-to-many interference channels to within a constant number of bits \cite{bpt10}. This scheme was employed for bursty interference channels in \cite{kpv09}. For symmetric interference channels, Sridharan \etal developed a layered lattice strategy in \cite{sjvjs08}. Structured codes are also useful for ergodic alignment over fast fading interference channels \cite{ngjv09ISIT} and multi-hop networks \cite{jc09} as well as decentralized processing in cellular networks \cite{sps08,nsgs09}.

Distributed source coding can also benefit from the use of structured codes.  Krithivasan and Pradhan have employed nested lattice codes for the distributed compression of linear functions of jointly Gaussian sources \cite{kp09} as well as nested group codes for discrete memoryless sources \cite{kp10}. Wagner improved the performance of this lattice scheme in the low rate regime via binning and developed novel outer bounds \cite{wagner11}.  

Large gains are possible in multi-user source-channel coding \cite{sng07,ng08IZS,sv09}. For Gaussian settings, the modulo-lattice modulation scheme of Kochman and Zamir is particularly useful \cite{kz08}. Finally, recent work by He and Yener has shown that lattices are useful for physical layer secrecy \cite{hy09}. See also \cite{av09}. 

Finally, we mention several recent papers that have developed practical codes for compute-and-forward \cite{fsk10,hn10,oe10}.

\subsection{Summary of Paper Results}
 
 Our basic strategy is to take messages from a finite field, map them onto lattice points, and transmit these across the channel. Each relay observes a linear combination of these lattice points and attempts to decode an integer combination of them. This equation of lattice points is finally mapped back to a linear equation over a finite field. Our main theorems are summarized below:
\begin{itemize}
\item Theorems \ref{t:basiccompreal} and \ref{t:optcompreal} give our achievable rates for sending equations over a finite field from transmitters to relays over real-valued channel models. The strategy relies on a nested lattice coding strategy which is developed in Theorem \ref{t:latcompreal}. The corresponding results for complex-valued channel models are stated in Theorems \ref{t:basiccompcomplex}, \ref{t:optcompcomplex}, and \ref{t:latcompcomplex}. 
\item Theorems \ref{t:recoverallreal} through \ref{t:rankcomplex} give sufficient conditions on the equation coefficients so that a destination can recover one or more of the original messages.
\item Theorems \ref{t:twocomp} and \ref{t:twolevels} generalize the compute-and-forward scheme to include successive cancellation and superposition coding.
\item Theorem \ref{t:upper} is a simple upper bound on the rates for sending equations.
 \end{itemize}

We extend our framework to the slow fading setting in Section \ref{s:outage}. We then compare the performance of compute-and-forward to that of classical relaying strategies via a distributed MIMO case study in Section \ref{s:distmimo}.

\section{Problem Statement}

Our relaying strategy is applicable to any configuration of sources, relays, and destinations that are linked through linear\footnote{Erez and Zamir have recently investigated applying this framework to non-linear scenarios \cite{ez08}.} channels with additive white Gaussian\footnote{In fact, our strategy is applicable to a much broader class of additive noise statistics since we employ a minimum-distance decoder.} noise (AWGN). We will refer to such configurations as AWGN networks. To simplify the description of the scheme, we will first focus on how to deliver equations to a single set of relays. We will then show how a destination, given sufficiently many equations, can recover the intended messages. These two components are sufficient to completely describe an achievable rate region for any AWGN network. We will begin with definitions for real-valued channel models and then modify these to fit complex-valued channel models. 

%

\subsection{Real-Valued Channels}

Let $\Rbb$ denote the reals and $\Fbb_p$ denote the finite field of size $p$ where $p$ is always assumed to be prime. Let $+$ denote addition over the reals and $\oplus$ addition over the finite field. Furthermore, let $\sum$ denote summation over the reals and $\bigoplus$ denote summation over the finite field. It will be useful to map between the prime-sized finite field $\Fbb_p$ and the corresponding subset of the integers, $\{0,1,2,\ldots, p-1\}$. We will use the function $g(\cdot)$ to denote this map. This is essentially an identity map except for the change of alphabet. If $g$ or its inverse $g^{-1}$ are applied to a vector we assume they operate element-wise. We assume that the $\log$ operation is with respect to base $2$.

We will use boldface lowercase letters to denote column vectors and boldface uppercase letters to denote matrices. For example, $\hb \in \Rbb^L$ and $\Hb \in \Rbb^{M \times L}$. Let $\| \mathbf{h} \| \triangleq \sqrt{ \sum_{i=1}^L{|h[i]|^2}}$ denote the $\ell^2$-norm of $\hb$. Also, let $\hb^T$ denote the transpose of $\hb$. Finally, let $\mathbf{0}$ denote the zero vector,  $\delta_\e$ denote the unit vector with $1$ in the $\ellth$ entry and 0 elsewhere, and $\mathbf{I}^{M \times M}$ denote the identity matrix of size $M$. 


\begin{definition}[Messages] Each transmitter (indexed by $\ell = 1, 2, \ldots, L$) has a \mbox{length-$k_\ell$} \textit{message vector} that is drawn independently and uniformly over a prime-size finite field, $\mathbf{w}_\ell \in \mathbb{F}_p^{k_\ell}$. Without loss of generality, we assume that the transmitters are indexed by increasing message length. Since we are interested in functions of these message vectors, we zero-pad them to a common length $k \triangleq \max_\ell k _\ell$. 
\end{definition}
\begin{figure}[ht]
\begin{center}
\psset{unit=0.68mm}
\begin{pspicture}(-6,-20)(115,46)

\rput(0,-7){
\rput(1,32){$\wb_1$} \psline{->}(5,32)(10,32) \psframe(10,27)(21,37)
\rput(16,32){$\mathcal{E}_1$} \rput(27,35.5){$\xb_1$}
\psline{->}(21,32)(34,32)
}

\rput(0,-3){
\rput(1,12){$\wb_2$} \psline{->}(5,12)(10,12) \psframe(10,7)(21,17)
\rput(16,12){$\mathcal{E}_2$} \rput(27,15.5){$\xb_2$}
\psline{->}(21,12)(34,12)
}
\rput(16,0){$\vdots$}

\rput(0,5){
\rput(1,-18){$\wb_L$} \psline{->}(5,-18)(10,-18) \psframe(10,-23)(21,-13)
\rput(16,-18){$\mathcal{E}_L$} \rput(27,-14.5){$\xb_L$}
\psline{->}(21,-18)(34,-18)
}


\psframe(34,-23)(66,37)
\rput(50,7){\Large{$\mathbf{H}$}}

\rput(25,0){
\psline(41,32)(46.5,32)
\pscircle(49,32){2.5} \psline{-}(47.75,32)(50.25,32)
\psline{-}(49,30.75)(49,33.25) \psline{->}(49,39.5)(49,34.5) \rput(49,42.5){$\zb_1$}
\psline{->}(51.5,32)(65,32) \rput(57,35.5){$\yb_1$}

\psline(41,12)(46.5,12)
\pscircle(49,12){2.5} \psline{-}(47.75,12)(50.25,12)
\psline{-}(49,10.75)(49,13.25) \psline{->}(49,19.5)(49,14.5) \rput(49,22.5){$\zb_2$}
\psline{->}(51.5,12)(65,12) \rput(57,15.5){$\yb_2$}

\psline(41,-18)(46.5,-18)
\pscircle(49,-18){2.5} \psline{-}(47.75,-18)(50.25,-18)
\psline{-}(49,-19.25)(49,-16.75) \psline{->}(49,-10.5)(49,-15.5) \rput(49,-7.5){$\zb_M$}
\psline{->}(51.5,-18)(65,-18) \rput(57,-14.5){$\yb_M$}

\psframe(65,27)(77,37) \rput(71.5,32){$\mathcal{D}_1$}
\psline{->}(77,32)(82,32)
\rput(86,32){$\ubh_1$}

\psframe(65,7)(77,17) \rput(71.5,12){$\mathcal{D}_2$}
\psline{->}(77,12)(82,12)
\rput(86,12){$\ubh_2$}

\rput(71.5,-1){$\vdots$}

\psframe(65,-23)(77,-13) \rput(71.5,-18){$\mathcal{D}_M$}
\psline{->}(77,-18)(82,-18)
\rput(87,-18){$\ubh_M$}
}

\end{pspicture}
\end{center}
\caption{$L$ transmitters reliably communicate linear functions $\ub_m = \bigoplus_{\ell = 1}^L q_{m\ell} \wb_\ell$ to $M$ relays over a real-valued AWGN network.} \label{f:probstate}
\end{figure}
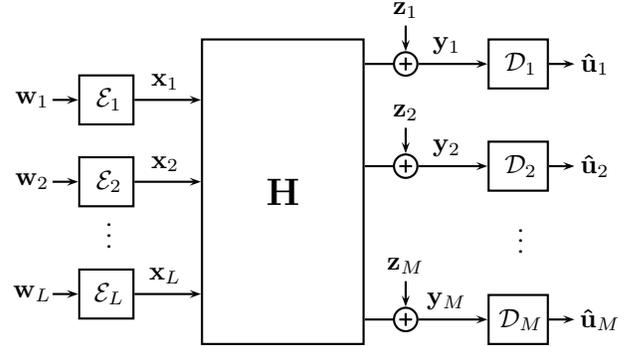

\begin{definition}[Encoders] Each transmitter is equipped with an \textit{encoder}, $\mathcal{E}_\ell: \mathbb{F}_p^k \rightarrow \Rbb^n$, that maps length-$k$ messages over the finite field to length-$n$ real-valued codewords, $\mathbf{x}_\ell = \mathcal{E}(\mathbf{w}_\ell)$. Each codeword is subject to the usual \textit{power constraint}, 
\begin{align}
 \| \mathbf{x}_\ell \|^2 \leq nP \ . \label{e:powerconstraint}
\end{align}
\end{definition}

\begin{remark}
Note that asymmetric power constraints can be incorporated by scaling the channel coefficients appropriately. 
\end{remark}

\begin{definition}[Message Rate] The \textit{message rate} $R_\ell$ of each transmitter is the length of its message (measured in bits) normalized by the number of channel uses,
\begin{align}
R_\ell = \frac{k_\ell}{n} \log{p}  \ . 
\end{align} Note that with our choice of indexing, the rates are in decreasing order, $R_1 \ge R_2 \ge \cdots \ge R_L$.
\end{definition}
\vspace{0.1in}
\begin{definition}[Channel Model] Each relay (indexed by $m = 1,2,\ldots, M$) observes a noisy linear combination of the transmitted signals through the \textit{channel},
\begin{align}
\mathbf{y}_m = \sum_{\ell=1}^L{h_{m\ell} \mathbf{x}_\ell} + \mathbf{z}_m  \ , \label{e:channelmodel}
\end{align}where $h_{m\ell} \in \Rbb$ are the channel coefficients and $\mathbf{z}$ is i.i.d. Gaussian noise, $\mathbf{z} \sim \mathcal{N}(\mathbf{0},\mathbf{I}^{n \times n})$. Let $\mathbf{h}_m = [h_{m1} \cdots h_{mL}]^T$ denote the vector of channel coefficients to relay $m$ and let $\mathbf{H} = \{ h_{m \ell} \}$ denote the entire channel matrix. Note that by this convention the $\mth$ row of $\Hb$ is $\hb_m^T$.
\end{definition}

\begin{remark} For our initial analysis, we will assume that the channel coefficients are fixed for all time. However, these results can easily be extended to the slow fading case under an outage formulation which we develop in Section \ref{s:outage}.
\end{remark}

\begin{remark} Our coding scheme only requires that each relay knows the channel coefficients from each transmitter to itself. Specifically, relay $m$ only needs to know $\hb_m$. Each transmitter only needs to know the desired message rate, not the realization of the channel.
\end{remark}

\begin{definition}[Desired Equations] \label{d:equations}
The goal of each relay is to reliably recover a \textit{linear combination} of the messages
\begin{align}
\ub_m &= \bigoplus_{\ell = 1}^L{q_{m \ell } \wb_\ell} \ . \label{e:msgeqn}
\end{align} where $q_{m \ell}$ are coefficients taking values in $\Fbb_p$. Each relay is equipped with a \textit{decoder}, $\mathcal{D}_m: \Rbb^n \rightarrow \mathbb{F}_p^k$, that maps the observed channel output $\mathbf{y}_m$ to an estimate $\mathbf{\hat{u}}_m = \mathcal{D}_m (\mathbf{y}_m)$ of the equation $\mathbf{u}_m$. 
\end{definition}

Although our desired equations are evaluated over the finite field $\Fbb_p$, the channel operates over the reals $\Rbb$. Our coding scheme will allow us to efficiently exploit the channel for reliable computation if the desired equation coefficients are close to the channel coefficients in an appropriate sense. The definition below provides an embedding from the finite field to the reals that will be useful in quantifying this closeness.

\begin{definition}[Coefficient Vector] The \textit{equation with coefficient vector} $\ab_m= \left[a_{m1} ~ a_{m2} ~ \cdots ~ a_{mL}\right]^T \in \Zbb^L$ is the linear combination of the transmitted messages $\mathbf{u}_m$ with coefficients given by 
\begin{align}
q_{m\ell} &= g^{-1}\big(\left[{a}_{m\ell} \right] \modp \big) \ . 
\end{align} Recall that $g^{-1}$ maps elements of $\{0,1,2,\ldots,p-1\}$ to the corresponding element in $\Fbb_p$.
\end{definition}

\begin{definition}[Probability of Error] \label{d:perror} We say that the equations with coefficient vectors $\ab_1, \ab_2, \ldots, \ab_M \in \Zbb^L$ are decoded with \textit{average probability of error} $\epsilon$ if
\begin{align}
\Pr\left( \bigcup_{m=1}^M \{ \ubh_m \neq  \ub_m \} \right)  < \epsilon \ .
\end{align}
\end{definition}

We would like to design a coding scheme that allows the transmitters to be oblivious of the channel coefficients and enables the relays to use their channel state information to select which equation to decode. Intuitively, equations whose coefficient vectors closely approximate the channel coefficients will be available at the highest rates. 

\begin{definition}[Computation Rate] \label{d:comprate} We say that the \textit{computation rate region} $\mathcal{R}(\hb_m, \ab_m)$ is achievable if for any $\epsilon > 0$ and $n$ large enough, there exist encoders and decoders, $\Em_1, \ldots, \Em_L, \Dm_1, \ldots, \Dm_M$, such that all relays can recover their desired equations with average probability of error $\epsilon$ so long as the underlying message rates $R_1, \ldots, R_L$ satisfy
\begin{align}
{R_\ell} <\min_{m:a_{m\ell}\neq 0} \mathcal{R}(\hb_m, \ab_m) \ . \label{e:compratedef}
\end{align}
\end{definition} 

In other words, a relay can decode an equation if the involved messages (i.e. those with non-zero coefficients) have message rates less than the computation rate between the channel and equation coefficient vectors. In fact, a relay will often be able to decode more than one equation and will have to decide which to forward into the network based on the requirements of the destinations. 

%
%

Although our scheme can be employed in any AWGN network, we will omit formal definitions for such networks and simply give recoverability conditions for equations of messages collected by a destination. This may occur via a single layer of relays as described above or through multiple layers.

\begin{definition}[Recovery] \label{d:recovery}
We say that message $\wb_\ell \in \Fbb_p^{k_\ell}$ can be \textit{recovered} at rate $R_\ell$ from the equations $\ub_m$ with coefficient vectors $\ab_1, \ldots, \ab_M \in \Zbb^L$ if for any $\epsilon > 0$ and $n$ large enough, there exists a decoder \mbox{$\Dm:\{ \Fbb_p^k \}^M \rightarrow  \Fbb_p^{k_\ell}$} such that
\begin{align}
&\wbh_\ell = \Dm\left(\ub_1, \ldots, \ub_M\right) \\
&\Pr\left(\wbh_\ell \neq \wb_\ell\right) < \epsilon \ . 
\end{align} \end{definition}

\subsection{Complex-Valued Channels} \label{s:complexchannel}
Let $\Cbb$ denote the complex field and $\mathbf{h}^*$ the Hermitian (or conjugate) transpose of a complex vector $\mathbf{h} \in \Cbb^L$. We also define $ j = \sqrt{-1}$. We are primarily interested in narrowband wireless channel models so we will specify our encoding and decoding schemes for complex baseband. Specifically, each transmitter sends a length-$n$ complex vector $\mathbf{x}_\ell \in \mathbb{C}^n$, which must obey the power constraint $\| \mathbf{x} \|_2 \leq nP$. Each relay observes a noisy linear superposition of the codewords, $\mathbf{y}_m = \sum_\ell h_{m\ell} \mathbf{x}_\ell + \mathbf{z}_m$, where $h_{m\ell} \in \Cbb$ are complex-valued channel coefficients and $\mathbf{z}_m$ is i.i.d. circularly symmetric complex Gaussian noise, $\mathbf{z}_m \sim \mathcal{CN}(\mathbf{0}, \mathbf{I}^{M\times M})$.

One simple possibility is to directly employ the framework developed above using the real-valued representation for complex vectors,
\begin{align*}
 \Re(\yb_m) &=  \sum_{\ell=1}^L{\left(\Re(h_{m\ell})\Re(\xb_\ell) - \Im(h_{m\ell})\Im(\xb_\ell) \right)} + \Re(\zb_m) \\
 \Im(\yb_m) &=  \sum_{\ell=1}^L{\left(\Im(h_{m\ell})\Re(\xb_\ell) + \Re(h_{m\ell})\Im(\xb_\ell) \right)} + \Im(\zb_m) \end{align*} From here, we can treat a complex-valued network with $L$ transmitters and $M$ relays as a real-valued network with $2L$ transmitters and $2M$ relays. However, there is a more elegant solution that takes advantage of the special structure of complex symbols. Below, we modify definitions to fit the complex case.

\begin{definition}[Complex Messages] Each transmitter has two \mbox{length-$k_\ell$} vectors that are drawn independently and uniformly over a prime-size finite field, $\mathbf{w}_\ell^R, \mathbf{w}_\ell^I \in \mathbb{F}_p^{k_\ell}$. The superscript denotes whether the vector is intended for the real part or the imaginary part of the channel. Together these vectors are the \textit{message} of transmitter $\ell$, $\wb_\e = (\wb_\e^R, \wb_\e^I)$. As before, we assume that the transmitters are indexed by increasing message length and zero-pad them to a common length $k \triangleq \max_\ell k _\ell$ prior to encoding. The \textit{message rate} of each transmitter is double the prior definition $R_\ell = (2 k_\ell/n) \log{p}$. 
\end{definition}

\begin{definition}[Desired Complex Equations] \label{d:equations}
The goal of each relay is to reliably recover a \textit{linear combination} of the messages, \begin{align}
\ub_m^R &= \bigoplus_{\ell = 1}^L\Big({q_{m \ell }^R \wb_\ell^R \oplus (-q_{m \ell }^I) \wb_\ell^I}\Big) \label{e:msgeqnR}\\
\ub_m^I &= \bigoplus_{\ell = 1}^L\Big({q_{m \ell }^I \wb_\ell^R \oplus q_{m \ell }^R \wb_\ell^I}\Big) \ ,  \label{e:msgeqnI}
\end{align} where the $q_{m\ell}$ are coefficients taking values in $\Fbb_p$ and $(-q_{m \ell})$ denotes the additive inverse of $q_{m \ell}$. The \textit{equation with coefficient vector} $\ab_m= \left[a_{m1} ~ a_{m2} ~ \cdots ~ a_{mL}\right]^T \in \{\Zbb + j \Zbb\}^L$ are the linear combinations with coefficients given by 
\begin{align}
q_{m\ell}^R &= g^{-1}\big(\left[\Re({a}_{m\ell}) \right]\hspace{-0.1in}\mod p \big) \\
q_{m\ell}^I &=  g^{-1}\big(\left[\Im({a}_{m\ell}) \right]\hspace{-0.1in}\mod p \big).
\end{align}
\end{definition}
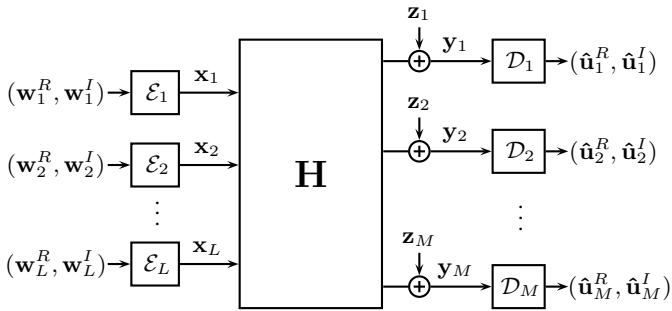
\begin{figure}[ht]
\begin{center}
\psset{unit=0.6mm}
\begin{pspicture}(-15,-20)(127,46)

\small

\rput(0,-7){
\rput(-6.5,32){$(\wb_1^R,\wb_1^I)$} \psline{->}(5,32)(10,32) \psframe(10,27)(21,37)
\rput(16,32){$\mathcal{E}_1$} \rput(27,35.5){$\xb_1$}
\psline{->}(21,32)(34,32)
}

\rput(0,-3){
\rput(-6.5,12){$(\wb_2^R,\wb_2^I)$} \psline{->}(5,12)(10,12) \psframe(10,7)(21,17)
\rput(16,12){$\mathcal{E}_2$} \rput(27,15.5){$\xb_2$}
\psline{->}(21,12)(34,12)
}
\rput(16,0){$\vdots$}

\rput(0,5){
\rput(-6.5,-18){$(\wb_L^R,\wb_L^I)$} \psline{->}(5,-18)(10,-18) \psframe(10,-23)(21,-13)
\rput(16,-18){$\mathcal{E}_L$} \rput(27,-14.5){$\xb_L$}
\psline{->}(21,-18)(34,-18)
}


\psframe(34,-23)(66,37)
\rput(50,7){\Large{$\mathbf{H}$}}

\rput(25,0){
\psline(41,32)(46.5,32)
\pscircle(49,32){2.5} \psline{-}(47.75,32)(50.25,32)
\psline{-}(49,30.75)(49,33.25) \psline{->}(49,39.5)(49,34.5) \rput(49,42.5){$\zb_1$}
\psline{->}(51.5,32)(65,32) \rput(57,35.5){$\yb_1$}

\psline(41,12)(46.5,12)
\pscircle(49,12){2.5} \psline{-}(47.75,12)(50.25,12)
\psline{-}(49,10.75)(49,13.25) \psline{->}(49,19.5)(49,14.5) \rput(49,22.5){$\zb_2$}
\psline{->}(51.5,12)(65,12) \rput(57,15.5){$\yb_2$}

\psline(41,-18)(46.5,-18)
\pscircle(49,-18){2.5} \psline{-}(47.75,-18)(50.25,-18)
\psline{-}(49,-19.25)(49,-16.75) \psline{->}(49,-10.5)(49,-15.5) \rput(49,-7.5){$\zb_M$}
\psline{->}(51.5,-18)(65,-18) \rput(57,-14.5){$\yb_M$}

\psframe(65,27)(77,37) \rput(71.5,32){$\mathcal{D}_1$}
\psline{->}(77,32)(82,32)
\rput(92,32){$(\ubh_1^R,\ubh_1^I)$}

\psframe(65,7)(77,17) \rput(71.5,12){$\mathcal{D}_2$}
\psline{->}(77,12)(82,12)
\rput(92,12){$(\ubh_2^R,\ubh_2^I)$}

\rput(71.5,-1){$\vdots$}

\psframe(65,-23)(77,-13) \rput(71.5,-18){$\mathcal{D}_M$}
\psline{->}(77,-18)(82,-18)
\rput(93.5,-18){$(\ubh_M^R,\ubh_M^I)$}
}

\end{pspicture}
\end{center}
\caption{$L$ transmitters reliably communicate linear functions $\ub_m^R = \bigoplus_{\ell = 1}^L\Big({q_{m \ell }^R \wb_\ell^R \oplus (-q_{m \ell }^I) \wb_\ell^I}\Big)$ and $\ub_m^I = \bigoplus_{\ell = 1}^L\Big({q_{m \ell }^I \wb_\ell^R \oplus q_{m \ell }^R \wb_\ell^I}\Big)$ to $M$ relays over a complex-valued AWGN network.} \label{f:probstate}
\end{figure}

Note that the coefficient choices for the real and imaginary part are coupled, which means that each relay only needs to decide on $2L$ coefficients instead of the $4L$ needed for a real-valued system with $2L$ transmitters. The definitions for the probability of error, the computation rate region, and recovery are identical to Definitions \ref{d:perror}, \ref{d:comprate}, and \ref{d:recovery} except with $\Cbb$ and $\ZC$ taking the place of $\Rbb$ and $\Zbb$, respectively.

\section{Main Results}

Our main result is that relays can often recover an equation of messages at a higher rate than any individual message (or subset of messages). The rates are highest when the equation coefficients closely approximate the channel coefficients. Below, we give a formal statement of this result for real-valued channels. Let $\log^+(x) \triangleq \max{(\log(x),0)}$. 

\begin{theorem}\label{t:basiccompreal}
For real-valued AWGN networks with channel coefficient vectors $\hb_m \in \Rbb^L$ and equation coefficient vectors $\ab_m \in \Zbb^L$, the following computation rate region is achievable:
\begin{align}
\mathcal{R}(\hb_m, \ab_m) =& \max_{\alpha_m \in \Rbb} \frac{1}{2} \log^+\left(\frac{P}{\alpha_m^2 + P \| \alpha_m \mathbf{h}_m - \mathbf{a}_m \|^2}\right)  \ . \nonumber 
\end{align}
\end{theorem}
\vspace{0.1in} A detailed proof is given in Section \ref{s:realcompproof}.

\begin{theorem} \label{t:optcompreal}
The computation rate given in Theorem \ref{t:basiccompreal} is uniquely maximized by choosing $\alpha_m$ to be the MMSE coefficient
\begin{align}
\ammse = \frac{P ~\mathbf{h}_m^T\mathbf{a}_m}{1 + P \|\mathbf{h}_m \|^2}
\end{align}  which results in a computation rate region of
\begin{align}
\mathcal{R}(\hb_m, \ab_m) = \frac{1}{2}\log^+{\left(\left(\| \mathbf{a}_m \|^2 - \frac{P ~ ( \mathbf{h}_m^T \mathbf{a}_m )^2}{1 + P\|\mathbf{h}_m\|^2}\right)\n\right)} \nonumber
\end{align}
\end{theorem} The proof is nearly identical to that of Theorem \ref{t:optcompcomplex}. 

The computation rate expression for the complex-valued case is simply twice the expression for the real-valued case.

%
%
%
%

\begin{theorem}\label{t:basiccompcomplex}
For complex-valued AWGN networks with channel coefficient vectors $\hb_m \in \Rbb^L$ and equation coefficient vectors $\ab_m \in \ZC^L$, the following computation rate region is achievable:
\begin{align}
\mathcal{R}(\hb_m, \ab_m) =& \max_{\alpha_m \in \Cbb} \log^+\left(\frac{P}{|\alpha_m|^2 + P \| \alpha_m \mathbf{h}_m - \mathbf{a}_m \|^2}\right) \ . \nonumber 
\end{align} 
\end{theorem}\vspace{0.1in} A detailed proof is given in Section \ref{s:complexcompproof}.

\begin{theorem} \label{t:optcompcomplex}
The computation rate given in Theorem \ref{t:basiccompcomplex} is uniquely maximized by choosing $\alpha_m$ to be the MMSE coefficient
\begin{align}
\ammse = \frac{P ~\mathbf{h}_m^*\mathbf{a}_m}{1 + P \|\mathbf{h}_m \|^2}
\end{align}  which results in a computation rate region of
\begin{align}
\mathcal{R}(\hb_m, \ab_m) = \log^+{\left(\left(\| \mathbf{a}_m \|^2 - \frac{P ~ | \mathbf{h}_m^* \mathbf{a}_m |^2}{1 + P\|\mathbf{h}_m\|^2}\right)\n\right)} 
\end{align}
\end{theorem} 

\begin{proof}
Let $f(\alpha_m)$ denote the denominator of the computation rate in Theorem \ref{t:basiccompcomplex}. Since it is quadratic in $\alpha_m$, it can be uniquely minimized by setting its first derivative to zero. \begin{align}
f(\alpha_m) &= \alpha_m^* \alpha_m + P (\alpha_m \hb_m - \ab_m)^*(\alpha_m \hb_m - \ab_m) \nonumber \\
\frac{df}{d \alpha_m}&= 2 \alpha_m + P(2 \alpha_m \hb_m^* \hb_m - 2\hb_m^* \ab_m) = 0\\
&\alpha_m (2 + 2 P \| \hb_m \| ^2 ) = 2 P~ \hb_m^* \ab_m 
\end{align} 
We solve this to get $\ammse$ and plug back into $f(\alpha_m)$. 
\begin{align}
f(\ammse) &=~ \frac{ P^2 | \mathbf{h}_m^* \mathbf{a}_m |^2}{\left(1 + P\|\mathbf{h}_m\|^2\right)^2} + \frac{P^3  \| \hb_m \|^2  | \mathbf{h}_m^* \mathbf{a}_m |^2}{\left(1 + P\|\mathbf{h}_m\|^2\right)^2}~ \nonumber \\
&\qquad - ~2 \frac{P^2  | \mathbf{h}_m^* \mathbf{a}_m |^2}{1 + P\|\mathbf{h}_m\|^2}+ P \|\ab_m\|^2 \\ 
  &=  -  \frac{P^2 | \mathbf{h}_m^* \mathbf{a}_m |^2}{1 + P\|\mathbf{h}_m\|^2}+ P \|\ab_m\|^2 
\end{align} Substituting this into $\log^+\left(\frac{P}{f(\ammse)}\right)$ yields the desired computation rate. 
\end{proof}

The main interpretation of Theorems \ref{t:basiccompreal} and \ref{t:basiccompcomplex} is that all relays can simultaneously decode equations with coefficient vectors $\ab_m$ so long as the involved messages' rates are within the computation rate region 
\begin{align}
R_\ell < \min_{a_{m\ell} \neq 0} \mathcal{R}(\hb_m, \ab_m) \ . 
\end{align}  In other words, exactly which equation to decode is left up to the relays. The scalar parameter $\alpha_m$ is used to move the channel coefficients closer to the desired integer coefficients. For instance, if $\alpha_m = 1$, then the effective signal-to-noise ratio is $$\snr = \frac{P}{1 + P\| \hb_m - \ab_m \|^2} \ ,$$ meaning that the non-integer part of the channel coefficients acts as additional noise.  More generally, the scaled channel output $\alpha_m \yb_m = \sum{\alpha_m h_{m\ell} \xb_\ell} + \alpha_m \zb_m$ can be equivalently written as a channel output $\tilde{\yb}_m = \sum{\tilde{h}_{m\ell} \xb_\ell} + \tilde{\zb}_m$ where $\tilde{h}_{m \ell} = \alpha_m h_{m \ell}$ and $\tilde{\zb}_m$ is i.i.d. according to $\mathcal{CN}(0,|\alpha_m|^2)$. In this case, the effective signal-to-noise ratio is $$\snr = \frac{P}{| \alpha_m|^2 + P \| \alpha_m \hb_m - \ab_m \|^2} \ .$$ Since there is a rate penalty both for noise and for non-integer channel coefficients, then $\alpha_m$ should be used to optimally balance between the two as in Theorems \ref{t:optcompreal} and \ref{t:optcompcomplex}. This is quite similar to the role of the MMSE scaling coefficient used by Erez and Zamir to achieve the capacity of the point-to-point AWGN channel in \cite{ez04}.
 
\begin{example}\label{e:intcomp} Let the channel matrix take values on the complex integers, $\Hb \in \{\Zbb +j\Zbb\}^{M \times L}$,  and assume that each relay wants a linear equation with a coefficient vector that corresponds exactly to the channel coefficients, $\ab_m = \hb_m$. Using Theorem \ref{t:optcompcomplex}, the relays can decode so long as
\begin{align}
R_\ell &< \min_{m: h_{m\ell} \neq 0 }\log^+{\left(\left(\|\hb_m\|^2 - \frac{P  \|\hb_m \|^4 }{1 +P  \| \hb_m \| ^2  } \right)\n \right)} \nonumber \\
&= \min_{m: h_{m\ell} \neq 0 } \log^+{\left(\frac{1 +P  \| \hb_m \| ^2}{\|\hb_m\|^2 + P  \|\hb_m \|^4 - P  \|\hb_m \|^4}\right)} \nonumber \\
&= \min_{m: h_{m\ell} \neq 0 }\log^+{\left(\frac{1}{\|\hb_m\|^2} + P\right)}
\end{align}
\end{example}
\begin{remark} One interesting special case of Example \ref{e:intcomp} is computing the modulo sum of codewords $\wb_1 \oplus \wb_2$ over a two-user Gaussian multiple-access channel $\mathbf{y} = \xb_1 + \xb_2 + \zb$. To date, the best known achievable computation rate for this scenario is $\log^+{\left(\frac{1}{2} + P\right)}$. Several papers (including our own) have studied this special case and it is an open problem as to whether the best known outer bound $\log{\left(1 + P\right)}$ is achievable \cite{ng07allerton, wnps10, ncl08}. Clearly, one can do better in the low SNR regime using standard multiple-access codes to recover all the messages then compute the sum to get $\frac{1}{2}\log{\left(1 + 2 P\right)}$.
\end{remark}

\begin{example} \label{ex:onemsg}
Assume there are $M$ transmitters and $M$ relays. Relay $m$ wants to recover the message from transmitter $m$. This corresponds to setting the desired coefficient vector to be a unit vector $\ab_m = \delta_m$. Substituting this choice into Theorem \ref{t:optcompcomplex}, we get that the messages can be decoded if their rates satisfy
\begin{align}
R_m &< \log^+{\left(\left(1 - \frac{P |h_{mm}|^2}{1 + P \| \hb_m \|^2}\right)\n \right)} \\
&= \log^+{\left(\left( \frac{1 + P \sum_{\ell \neq m} |h_{m\ell}|^2}{1 + P \| \hb_m \|^2}\right)\n\right)} \\
& = \log{\left(1 + \frac{P |h_{mm}|^2}{1 + P  \sum_{\ell \neq m} |h_{m\ell}|^2}\right)} \ .
\end{align} This is exactly the rate achievable with standard multiple-access techniques if the relays ignore all other messages as noise. In Section \ref{s:cancel}, we will use successive cancellation of lattice equations to show that if a relay wants all of the messages, any point in the Gaussian multiple-access rate region is achievable with compute-and-forward.
\end{example}

\begin{remark}
The setup in Example \ref{ex:onemsg} is exactly that of an $M$-user Gaussian interference channel. Higher rates are possible by incorporating techniques such as the superposition of public and private messages \cite{hk81,etw08} and interference alignment \cite{mmk08,cj08}. Note that these can be implemented in concert with compute-and-forward. For instance, in Section \ref{s:superposition}, we describe a superposition compute-and-forward strategy. \end{remark}

In general, the choice of the coefficient vector $\ab_m$ at each relay will depend both on the channel coefficients and the message demands at the destinations. Relays should make use of their available channel state information (CSI) to determine the most valuable equation to forward. One simple greedy approach is to choose coefficient vectors with the highest computation rate
\begin{align}
\ab_{m} = \argmax_{\mathbf{\tilde{a}}} \mathcal{R}(\hb_m, \mathbf{\tilde{a}}) \ . \label{e:compmax} 
\end{align} This is a compelling strategy for scenarios where only local CSI is available. It resembles random linear network coding \cite{hkmesk06} except here the randomness stems entirely from the channel coefficients. In the next lemma, we demonstrate that this maximization does not require a search over all integer vectors. 

\begin{lemma} \label{l:coeffbound}
For a given channel vector $\hb$, the computation rate $\mathcal{R}(\hb_m, \ab_m)$ from Theorems \ref{t:optcompreal} and \ref{t:optcompcomplex} are zero if the coefficient vector $\ab$ satisfies:
\begin{align}
\| \ab_m \|^2 \geq 1 + \| \hb_m \|^2 P. \label{e:coeffbound}
\end{align}
\end{lemma}
\begin{proof}
Note that $| \mathbf{h}_m^* \mathbf{a}_m |^2 \leq  \|\mathbf{h}_m \|^2 \|\mathbf{a}_m \|^2$ by the Cauchy-Schwarz inequality. Using this, we can upper bound the computation rate:
\begin{align}
&\log^+{\left(\left(\| \mathbf{a}_m \|^2 - \frac{P ~ | \mathbf{h}_m^* \mathbf{a}_m |^2}{1 + P\|\mathbf{h}_m\|^2}\right)\n\right)} \\
& = \log^+{\left(\frac{1 + P\|\mathbf{h}_m\|^2}{\| \mathbf{a}_m \|^2 + P \|\mathbf{h}_m \|^2\|\mathbf{a}_m \|^2 -P ~ | \mathbf{h}_m^* \mathbf{a}_m |^2}\right)} \nonumber \\
&\leq  \log^+{\left(\frac{1 + P\|\mathbf{h}_m\|^2}{\| \mathbf{a}_m \|^2 }\right)} \ . 
\end{align} The result follows immediately.
\end{proof}

\begin{figure}[h]
\centering
\includegraphics[width=3.75in]{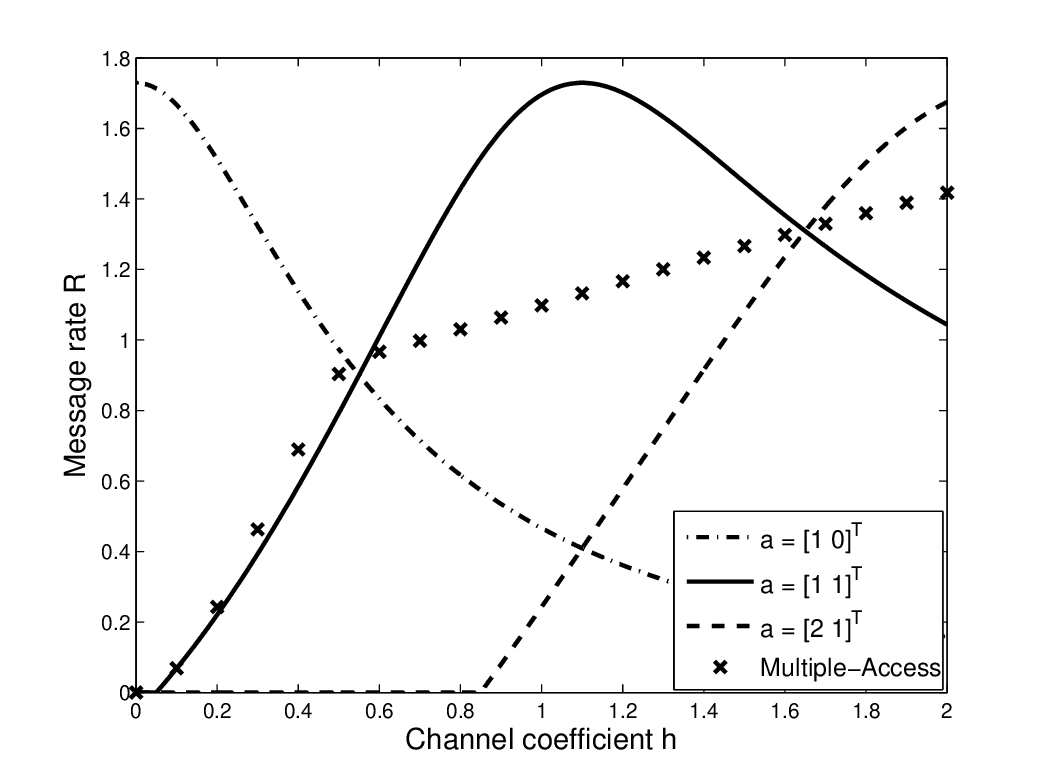}
\caption{Recovering equations with coefficient vectors $\ab = [1 ~0]^T, [1 ~1]^T, [2~ 1]^T$ over a multiple-access channel with channel vector $\mathbf{h} = [h ~1]$ where $h$ varies between $0$ and $2$. The message rates are symmetric $R_1 = R_2 = R$ and the power is $P = 10$dB. For comparison, we have also plotted the symmetric multiple-access capacity. }\label{f:rateprofile}
\end{figure}

In Figure \ref{f:rateprofile}, we have plotted how the computation rate from Theorem \ref{t:optcompreal} varies as the channel coefficients change for several possible coefficients vectors. In this example, the message rates are symmetric $R_1 = R_2 = R$ and the power is $10$dB. The channel vector $\hb = [h~1]^T$ is parametrized by $h$ which is varied between $0$ and $2$. The coefficient vectors are $\ab = [1 ~0]^T,~[1~1]^T,$ and $[2~1]^T$. Each of these vectors attain its maximum computation rate when the channel vector is an exact match.

\begin{remark}As the power increases, more coefficient vectors should be used to approximate the channel more finely. However, in the high SNR limit, it has recently been shown by Niesen and Whiting that the degrees-of-freedom (DoF) of our scheme becomes discontinuous \cite{nw11}. Specifically, at rational channel vectors, our scheme attains the maximum DoF but, at irrational vectors, the DoF is upper bounded by a constant as the number of users increases. Under the assumption that the transmitters know the channel realization, they can attain the maximum DoF (up to a set of measure zero) by coupling compute-and-forward with the interference alignment scheme of Motahari \textit{et al.} for fixed channels \cite{mgmk09}. 
\end{remark}

\begin{remark}\label{r:piecewise}
Note that each relay is free to decode more than one equation, so long as all the appropriate computation rates are satisfied. In some cases, it may be beneficial to recover a desired equation by first decoding equations of subsets of messages and then combining them. 
\end{remark}

The following example shows that it is useful to allow for a different rate at each transmitter. 

\begin{example}
Consider a complex-valued AWGN network with $L = 4$ transmitters and $M = 2$ relays. The channel vectors are $\hb_1 = [4 ~-4~~ 1~ -1]^T$ and $\hb_2 = [1~ ~1~ ~2 ~~2]^T$. The desired coefficient vectors are $\ab_1 = \hb_1$ and $\ab_2 = [ 0~~ 0~~1~~1]^T$. These equations can be reliably recovered so long as the message rates satisfy:
\begin{align} 
R_\ell <
\begin{cases} 
{\displaystyle \log^+{ \left(\frac{1}{34} + P \right)}} & \ell=1,2\\
{\displaystyle \log^+{\left(\frac{1}{2} + \frac{4 P}{1 + 2P}\right)}} & \ell = 3,4
\end{cases}
\end{align}
\end{example}

\section{Nested Lattice Codes} \label{s:lattices}

In order to allow relays to decode integer combinations of codewords, we need codebooks with a linear structure. Specifically, we will use nested lattice codes that have both good statistical and good algebraic properties. Erez and Zamir developed a class of nested lattice codes that can approach the capacity of point-to-point AWGN channels in \cite{ez04}. These codes operate under a modulo arithmetic that is well-suited for mapping operations over a finite field to the complex field.

First, we will provide some necessary definitions from \cite{ez04} on nested lattice codes. Note that all of these definitions are given over $\Rbb^n$. For complex-valued channels, our scheme will use the same lattice code over the real and imaginary parts of the channel input (albeit with different messages). 

\subsection{Lattice Definitions}
\begin{definition}[Lattice] An $n$-dimensional \textit{lattice}, $\Lambda$, is a set of
points in $\mathbb{R}^n$ such that if $\sb,\tb \in \Lambda$, then
$\sb + \tb \in \Lambda$, and if $\sb \in \Lambda$, then $-\sb \in
\Lambda$. A lattice can always be written in terms of a lattice generator
matrix $\mathbf{B} \in \mathbb{R}^{n \times n}$:
\begin{align}
\Lambda = \{\sb =\mathbf{B} \mathbf{c}: \mathbf{c} \in \mathbb{Z}^n\} \ .
\end{align} 
\end{definition}

\begin{definition}[Nested Lattices] A lattice $\Lambda$ is said to be \textit{nested} in a lattice $\Lambda_1$ if $\Lambda \subseteq \Lambda_1$. We will sometimes refer to $\Lambda$ as the coarse lattice and $\Lambda_1$ as the fine lattice. More generally, a sequence of lattices $\Lambda, \Lambda_1, \ldots, \Lambda_L$ is nested if $\Lambda \subseteq \Lambda_1 \subseteq \cdots \subseteq \Lambda_L$.
\end{definition}

\begin{definition}[Quantizer] A \textit{lattice quantizer} is a map, $Q_\Lambda: \mathbb{R}^n
\rightarrow \Lambda$, that sends a point, $\sb$, to the nearest
lattice point in Euclidean distance:  
\begin{align} 
Q_\Lambda(\mathbf{s}) = \arg \min_{\mathbf{\lambda} \in \Lambda}{||\sb -
\lambda||} \ .\end{align}
\end{definition}

\begin{definition}[Voronoi Region] The \textit{fundamental Voronoi region}, $\mathcal{V}$, of a lattice, is the set of all points in $\Rbb^n$ that are closest to the zero
vector: $\mathcal{V} = \{\sb : Q_\Lambda(\sb) = \mathbf{0}\}$. Let $\mathsf{Vol}(\mathcal{V})$ denote the volume of $\mathcal{V}$. 
\end{definition}

\begin{definition}[Modulus] Let $[\sb]~\modl$ denote the quantization error of $\sb \in \Rbb^n$ with respect to the lattice $\Lambda$,
\begin{align}
[\sb] \modl = \sb - Q_\Lambda(\sb) \ .
\end{align}For all $\sb, \tb \in \mathbb{R}^n$ and $\L \subseteq \L_1$, the~$\modl$ operation satisfies:
\begin{align}
[\sb + \tb] \modl&=\big[[\sb] \modl + \tb\big] \modl \label{e:modprop1}   \\
 \big[Q_{\L_1}(\sb)\big]\modl &= \big[Q_{\L_1}\big([\sb]\modl\big)\big]\modl \label{e:modprop2} \\
[a \sb] \modl &= [a [\sb] \modl] \modl  ~~~~~~\forall a \in \mathbb{Z} \label{e:modprop3}\\
 \beta [\sb] \modl &= [\beta\sb] \modbl  \, ~~~~~~~~~~~~~~\forall \beta \in \mathbb{R} \label{e:modprop4}
\end{align} 
\end{definition}

\begin{definition}[Nested Lattice Codes] A \textit{nested lattice code} $\mathcal{L}$ is the set of all points of a fine lattice $\Lambda_1$ that are within the fundamental Voronoi region $\Vm$ of a coarse lattice $\Lambda$,
\begin{align}
\mathcal{L} = \Lambda_1 \cap \Vm = \{ \tb: \tb = \lambda \modl, \lambda \in \Lambda_1 \} \ .
\end{align}
The rate of a nested lattice code is
\begin{align}
r = \onen \log{|\mathcal{L}|} = \onen \log{\frac{\V(\Vm)}{\V(\Vm_1)}} \ . 
\end{align}
\end{definition}
 
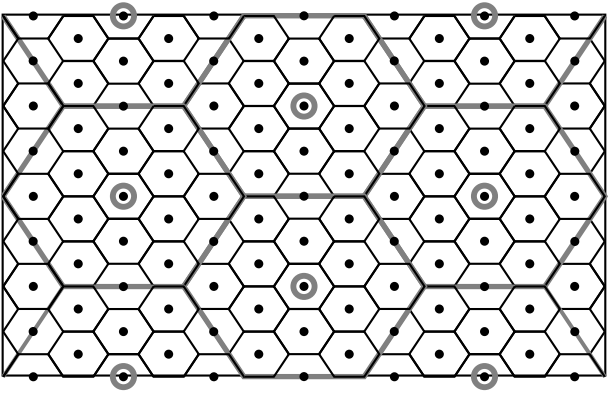
\begin{figure}[h]
\begin{center}
\psset{unit=0.6mm}
\begin{pspicture}(-28,-20)(108,60)

\psframe(-27,-20)(107,60.5)

\rput(120,40){

\rput(-20,0){
\psline(-3.333,5)(3.333,5)(6.667,0)(3.333,-5)(-3.333,-5)(-6.667,0)(-3.333,5)
\pscircle[fillstyle=solid,fillcolor=black](0,0){1}
}
\rput(-10,15){
\psline(-3.333,-5)(-6.667,0)(-3.333,5)
}
\rput(-10,5){
\psline(-3.333,-5)(-6.667,0)(-3.333,5)
}
\rput(-10,-5){
\psline(-3.333,-5)(-6.667,0)(-3.333,5)
}
\rput(-10,-15){
\psline(-3.333,-5)(-6.667,0)(-3.333,5)
}

}
\rput(120,0){

\rput(-20,0){
\psline(-3.333,5)(3.333,5)(6.667,0)(3.333,-5)(-3.333,-5)(-6.667,0)(-3.333,5)
\pscircle[fillstyle=solid,fillcolor=black](0,0){1}
}
\rput(-10,15){
\psline(-3.333,-5)(-6.667,0)(-3.333,5)
}
\rput(-10,5){
\psline(-3.333,-5)(-6.667,0)(-3.333,5)
}
\rput(-10,-5){
\psline(-3.333,-5)(-6.667,0)(-3.333,5)
}
\rput(-10,-15){
\psline(-3.333,-5)(-6.667,0)(-3.333,5)
}

}

\rput(80,60){

\psline[linewidth=2pt,linecolor=black!50!white](26.667,0)(13.333,-20)(-13.333,-20)(-26.667,0)

\rput(-20,0){
\psline(6.667,0)(3.333,-5)(-3.333,-5)(-6.667,0)\pscircle[fillstyle=solid,fillcolor=black](0,0){1}
}
\rput(-10,-5){
\psline(-3.333,5)(3.333,5)(6.667,0)(3.333,-5)(-3.333,-5)(-6.667,0)(-3.333,5)
\pscircle[fillstyle=solid,fillcolor=black](0,0){1}
}
\rput(-10,-15){
\psline(-3.333,5)(3.333,5)(6.667,0)(3.333,-5)(-3.333,-5)(-6.667,0)(-3.333,5)
\pscircle[fillstyle=solid,fillcolor=black](0,0){1}
}
\rput(20,0){
\psline(6.667,0)(3.333,-5)(-3.333,-5)(-6.667,0)\pscircle[fillstyle=solid,fillcolor=black](0,0){1}
}
\rput(0,0){
\psline(6.667,0)(3.333,-5)(-3.333,-5)(-6.667,0)
\pscircle[fillstyle=solid,fillcolor=black](0,0){1}
}
\rput(0,-10){
\psline(-3.333,5)(3.333,5)(6.667,0)(3.333,-5)(-3.333,-5)(-6.667,0)(-3.333,5)
\pscircle[fillstyle=solid,fillcolor=black](0,0){1}
}
\rput(10,-5){
\psline(-3.333,5)(3.333,5)(6.667,0)(3.333,-5)(-3.333,-5)(-6.667,0)(-3.333,5)
\pscircle[fillstyle=solid,fillcolor=black](0,0){1}
}
\rput(10,-15){
\psline(-3.333,5)(3.333,5)(6.667,0)(3.333,-5)(-3.333,-5)(-6.667,0)(-3.333,5)
\pscircle[fillstyle=solid,fillcolor=black](0,0){1}
}

\pscircle[fillstyle=solid,fillcolor=black](20,-10){1}

}
\rput(80,20){

\psline[linewidth=2pt,linecolor=black!50!white](-26.667,0)(-13.333,20)(13.333,20)(26.667,0)(13.333,-20)(-13.333,-20)(-26.667,0)

\rput(-20,0){
\psline(-3.333,5)(3.333,5)(6.667,0)(3.333,-5)(-3.333,-5)(-6.667,0)(-3.333,5)
\pscircle[fillstyle=solid,fillcolor=black](0,0){1}
}
\rput(-10,15){
\psline(-3.333,5)(3.333,5)(6.667,0)(3.333,-5)(-3.333,-5)(-6.667,0)(-3.333,5)
\pscircle[fillstyle=solid,fillcolor=black](0,0){1}
}
\rput(-10,5){
\psline(-3.333,5)(3.333,5)(6.667,0)(3.333,-5)(-3.333,-5)(-6.667,0)(-3.333,5)
\pscircle[fillstyle=solid,fillcolor=black](0,0){1}
}
\rput(-10,-5){
\psline(-3.333,5)(3.333,5)(6.667,0)(3.333,-5)(-3.333,-5)(-6.667,0)(-3.333,5)
\pscircle[fillstyle=solid,fillcolor=black](0,0){1}
}
\rput(-10,-15){
\psline(-3.333,5)(3.333,5)(6.667,0)(3.333,-5)(-3.333,-5)(-6.667,0)(-3.333,5)
\pscircle[fillstyle=solid,fillcolor=black](0,0){1}
}
\rput(20,0){
\psline(-3.333,5)(3.333,5)(6.667,0)(3.333,-5)(-3.333,-5)(-6.667,0)(-3.333,5)
\pscircle[fillstyle=solid,fillcolor=black](0,0){1}
}
\rput(0,10){
\psline(-3.333,5)(3.333,5)(6.667,0)(3.333,-5)(-3.333,-5)(-6.667,0)(-3.333,5)
\pscircle[fillstyle=solid,fillcolor=black](0,0){1}
}
\rput(0,0){
\psline(-3.333,5)(3.333,5)(6.667,0)(3.333,-5)(-3.333,-5)(-6.667,0)(-3.333,5)
\pscircle[fillstyle=solid,fillcolor=black](0,0){1}
}
\rput(0,-10){
\psline(-3.333,5)(3.333,5)(6.667,0)(3.333,-5)(-3.333,-5)(-6.667,0)(-3.333,5)
\pscircle[fillstyle=solid,fillcolor=black](0,0){1}
}
\rput(10,15){
\psline(-3.333,5)(3.333,5)(6.667,0)(3.333,-5)(-3.333,-5)(-6.667,0)(-3.333,5)
\pscircle[fillstyle=solid,fillcolor=black](0,0){1}
}
\rput(10,5){
\psline(-3.333,5)(3.333,5)(6.667,0)(3.333,-5)(-3.333,-5)(-6.667,0)(-3.333,5)
\pscircle[fillstyle=solid,fillcolor=black](0,0){1}
}
\rput(10,-5){
\psline(-3.333,5)(3.333,5)(6.667,0)(3.333,-5)(-3.333,-5)(-6.667,0)(-3.333,5)
\pscircle[fillstyle=solid,fillcolor=black](0,0){1}
}
\rput(10,-15){
\psline(-3.333,5)(3.333,5)(6.667,0)(3.333,-5)(-3.333,-5)(-6.667,0)(-3.333,5)
\pscircle[fillstyle=solid,fillcolor=black](0,0){1}
}

\pscircle[fillstyle=solid,fillcolor=black](0,20){1}
\pscircle[fillstyle=solid,fillcolor=black](20,10){1}
\pscircle[fillstyle=solid,fillcolor=black](20,-10){1}

}
\rput(80,-20){

\psline[linewidth=1.5pt,linecolor=black!50!white](-26.667,0)(-13.333,20)(13.333,20)(26.667,0)

\rput(-20,0){
\psline(-6.667,0)(-3.333,5)(3.333,5)(6.667,0)
\pscircle[fillstyle=solid,fillcolor=black](0,0){1}
}
\rput(-10,15){
\psline(-3.333,5)(3.333,5)(6.667,0)(3.333,-5)(-3.333,-5)(-6.667,0)(-3.333,5)
\pscircle[fillstyle=solid,fillcolor=black](0,0){1}
}
\rput(-10,5){
\psline(-3.333,5)(3.333,5)(6.667,0)(3.333,-5)(-3.333,-5)(-6.667,0)(-3.333,5)
\pscircle[fillstyle=solid,fillcolor=black](0,0){1}
}

\rput(20,0){
\psline(-6.667,0)(-3.333,5)(3.333,5)(6.667,0)
\pscircle[fillstyle=solid,fillcolor=black](0,0){1}
}
\rput(0,10){
\psline(-3.333,5)(3.333,5)(6.667,0)(3.333,-5)(-3.333,-5)(-6.667,0)(-3.333,5)
\pscircle[fillstyle=solid,fillcolor=black](0,0){1}
}
\rput(0,0){
\psline(-6.667,0)(-3.333,5)(3.333,5)(6.667,0)
\pscircle[fillstyle=solid,fillcolor=black](0,0){1}
}

\rput(10,15){
\psline(-3.333,5)(3.333,5)(6.667,0)(3.333,-5)(-3.333,-5)(-6.667,0)(-3.333,5)
\pscircle[fillstyle=solid,fillcolor=black](0,0){1}
}
\rput(10,5){
\psline(-3.333,5)(3.333,5)(6.667,0)(3.333,-5)(-3.333,-5)(-6.667,0)(-3.333,5)
\pscircle[fillstyle=solid,fillcolor=black](0,0){1}
}

\pscircle[fillstyle=solid,fillcolor=black](0,20){1}
\pscircle[fillstyle=solid,fillcolor=black](20,10){1}

}
\rput(40,40){

\psline[linewidth=2pt,linecolor=black!50!white](-26.667,0)(-13.333,20)(13.333,20)(26.667,0)(13.333,-20)(-13.333,-20)(-26.667,0)

\rput(-20,0){
\psline(-3.333,5)(3.333,5)(6.667,0)(3.333,-5)(-3.333,-5)(-6.667,0)(-3.333,5)
\pscircle[fillstyle=solid,fillcolor=black](0,0){1}
}
\rput(-10,15){
\psline(-3.333,5)(3.333,5)(6.667,0)(3.333,-5)(-3.333,-5)(-6.667,0)(-3.333,5)
\pscircle[fillstyle=solid,fillcolor=black](0,0){1}
}
\rput(-10,5){
\psline(-3.333,5)(3.333,5)(6.667,0)(3.333,-5)(-3.333,-5)(-6.667,0)(-3.333,5)
\pscircle[fillstyle=solid,fillcolor=black](0,0){1}
}
\rput(-10,-5){
\psline(-3.333,5)(3.333,5)(6.667,0)(3.333,-5)(-3.333,-5)(-6.667,0)(-3.333,5)
\pscircle[fillstyle=solid,fillcolor=black](0,0){1}
}
\rput(-10,-15){
\psline(-3.333,5)(3.333,5)(6.667,0)(3.333,-5)(-3.333,-5)(-6.667,0)(-3.333,5)
\pscircle[fillstyle=solid,fillcolor=black](0,0){1}
}
\rput(20,0){
\psline(-3.333,5)(3.333,5)(6.667,0)(3.333,-5)(-3.333,-5)(-6.667,0)(-3.333,5)
\pscircle[fillstyle=solid,fillcolor=black](0,0){1}
}
\rput(0,10){
\psline(-3.333,5)(3.333,5)(6.667,0)(3.333,-5)(-3.333,-5)(-6.667,0)(-3.333,5)
\pscircle[fillstyle=solid,fillcolor=black](0,0){1}
}
\rput(0,0){
\psline(-3.333,5)(3.333,5)(6.667,0)(3.333,-5)(-3.333,-5)(-6.667,0)(-3.333,5)
\pscircle[fillstyle=solid,fillcolor=black](0,0){1}
}
\rput(0,-10){
\psline(-3.333,5)(3.333,5)(6.667,0)(3.333,-5)(-3.333,-5)(-6.667,0)(-3.333,5)
\pscircle[fillstyle=solid,fillcolor=black](0,0){1}
}
\rput(10,15){
\psline(-3.333,5)(3.333,5)(6.667,0)(3.333,-5)(-3.333,-5)(-6.667,0)(-3.333,5)
\pscircle[fillstyle=solid,fillcolor=black](0,0){1}
}
\rput(10,5){
\psline(-3.333,5)(3.333,5)(6.667,0)(3.333,-5)(-3.333,-5)(-6.667,0)(-3.333,5)
\pscircle[fillstyle=solid,fillcolor=black](0,0){1}
}
\rput(10,-5){
\psline(-3.333,5)(3.333,5)(6.667,0)(3.333,-5)(-3.333,-5)(-6.667,0)(-3.333,5)
\pscircle[fillstyle=solid,fillcolor=black](0,0){1}
}
\rput(10,-15){
\psline(-3.333,5)(3.333,5)(6.667,0)(3.333,-5)(-3.333,-5)(-6.667,0)(-3.333,5)
\pscircle[fillstyle=solid,fillcolor=black](0,0){1}
}

\pscircle[fillstyle=solid,fillcolor=black](0,20){1}
\pscircle[fillstyle=solid,fillcolor=black](20,10){1}
\pscircle[fillstyle=solid,fillcolor=black](20,-10){1}

}
\rput(40,0){

\psline[linewidth=2pt,linecolor=black!50!white](-26.667,0)(-13.333,20)(13.333,20)(26.667,0)(13.333,-20)(-13.333,-20)(-26.667,0)

\rput(-20,0){
\psline(-3.333,5)(3.333,5)(6.667,0)(3.333,-5)(-3.333,-5)(-6.667,0)(-3.333,5)
\pscircle[fillstyle=solid,fillcolor=black](0,0){1}
}
\rput(-10,15){
\psline(-3.333,5)(3.333,5)(6.667,0)(3.333,-5)(-3.333,-5)(-6.667,0)(-3.333,5)
\pscircle[fillstyle=solid,fillcolor=black](0,0){1}
}
\rput(-10,5){
\psline(-3.333,5)(3.333,5)(6.667,0)(3.333,-5)(-3.333,-5)(-6.667,0)(-3.333,5)
\pscircle[fillstyle=solid,fillcolor=black](0,0){1}
}
\rput(-10,-5){
\psline(-3.333,5)(3.333,5)(6.667,0)(3.333,-5)(-3.333,-5)(-6.667,0)(-3.333,5)
\pscircle[fillstyle=solid,fillcolor=black](0,0){1}
}
\rput(-10,-15){
\psline(-3.333,5)(3.333,5)(6.667,0)(3.333,-5)(-3.333,-5)(-6.667,0)(-3.333,5)
\pscircle[fillstyle=solid,fillcolor=black](0,0){1}
}
\rput(20,0){
\psline(-3.333,5)(3.333,5)(6.667,0)(3.333,-5)(-3.333,-5)(-6.667,0)(-3.333,5)
\pscircle[fillstyle=solid,fillcolor=black](0,0){1}
}
\rput(0,10){
\psline(-3.333,5)(3.333,5)(6.667,0)(3.333,-5)(-3.333,-5)(-6.667,0)(-3.333,5)
\pscircle[fillstyle=solid,fillcolor=black](0,0){1}
}
\rput(0,0){
\psline(-3.333,5)(3.333,5)(6.667,0)(3.333,-5)(-3.333,-5)(-6.667,0)(-3.333,5)
\pscircle[fillstyle=solid,fillcolor=black](0,0){1}
}
\rput(0,-10){
\psline(-3.333,5)(3.333,5)(6.667,0)(3.333,-5)(-3.333,-5)(-6.667,0)(-3.333,5)
\pscircle[fillstyle=solid,fillcolor=black](0,0){1}
}
\rput(10,15){
\psline(-3.333,5)(3.333,5)(6.667,0)(3.333,-5)(-3.333,-5)(-6.667,0)(-3.333,5)
\pscircle[fillstyle=solid,fillcolor=black](0,0){1}
}
\rput(10,5){
\psline(-3.333,5)(3.333,5)(6.667,0)(3.333,-5)(-3.333,-5)(-6.667,0)(-3.333,5)
\pscircle[fillstyle=solid,fillcolor=black](0,0){1}
}
\rput(10,-5){
\psline(-3.333,5)(3.333,5)(6.667,0)(3.333,-5)(-3.333,-5)(-6.667,0)(-3.333,5)
\pscircle[fillstyle=solid,fillcolor=black](0,0){1}
}
\rput(10,-15){
\psline(-3.333,5)(3.333,5)(6.667,0)(3.333,-5)(-3.333,-5)(-6.667,0)(-3.333,5)
\pscircle[fillstyle=solid,fillcolor=black](0,0){1}
}

\pscircle[fillstyle=solid,fillcolor=black](0,20){1}
\pscircle[fillstyle=solid,fillcolor=black](20,10){1}
\pscircle[fillstyle=solid,fillcolor=black](20,-10){1}

}

\rput(0,60){

\psline[linewidth=2pt,linecolor=black!50!white](26.667,0)(13.333,-20)(-13.333,-20)(-26.667,0)

\rput(-20,0){
\psline(6.667,0)(3.333,-5)(-3.333,-5)(-6.667,0)\pscircle[fillstyle=solid,fillcolor=black](0,0){1}
}
\rput(-10,-5){
\psline(-3.333,5)(3.333,5)(6.667,0)(3.333,-5)(-3.333,-5)(-6.667,0)(-3.333,5)
\pscircle[fillstyle=solid,fillcolor=black](0,0){1}
}
\rput(-10,-15){
\psline(-3.333,5)(3.333,5)(6.667,0)(3.333,-5)(-3.333,-5)(-6.667,0)(-3.333,5)
\pscircle[fillstyle=solid,fillcolor=black](0,0){1}
}
\rput(20,0){
\psline(6.667,0)(3.333,-5)(-3.333,-5)(-6.667,0)\pscircle[fillstyle=solid,fillcolor=black](0,0){1}
}
\rput(0,0){
\psline(6.667,0)(3.333,-5)(-3.333,-5)(-6.667,0)
\pscircle[fillstyle=solid,fillcolor=black](0,0){1}
}
\rput(0,-10){
\psline(-3.333,5)(3.333,5)(6.667,0)(3.333,-5)(-3.333,-5)(-6.667,0)(-3.333,5)
\pscircle[fillstyle=solid,fillcolor=black](0,0){1}
}
\rput(10,-5){
\psline(-3.333,5)(3.333,5)(6.667,0)(3.333,-5)(-3.333,-5)(-6.667,0)(-3.333,5)
\pscircle[fillstyle=solid,fillcolor=black](0,0){1}
}
\rput(10,-15){
\psline(-3.333,5)(3.333,5)(6.667,0)(3.333,-5)(-3.333,-5)(-6.667,0)(-3.333,5)
\pscircle[fillstyle=solid,fillcolor=black](0,0){1}
}

\pscircle[fillstyle=solid,fillcolor=black](20,-10){1}

}
\rput(0,20){

\psline[linewidth=2pt,linecolor=black!50!white](-26.667,0)(-13.333,20)(13.333,20)(26.667,0)(13.333,-20)(-13.333,-20)(-26.667,0)

\rput(-20,0){
\psline(-3.333,5)(3.333,5)(6.667,0)(3.333,-5)(-3.333,-5)(-6.667,0)(-3.333,5)
\pscircle[fillstyle=solid,fillcolor=black](0,0){1}
}
\rput(-10,15){
\psline(-3.333,5)(3.333,5)(6.667,0)(3.333,-5)(-3.333,-5)(-6.667,0)(-3.333,5)
\pscircle[fillstyle=solid,fillcolor=black](0,0){1}
}
\rput(-10,5){
\psline(-3.333,5)(3.333,5)(6.667,0)(3.333,-5)(-3.333,-5)(-6.667,0)(-3.333,5)
\pscircle[fillstyle=solid,fillcolor=black](0,0){1}
}
\rput(-10,-5){
\psline(-3.333,5)(3.333,5)(6.667,0)(3.333,-5)(-3.333,-5)(-6.667,0)(-3.333,5)
\pscircle[fillstyle=solid,fillcolor=black](0,0){1}
}
\rput(-10,-15){
\psline(-3.333,5)(3.333,5)(6.667,0)(3.333,-5)(-3.333,-5)(-6.667,0)(-3.333,5)
\pscircle[fillstyle=solid,fillcolor=black](0,0){1}
}
\rput(20,0){
\psline(-3.333,5)(3.333,5)(6.667,0)(3.333,-5)(-3.333,-5)(-6.667,0)(-3.333,5)
\pscircle[fillstyle=solid,fillcolor=black](0,0){1}
}
\rput(0,10){
\psline(-3.333,5)(3.333,5)(6.667,0)(3.333,-5)(-3.333,-5)(-6.667,0)(-3.333,5)
\pscircle[fillstyle=solid,fillcolor=black](0,0){1}
}
\rput(0,0){
\psline(-3.333,5)(3.333,5)(6.667,0)(3.333,-5)(-3.333,-5)(-6.667,0)(-3.333,5)
\pscircle[fillstyle=solid,fillcolor=black](0,0){1}
}
\rput(0,-10){
\psline(-3.333,5)(3.333,5)(6.667,0)(3.333,-5)(-3.333,-5)(-6.667,0)(-3.333,5)
\pscircle[fillstyle=solid,fillcolor=black](0,0){1}
}
\rput(10,15){
\psline(-3.333,5)(3.333,5)(6.667,0)(3.333,-5)(-3.333,-5)(-6.667,0)(-3.333,5)
\pscircle[fillstyle=solid,fillcolor=black](0,0){1}
}
\rput(10,5){
\psline(-3.333,5)(3.333,5)(6.667,0)(3.333,-5)(-3.333,-5)(-6.667,0)(-3.333,5)
\pscircle[fillstyle=solid,fillcolor=black](0,0){1}
}
\rput(10,-5){
\psline(-3.333,5)(3.333,5)(6.667,0)(3.333,-5)(-3.333,-5)(-6.667,0)(-3.333,5)
\pscircle[fillstyle=solid,fillcolor=black](0,0){1}
}
\rput(10,-15){
\psline(-3.333,5)(3.333,5)(6.667,0)(3.333,-5)(-3.333,-5)(-6.667,0)(-3.333,5)
\pscircle[fillstyle=solid,fillcolor=black](0,0){1}
}

\pscircle[fillstyle=solid,fillcolor=black](0,20){1}
\pscircle[fillstyle=solid,fillcolor=black](20,10){1}
\pscircle[fillstyle=solid,fillcolor=black](20,-10){1}

}
\rput(0,-20){

\psline[linewidth=1.5pt,linecolor=black!50!white](-26.667,0)(-13.333,20)(13.333,20)(26.667,0)

\rput(-20,0){
\psline(-6.667,0)(-3.333,5)(3.333,5)(6.667,0)
\pscircle[fillstyle=solid,fillcolor=black](0,0){1}
}
\rput(-10,15){
\psline(-3.333,5)(3.333,5)(6.667,0)(3.333,-5)(-3.333,-5)(-6.667,0)(-3.333,5)
\pscircle[fillstyle=solid,fillcolor=black](0,0){1}
}
\rput(-10,5){
\psline(-3.333,5)(3.333,5)(6.667,0)(3.333,-5)(-3.333,-5)(-6.667,0)(-3.333,5)
\pscircle[fillstyle=solid,fillcolor=black](0,0){1}
}

\rput(20,0){
\psline(-6.667,0)(-3.333,5)(3.333,5)(6.667,0)
\pscircle[fillstyle=solid,fillcolor=black](0,0){1}
}
\rput(0,10){
\psline(-3.333,5)(3.333,5)(6.667,0)(3.333,-5)(-3.333,-5)(-6.667,0)(-3.333,5)
\pscircle[fillstyle=solid,fillcolor=black](0,0){1}
}
\rput(0,0){
\psline(-6.667,0)(-3.333,5)(3.333,5)(6.667,0)
\pscircle[fillstyle=solid,fillcolor=black](0,0){1}
}

\rput(10,15){
\psline(-3.333,5)(3.333,5)(6.667,0)(3.333,-5)(-3.333,-5)(-6.667,0)(-3.333,5)
\pscircle[fillstyle=solid,fillcolor=black](0,0){1}
}
\rput(10,5){
\psline(-3.333,5)(3.333,5)(6.667,0)(3.333,-5)(-3.333,-5)(-6.667,0)(-3.333,5)
\pscircle[fillstyle=solid,fillcolor=black](0,0){1}
}

\pscircle[fillstyle=solid,fillcolor=black](0,20){1}
\pscircle[fillstyle=solid,fillcolor=black](20,10){1}

}

\rput(-40,40){

\rput(20,0){
\psline(-3.333,5)(3.333,5)(6.667,0)(3.333,-5)(-3.333,-5)(-6.667,0)(-3.333,5)
\pscircle[fillstyle=solid,fillcolor=black](0,0){1}
}

\rput(10,15){
\psline(3.333,5)(6.667,0)(3.333,-5)
}
\rput(10,5){
\psline(3.333,5)(6.667,0)(3.333,-5)
}
\rput(10,-5){
\psline(3.333,5)(6.667,0)(3.333,-5)
}
\rput(10,-15){
\psline(3.333,5)(6.667,0)(3.333,-5)
}

\pscircle[fillstyle=solid,fillcolor=black](20,10){1}
\pscircle[fillstyle=solid,fillcolor=black](20,-10){1}

}
\rput(-40,0){

\rput(20,0){
\psline(-3.333,5)(3.333,5)(6.667,0)(3.333,-5)(-3.333,-5)(-6.667,0)(-3.333,5)
\pscircle[fillstyle=solid,fillcolor=black](0,0){1}
}

\rput(10,15){
\psline(3.333,5)(6.667,0)(3.333,-5)
}
\rput(10,5){
\psline(3.333,5)(6.667,0)(3.333,-5)
}
\rput(10,-5){
\psline(3.333,5)(6.667,0)(3.333,-5)
}
\rput(10,-15){
\psline(3.333,5)(6.667,0)(3.333,-5)
}

\pscircle[fillstyle=solid,fillcolor=black](20,10){1}
\pscircle[fillstyle=solid,fillcolor=black](20,-10){1}

}

\pscircle[linewidth=1.3,linecolor=black!50!white](0,20){3}
\pscircle[linewidth=1.3,linecolor=black!50!white](40,0){3}
\pscircle[linewidth=1.3,linecolor=black!50!white](40,40){3}
\pscircle[linewidth=1.3,linecolor=black!50!white](80,20){3}
\pscircle[linewidth=1.3,linecolor=black!50!white](0,-20){3}
\pscircle[linewidth=1.3,linecolor=black!50!white](80,-20){3}
\pscircle[linewidth=1.3,linecolor=black!50!white](0,60){3}
\pscircle[linewidth=1.3,linecolor=black!50!white](80,60){3}

\pscircle[fillstyle=solid,fillcolor=black](40,-20){1}

\end{pspicture}
\caption{Part of a nested lattice $\L \subset \L_1 \subset \Rbb^2$. Black points are elements of the fine lattice $\L_1$ and gray circles are elements of the coarse lattice $\L$. The Voronoi regions for the fine and coarse lattice are drawn in black and gray respectively. A nested lattice code is the set of all fine lattice points within the Voronoi region of the coarse lattice centered on the origin.}
\label{f:nestedlattices}
\end{center}
\end{figure}

Let $\mathcal{B}(r)$ denote an $n$-dimensional ball of radius $r$,
\begin{align}
\Bm(r) \triangleq \{\sb : \|\sb\| \leq r,~\sb \in \Rbb^n\}
\end{align} and let Vol$(\Bm(r))$ denote its volume.

\begin{definition}[Covering Radius] The \textit{covering radius} of a lattice $\Lambda$ is the smallest real number $\rcov$ such that \mbox{$\Rbb^n \subseteq \Lambda + \Bm(\rcov)$}.
\end{definition} 

\begin{definition}[Effective Radius] \label{d:reffec} The \textit{effective radius} of a lattice with Voronoi region $\Vm$ is the real number $\reffec$ that satisfies $\mathsf{Vol}(\mathcal{B}(\reffec)) = \mathsf{Vol}(\mathcal{V})$.
\end{definition} 

\begin{definition}[Moments] The \textit{second moment} of a lattice $\Lambda$ is  defined as the second moment per dimension of a uniform distribution over the fundamental Voronoi region $\Vm$,
\begin{align}
\sigma_\Lambda^2 = \frac{1}{n\mathsf{Vol}(\Vm)} \int_{\Vm}{\|\xb\|^2 d\xb} \ .
\end{align} The \textit{normalized second moment} of a lattice is given by
\begin{align}
G(\Lambda) = \frac{\sigma_\Lambda^2}{(\mathsf{Vol}(\Vm))^{2/n}} \ . \label{e:normsecmoment}
\end{align}
\end{definition}

The following three definitions are the basis for proving AWGN channel coding theorems using nested lattice codes. Let $\Lambda^{(n)}$ denote a sequence of lattices indexed by their dimension.

\begin{definition}[Covering Goodness] A sequence of lattices $\Lambda^{(n)} \subset \Rbb^n$ is \textit{good for covering} if
\begin{align}
\lim_{n \rightarrow \infty}{\frac{\rcov^{(n)}}{\reffec^{(n)}}} = 1 \ .
\end{align} Such lattices were shown to exist by Rogers \cite{rogers59}. \end{definition}

\begin{definition}[Quantization Goodness] A sequence of lattices $\Lambda^{(n)} \subset \Rbb^n$ is \textit{good for mean-squared error (MSE) quantization} if
\begin{align}
\lim_{n \rightarrow \infty}{G(\Lambda^{(n)})} = \frac{1}{2 \pi e} \ .
\end{align} Zamir, Feder, and Poltyrev showed that sequences of such lattices exist in \cite{zf96}.
\end{definition}

\begin{definition}[AWGN Goodness] \label{d:awgngood} Let $\mathbf{z}$ be a length-$n$ i.i.d. Gaussian vector, $\zb \sim \mathcal{N}(0, \sigma^2_Z \mathbf{I}^{n \times n})$. The volume-to-noise ratio of a lattice is given by
\begin{align}
\mu(\Lambda,\epsilon) = \frac{(\V(\Vm))^{2/n}}{\sigma^2_Z} \label{e:volnoiseratio}
\end{align} where $\sigma^2_Z$ is chosen such that $\Pr\{ \mathbf{z} \notin \Vm\} = \epsilon$. A sequence of lattices $\Lambda^{(n)}$ is \textit{good for AWGN} if 
\begin{align}
~~~~~~~~~~~~~~~~~ \lim_{n \rightarrow \infty}{\mu(\Lambda^{(n)},\epsilon)} = 2 \pi e  ~~~~ \forall \epsilon \in (0,1)
\end{align} and, for fixed volume-to-noise ratio greater than $2 \pi e$, $\Pr\{ \mathbf{z} \notin \Vm^{(n)}\}$ decays exponentially in $n$. In \cite{poltyrev94}, Poltyrev demonstrated the existence of such lattices. \end{definition}

\subsection{Lattice Constructions}\label{s:latticeconstruct}

Our nested lattice codes are a slight variant of those used by Erez and Zamir to approach the capacity of a point-to-point AWGN channel \cite{ez04}. As in their considerations, we will have a coarse lattice that is good for covering, quantization, and AWGN and a fine lattice that is good for AWGN. We generalize this construction to include multiple nested fine lattices all of which are good for AWGN. This will allow each transmitter to operate at a different rate. 

\begin{lemma}[Erez-Litsyn-Zamir] \label{l:coarselattice} There exists a sequence of lattices $\Lambda^{(n)}$ that is simultaneously good for covering, quantization, and AWGN. \end{lemma} This is a corollary of their main result which develops lattices that are good in all the above senses as well as for packing \cite[Theorem 5]{elz05}. Note that these lattices are built using Construction A which is described below.

We will use a coarse lattice $\Lambda$ of dimension $n$ from Lemma \ref{l:coarselattice} scaled such that its second moment is equal to $P$. Let $\mathbf{B} \in \Rbb^{n \times n}$ denote the generator matrix of this lattice. Our fine lattices are defined using the following procedure (the first three steps of which are often referred to as Construction A \cite{loeliger97,elz05}):

\begin{enumerate}
\item Draw a matrix $\mathbf{G}_L \in \Fbb^{n \times k_L}_p$ with every element chosen i.i.d. according to the uniform distribution over $\{0,1,2,\ldots, p-1\}$. Recall that $p$ is prime. 
\item Define the codebook $\mathcal{C}_L$ as follows:
\begin{align}
\mathcal{C}_L = \left\{ \mathbf{c} =  \mathbf{G}_L\wb:\mathbf{w} \in \mathbb{F}_p^{k_L} \right\}.
\end{align} All operations in this step are over $\mathbb{F}_p$. 

\item Form the lattice $\tilde{\Lambda}_L$ by projecting the codebook into the reals by $g(\cdot)$, scaling down by a factor of $p$, and placing a copy at every integer vector. This tiles the codebook over $\mathbb{R}^n$,
\begin{align}
\tilde{\Lambda}_L = p^{-1}g( \mathcal{C}_L) + \mathbb{Z}^n \ .
\end{align}
\item Rotate $\tilde{\Lambda}_L$ by the generator matrix of the coarse nested lattice to get the fine lattice for transmitter $L$,
\begin{align}
\Lambda_L = \mathbf{B}\tilde{\Lambda}_L \ . 
\end{align}
\item Repeat steps 1) - 4) for each transmitter $\ell = 1,2,\ldots, L-1$ by replacing $\Gb_L$ with $\Gb_\ell$ which is defined to be the first $k_\ell$ columns of $\Gb_L$.
\end{enumerate}

Recall that $k_1 \geq \cdots \geq k_L$. Any pair of fine lattices $\Lambda_{\ell_1},\Lambda_{\ell_2}, 1 \leq \ell_1 < \ell_2 < L$ are nested since all elements of $\mathcal{C}_{\ell_1}$ can be found from $\mathbf{G}_{\ell_2}$ by multiplying by all $\mathbf{w} \in \mathbb{F}^{n \times k_{\ell_2}}$ with zeros in the last $\ell_2 - \ell_1$ elements.  Also observe that $\Lambda = \mathbf{B} \Zbb^n$ is nested within each fine lattice by construction. Therefore, the lattices are nested in the desired order, $\Lambda \subseteq \Lambda_L \subseteq \cdots \subseteq \Lambda_1$. 

We now enforce that all the underlying generator matrices $\mathbf{G}_\ell$ are full rank. By the union bound, we get that:

\begin{align}
\Pr\left(\bigcup_{\ell=1}^L{\left\{\mathsf{rank}(\Gb_\ell) < k_\ell\right\}}\right) &\leq \sum_{\ell=1}^L{\sum_{ \begin{subarray}{l}\wb \in \Fbb_p^{k_\ell} \\ \wb \neq \mathbf{0}\end{subarray}}{\Pr\left\{\Gb_\ell \wb = \mathbf{0}\right\}}} \nonumber \\
& = p^{-n} \sum_{\ell = 1}^L{(p^{k_\ell} - 1)}
\end{align}

Thus, by choosing $p$ and $k_1,\ldots,k_L$ to grow appropriately with $n$, all matrices $\Gb_1, \ldots, \Gb_L$ are full rank with probability that goes to $1$ with $n$. Note that if $\mathbf{G}_\ell$ has full rank, then the number of fine lattice points in the fundamental Voronoi region $\mathcal{V}$ of the coarse lattice is given by $|\Lambda_\ell \cap \mathcal{V}| = p^{k_\ell}$ so that the rate of the $\ell^{\text{th}}$ nested lattice code $\mathcal{L}_\ell = \Lambda_\ell \cap \mathcal{V} $ is
\begin{align}
r_\ell = \frac{1}{n} \log{|\Lambda_\ell \cap \mathcal{V}|}=  \frac{k_\ell}{n}\log{p} = R_\ell
\end{align} as desired. (In the complex-valued case, we set $r_\ell = R_\ell/2$.)  
In Appendix \ref{s:awgngood}, we show that the fine lattices are AWGN good so long as $\frac{n}{p} \rightarrow 0$ as $n$ grows. There are many choices of $p$ and $k_1, \ldots, k_L$ that will ensure that the fine lattices have the desired properties. One possibility is to let $p$ grow like $n \log{n}$ and set $k_\ell = \lfloor nR_\e(\log{p})\n \rfloor$.

\begin{remark}
We require that the fine lattices are generated from full-rank submatrices of the same finite field codebook so that it is possible to compute linear equations over messages with different rates. The full rank condition on the coarse lattice allows us to move between lattice equations and equations of finite field messages.
\end{remark}

In \cite{elz05, kp07report}, some useful properties of nested lattices derived from Construction A are established. These apply to our construction as well and we repeat them below.

\begin{lemma} \label{l:latprops}
Let $\Lambda_\ell(i)$ denote the $\ith$ point in the $\ellth$ nested lattice code $\mathcal{L}_\ell = \Lambda_\ell \cap \Vm$ for $i = 0,1,2,\ldots, p^{k_\ell}-1$ from the random lattice construction above. We have that:
\begin{itemize}
\item $\Lambda_\ell(i)$ is uniformly distributed over $p^{-1} \Lambda \cap \Vm$.
\vspace{0.1in}
\item For any $i_1 \neq i_2, ~[\Lambda_\ell(i_1) - \Lambda_\ell(i_2)] ~\modl$ is uniformly distributed over $\{p^{-1} \Lambda\} \cap \Vm$.
\end{itemize}
\end{lemma}

Thus, each fine lattice can be interpreted as a diluted version of a scaled down coarse lattice $p^{-1} \Lambda$. 

\subsection{Integer Combinations of Lattice Points}

Our scheme relies on mapping messages from a finite field to codewords from a nested lattice code. The relay will first decode an integer combination of lattice codewords and then convert this into an equation of the messages. 

\begin{definition}[Lattice Equation] A \textit{lattice equation} $\vb$ is an integer combination of lattice codewords $\tb_\ell \in \mathcal{L}_\ell$ modulo the coarse lattice,
\begin{align}
\vb = \left[ \sum_{\ell = 1}^L{a_\ell \tb_\ell}\right] \modl
\end{align} for some coefficients $a_\ell \in \Zbb$.
\end{definition} Note that the lattice equation takes values on the finest lattice in the summation. That is, if $a_{1},\ldots, a_{\ell -1} = 0$ then the lattice equation $\vb$ only takes values on $\mathcal{L}_\ell = \Lambda_\ell \cap \mathcal{V}$. 

\begin{lemma} \label{l:consArank}
Any lattice $\Lambda$ that results from Construction A has a full-rank generator matrix $\mathbf{B}$.
\end{lemma}

\begin{proof}
Note that $\Zbb^n \subset \Lambda$ so that $\Lambda$ contains all of the unit vectors by default. Thus, $\mathbf{B}$ spans $\Rbb^n$ and is full rank.
\end{proof}


Since our nested lattice codes are built using nested finite field codes, it is possible to map messages to lattice points and back while preserving linearity. The next two lemmas make this notion precise.


\begin{lemma}\label{l:fieldtolattice}
Let $\mathbf{w}_\ell$ be a message in $\Fbb_p^{k_\ell}$ that is zero-padded to length $k$. The function
\begin{align}
\phi(\wb_\ell) = \left[\Lb p^{-1}g( \Gb \mathbf{w}_\ell) \right]\modl \label{e:fieldtolattice}
\end{align} is a one-to-one map between the set of such messages and the elements of the nested lattice code $\mathcal{L}_\ell = \L_\ell \cap \mathcal{V}$.
\end{lemma}

\begin{proof}
Since the last $k - k_{\ell}$ elements of $\mathbf{w}_\ell$ are zero, multiplying the message by $\mathbf{G}$ is the same as multiplying the first $k_{\ell}$ elements by $\mathbf{G}_\ell$. Since $\Gb_\ell$ is assumed to be full rank, it takes $\wb_\ell$ to a unique point in the finite field codebook $\Cm_\ell$. The function $g$ simply maps finite field elements to integers and $p^{-1}$ is a rescaling so $p^{-1} g(\Gb \wb_\ell)$ maps $\wb_\ell$ to a unique point in $[0,1)^n$.  Lemma \ref{l:consArank} shows that $\mathbf{B}$ is full rank so we just need show that the $\modl$ operation is a bijection between $\mathbf{B} [0,1)^n$ and $\Vm$. Assume, for the sake of a contradiction, $\exists x,y \in \Bb [0,1)^n, x\neq y$ such that {$[x]\mod\L = [y] \mod\L$}. This implies that $x - Q_\L(x) = y - Q_\L(y)$. Now multiply both sides by $\mathbf{B}^{-1}$ and then take the modulus with respect to $\Zbb^n$,
\begin{align*}
[\mathbf{B}^{-1} (x - Q_\L(x))]\modz &= [\mathbf{B}^{-1}(y - Q_\L(y))]\modz\\
[\mathbf{B}^{-1} x ]\modz &= [\mathbf{B}^{-1}y ]\modz \\
x&= y
\end{align*} where the second line follows since for any $\lambda \in \Lambda, ~\mathbf{B}\n \lambda \in \Zbb^n$. A contradiction has been reached which shows that $\modl$ is a bijection. Combining this with the fact that the finite field and the nested lattice code have the same number of elements, $|\Fbb_p^{k_\ell}| = |\L_\ell \cap \Vm| = p^{k_\ell}$, shows that $\phi_\ell$ is a one-to-one map.
\end{proof}


\begin{lemma}\label{l:msgeqn}
Let $\ub = \bigoplus_\ell{q_\ell \wb_\ell}$ be the desired equation for some coefficients $q_\ell \in \Fbb_p$ and messages $\wb_\ell \in \Fbb_p^{k_\ell}$ zero-padded to length $k$. Assume the messages are mapped to nested lattice codewords, $\tb_\ell = \phi(\wb_\ell)$, and let \mbox{$\vb = [\sum{a_\ell \tb_\ell}] ~\modl$} denote the lattice equation for some $a_\ell \in \Zbb$ such that $q_\ell = g\n([a_\ell] \mod{p})$. Then the desired equation can be obtained using $\ub = \phi^{-1}(\vb)$ where \begin{align}
\phi^{-1}(\vb) &= (\mathbf{G}^T \mathbf{G})^{-1} \mathbf{G}^T g\n\left( p [\mathbf{B}\n \mathbf{v}]\modz \right).
\end{align} 
\end{lemma}
\begin{proof} Recall that since $\Bb$ is the generator matrix of $\Lambda$, $\mathbf{B}\n \L= \Zbb^n$. Also note that since $\wb_\ell$ is zero-padded to length $k$, then multiplying by $\Gb$ has the same effect as multiplying the original message by $\Gb_\ell$. We have that
\begin{align}
&[\Bb \n \vb] \modz \\
&= \left[\Bb\n \sum_{\ell=1}^L{a_\ell} \tb_\ell -  \Bb\n Q_\L\left( \sum_{\ell=1}^L{a_\ell} \tb_\ell\right)\right] \modz \\
&\overset{\mbox{\footnotesize{(a)}}}{=}\left[\Bb\n \sum_{\ell=1}^L{a_\ell} \tb_\ell \right] \modz \\
&\overset{\mbox{\footnotesize{(b)}}}= \left[\sum_{\ell=1}^L{a_\ell} \bigg(p\n g(\Gb \wb_\e) \right. \nonumber \\ & \qquad \qquad~~~- ~\Bb\n Q_\L\Big(\Bb p\n g(\Gb \wb_\e)\Big)\bigg) \Bigg] \modz \\
&\overset{\mbox{\footnotesize{(c)}}}{=} \left[\sum_{\ell=1}^L{a_\ell} p\n g(\Gb \wb_\e )\right] \modz 
\end{align} where (a) and (c) follow since $Q_{\L}(\cdot)$ is an element of $\L$ so $\Bb\n Q_{\L}(\cdot)$ is an element of $\Zbb^n$ and (b) follows using (\ref{e:fieldtolattice}). Multiplying by $p$ and applying (\ref{e:modprop4}) yields 
\begin{align}
p[\Bb \n \vb] \modz &= \left[\sum_{\ell=1}^L{a_\ell} g(\Gb \wb_\e )\right] \modpz \\ &\overset{\mbox{\footnotesize{(d)}}}{=} \left[g\left(\bigoplus_{\ell=1}^L{q_\ell \Gb \wb_\ell}\right)\right]\modpz \\
&= g\left(\bigoplus_{\ell=1}^L{q_\ell \Gb \wb_\ell}\right)
\end{align} where (d) follows since $g$ maps between $\{0,1,\ldots,p-1\}$ and $\Fbb_p$ and $q_\ell = g\n([a] \mod{p})$.\\ Applying $g\n$ to move back to the finite field we get
\begin{align}
g\n\left(p[\Bb \n \vb] \modz \right)=  \Gb \bigoplus_{\ell=1}^L{q_\ell  \wb_\ell}
\end{align} Finally, note that $\left(\Gb^T \Gb\right)\n \Gb^T$ is the left-inverse of $\Gb$ which implies that $\phi\n(\vb) = \ub$.  
\end{proof}

\section{Compute-and-Forward} \label{s:basiccompproof}

In this section, we provide a detailed description of our coding scheme. See Figure \ref{f:systemdiagram} for a block diagram. The following four steps are a basic outline:
\begin{enumerate}
\item Each transmitter maps its message from the finite field onto an element of a nested lattice code.
\item The lattice codewords are transmitted over the channel.
\item Each relay decodes a linear equation of the lattice codewords.
\item These lattice equations are mapped back to the finite field to
get the desired linear combination of messages. 
\end{enumerate} We begin with the proof for the real-valued case and then move on to the complex-valued case.

\subsection{Real-Valued Channel Models} \label{s:realcompproof}

When a relay attempts to decode an integer combination of the lattice points, it must overcome two sources of noise. One is simply the channel noise $\zb$. The other is due to the fact that the channel coefficients that are often not exactly equal to the desired equation coefficients. As a result, part of the noise stems from the codewords themselves (sometimes referred to as ``self-noise''). To overcome this issue, the transmitters will dither their lattice points using common randomness that is also known to the relays. This dithering makes the transmitted codewords independent from the underlying lattice points. Since our scheme works with respect to expectation over these dither vectors, then it can be shown that (at least) one set of good fixed dither vectors exists (which means that no common randomness is actually necessary). We defer the proof of this fact to Appendix \ref{s:fixeddithers}. The following lemma from \cite{ez04} captures a key property of dithered nested lattice codes.

\begin{lemma}[Erez-Zamir] \label{l:dither}
Let $\mathbf{t}$ be a random vector with an arbitrary distribution over $\Rbb^n$. If $\db$ is independent of $\tb$ and uniformly distributed over $\Vm$, then $[\tb - \db]~\modl$ is also independent of $\tb$ and uniformly distributed over $\Vm$.
\end{lemma} 

We now set out to prove that the relays can reliably recover integer combinations of transmitted lattice points.

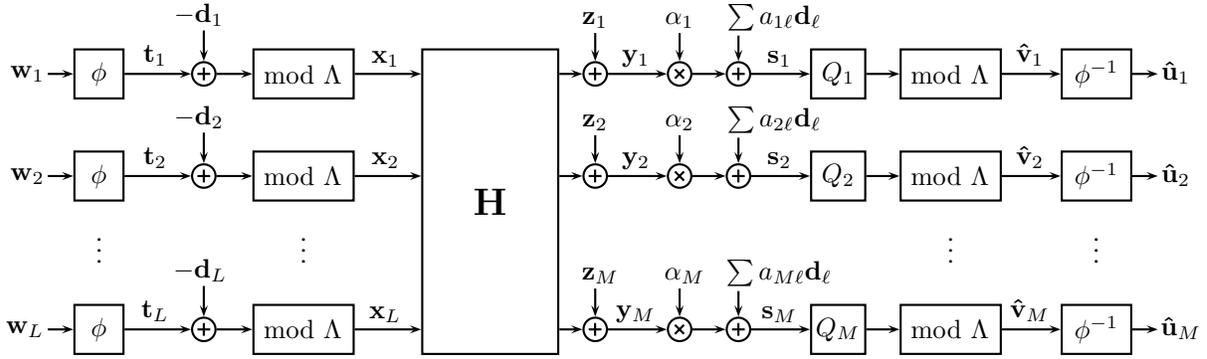
\begin{figure*}[!t]
\begin{center}
\psset{unit=0.68mm}
\begin{pspicture}(0,-23)(222,50)

\rput(0,32){
\rput(1,0){$\wb_1$}\psline{->}(5,0)(10,0)
\psframe(10,-5)(20,5) \rput(15,0){$\phi$} 
\psline{->}(20,0)(33,0) \rput(26,3.6){$\tb_1$}
\pscircle(35.5,0){2.5} \psline{-}(34.25,0)(36.75,0)\psline{-}(35.5,-1.25)(35.5,1.25) 
\psline{->}(35.5,8.5)(35.5,2.5) \rput(34.5,11.5){$-\db_1$}
\psline{->}(38,0)(45,0)
\psframe(45,-5)(65,5) \rput(55,0){$\modl$}
\psline{->}(65,0)(78,0) \rput(71,3){$\xb_1$}
}

\rput(0,12){
\rput(1,0){$\wb_2$}\psline{->}(5,0)(10,0)
\psframe(10,-5)(20,5) \rput(15,0){$\phi$} 
\psline{->}(20,0)(33,0) \rput(26,3.6){$\tb_2$}
\pscircle(35.5,0){2.5} \psline{-}(34.25,0)(36.75,0)\psline{-}(35.5,-1.25)(35.5,1.25) 
\psline{->}(35.5,8.5)(35.5,2.5) \rput(34.5,11.5){$-\db_2$}
\psline{->}(38,0)(45,0)
\psframe(45,-5)(65,5) \rput(55,0){$\modl$}
\psline{->}(65,0)(78,0) \rput(71,3){$\xb_2$}
}

\rput(15,-1){$\vdots$}
\rput(55,-1){$\vdots$}

\rput(0,-18){
\rput(0.5,0){$\wb_L$}\psline{->}(5,0)(10,0)
\psframe(10,-5)(20,5) \rput(15,0){$\phi$} 
\psline{->}(20,0)(33,0) \rput(26,3.6){$\tb_L$}
\pscircle(35.5,0){2.5} \psline{-}(34.25,0)(36.75,0)\psline{-}(35.5,-1.25)(35.5,1.25) 
\psline{->}(35.5,8.5)(35.5,2.5) \rput(34.8,11.5){$-\db_L$}
\psline{->}(38,0)(45,0)
\psframe(45,-5)(65,5) \rput(55,0){$\modl$}
\psline{->}(65,0)(78,0) \rput(71,3){$\xb_L$}
}

\psframe(78,-23)(105,37)
\rput(91.5,7){\Large{$\mathbf{H}$}}

\rput(105,0){

\rput(0,32){
\psline{->}(0,0)(4.5,0)
\pscircle(7,0){2.5}  \rput(7,0){\psline{-}(0,-1.25)(0,1.25) \psline{-}(-1.25,0)(1.25,0)} 
\psline{->}(7,8)(7,2.5) \rput(7,10.5){$\zb_1$}
\psline{->}(9.5,0)(21,0) \rput(15,3){$\yb_1$}
\pscircle(23.5,0){2.5}  \rput{45}(23.5,0){\psline{-}(0,-1.25)(0,1.25) \psline{-}(-1.25,0)(1.25,0)} 
\psline{->}(23.5,8)(23.5,2.5) \rput(23.5,10.5){$\alpha_1$}
\psline{->}(26,0)(32.5,0)
\pscircle(35,0){2.5}  \rput(35,0){\psline{-}(0,-1.25)(0,1.25) \psline{-}(-1.25,0)(1.25,0)} 
\psline{->}(35,7.5)(35,2.5) \rput(42,11){$\sum a_{1\ell}\db_\ell$}
\psline{->}(37.5,0)(49,0) \rput(43,3){$\sb_1$}
\psframe(49,-5)(60,5) \rput(54.5,0){$Q_1$} 
\psline{->}(60,0)(66.5,0)
\psframe(66.5,-5)(86.5,5) \rput(76.5,0){$\modl$}
\psline{->}(86.5,0)(98,0) \rput(92,3.7){$\vbh_1$}
\psframe(98,-5)(112,5) \rput(105,0){$\phi^{-1}$} 
\psline{->}(112,0)(117,0)\rput(120.5,0.3){$\ubh_1$}
}

\rput(0,12){
\psline{->}(0,0)(4.5,0)
\pscircle(7,0){2.5}  \rput(7,0){\psline{-}(0,-1.25)(0,1.25) \psline{-}(-1.25,0)(1.25,0)} 
\psline{->}(7,8)(7,2.5) \rput(7,10.5){$\zb_2$}
\psline{->}(9.5,0)(21,0) \rput(15,3){$\yb_2$}
\pscircle(23.5,0){2.5}  \rput{45}(23.5,0){\psline{-}(0,-1.25)(0,1.25) \psline{-}(-1.25,0)(1.25,0)} 
\psline{->}(23.5,8)(23.5,2.5) \rput(23.5,10.5){$\alpha_2$}
\psline{->}(26,0)(32.5,0)
\pscircle(35,0){2.5}  \rput(35,0){\psline{-}(0,-1.25)(0,1.25) \psline{-}(-1.25,0)(1.25,0)} 
\psline{->}(35,7.5)(35,2.5) \rput(42,11){$\sum a_{2\ell}\db_\ell$}
\psline{->}(37.5,0)(49,0) \rput(43,3){$\sb_2$}
\psframe(49,-5)(60,5) \rput(54.5,0){$Q_2$} 
\psline{->}(60,0)(66.5,0)
\psframe(66.5,-5)(86.5,5) \rput(76.5,0){$\modl$}
\psline{->}(86.5,0)(98,0) \rput(92,3.7){$\vbh_2$}
\psframe(98,-5)(112,5) \rput(105,0){$\phi^{-1}$} 
\psline{->}(112,0)(117,0)\rput(120.5,0.3){$\ubh_2$}
}

\rput(76.5,-1){$\vdots$}
\rput(105,-1){$\vdots$}

\rput(0,-18){
\psline{->}(0,0)(4.5,0)
\pscircle(7,0){2.5}  \rput(7,0){\psline{-}(0,-1.25)(0,1.25) \psline{-}(-1.25,0)(1.25,0)} 
\psline{->}(7,8)(7,2.5) \rput(8,10.5){$\zb_M$}
\psline{->}(9.5,0)(21,0) \rput(15,3){$\yb_M$}
\pscircle(23.5,0){2.5}  \rput{45}(23.5,0){\psline{-}(0,-1.25)(0,1.25) \psline{-}(-1.25,0)(1.25,0)} 
\psline{->}(23.5,8)(23.5,2.5) \rput(24.5,10.5){$\alpha_M$}
\psline{->}(26,0)(32.5,0)
\pscircle(35,0){2.5}  \rput(35,0){\psline{-}(0,-1.25)(0,1.25) \psline{-}(-1.25,0)(1.25,0)} 
\psline{->}(35,7.5)(35,2.5) \rput(43,11){$\sum a_{M\ell}\db_\ell$}
\psline{->}(37.5,0)(49,0) \rput(43,3){$\sb_M$}
\psframe(49,-5)(60,5) \rput(54.5,0){$Q_M$} 
\psline{->}(60,0)(66.5,0)
\psframe(66.5,-5)(86.5,5) \rput(76.5,0){$\modl$}
\psline{->}(86.5,0)(98,0) \rput(92,3.7){$\vbh_M$}
\psframe(98,-5)(112,5) \rput(105,0){$\phi^{-1}$} 
\psline{->}(112,0)(117,0)\rput(121.5,0.3){$\ubh_M$}
}
}
\end{pspicture}
\end{center}
\caption{System diagram of the nested lattice encoding and decoding operations employed as part of the compute-and-forward framework (for real-valued channel models). Each message $\wb_\e$ is mapped to a lattice codeword $\tb_\e$, dithered, and transmitted as $\xb_\e$. Each relay observes $\yb_m$ which it scales by $\alpha_m$. It then removes the dithers, quantizes the result onto the fine lattice using $Q_m$, and maps it onto the fundamental Voronoi region of the coarse lattice using $~\modl$. Finally, the relay maps its estimate $\vbh_m$ of the lattice equation $\vb_m =[\sum{a_{m\ell}\tb_{\ell}}]~\modl$ back to the finite field using $\phi^{-1}$ to get an estimate $\ubh_m$ of a linear equation of the messages $\ub_m = \bigoplus{q_{m\e} \wb_\e}$ where $q_{m\ell} = g^{-1}( [a_{m\ell}]~\modp )$ are the finite field representations of the coefficients.  } \label{f:systemdiagram}
\end{figure*}

\begin{theorem}\label{t:latcompreal}
For any $\epsilon > 0$ and $n$ large enough, there exist nested lattice codes $\Lambda \subseteq  \L_L \subseteq \cdots \subseteq \L_1$ with rates $r_1, \ldots, r_L$, such that for all channel vectors $\hb_1, \ldots, \hb_M \in \Rbb^L$ and coefficient vectors $\ab_1, \ldots, \ab_M \in \Zbb^L$, relay $m$ can decode the lattice equation
\begin{align}
\vb_m = \left[\sum_{\ell = 1}^L{a_{m\e}\tb_\e} \right]\modl
\end{align}
of transmitted lattice points $\tb_\ell \in \mathcal{L}_\e$ with average probability of error $\epsilon$ so long as
\begin{align}
r_\ell &< \min_{m:a_{m\e} \neq 0}\frac{1}{2}\log^+\left(\frac{P}{\alpha_m^2 + P \| \alpha_m \mathbf{h}_m - \mathbf{a}_m \|^2}\right) \nonumber 
\end{align}for some choice of $\alpha_1, \ldots, \alpha_M \in \Rbb$.
\end{theorem}
\begin{proof} Each encoder is given a {dither vector} $\mathbf{d}_\ell$ which is generated independently according to a uniform distribution over $\mathcal{V}$. All dither vectors are made available to each relay. Encoder $\ell$ dithers its lattice point, takes $\modl$, and transmits the result:
\begin{align}
\mathbf{x}_\ell = [\mathbf{t}_\ell - \mathbf{d}_\ell]\modl \ .  \label{e:xell}
\end{align} By Lemma \ref{l:dither}, $\xb_\e$ is uniform over $\Vm$ so $E[\|\mathbf{x}_\ell\|^2] = nP$, where the expectation is taken over the dithers. In Appendix \ref{s:fixeddithers}, we argue that there exist fixed dithers that meet the power constraint set forth in (\ref{e:powerconstraint}).

The channel output at relay $m$ is 
\begin{align}
\yb_m = \sum_{\ell = 1}^L{h_{m\e} \xb_\e} + \zb_m \ .
\end{align} Recall that the transmitters are ordered by decreasing message rates. Let $\e_{\text{MAX}}(m) = \max{\{\e: a_{m \e} \neq 0\}}$ denote the highest index value of the non-zero coefficients in $\ab_m$. Also, let $Q_m$ denote the lattice quantizer for the corresponding fine lattice $\L_{\e_{\text{MAX}}(m)}$. Note that this is the highest rate message in the equation and thus the rate of the equation itself. Each relay computes
\begin{align}
\sb_m &= \alpha_m \mathbf{y}_m + \sum_{\ell=1}^L{a_{m \e} \mathbf{d}_\ell}  \ .\label{e:decodedither}
\end{align} To get an estimate of the lattice equation $\vb_m$, this vector is quantized onto $\L_{\e_{\text{MAX}}(m)}$ modulo the coarse lattice $\Lambda$:
\begin{align}
\vbh_m&= \big[Q_m(\sb_m)\big]\modl \ . 
\end{align} Using (\ref{e:modprop2}), we get that
\begin{align}
 \big[Q_m(\sb_m)\big]\modl =  \big[Q_m([\sb_m]\modl)\big]\modl \ .
 \end{align} We now show that $[\sb_m]~\modl$ is equivalent to $\vb_m$ plus some noise terms. Let $\theta_{m\ell} = \alpha_m h_{m \ell} - a_{m \ell}$.
 \begin{align}
&[\sb_m]\modl \\
&=\left[\sum_{\e=1}^L{\Big(\alpha_m h_{m\e}\xb_\e + a_{m\e} \db_\e\Big)} + \alpha_m \zb_m \right] \modl \label{e:latmanip1real}\\
&= \left[\sum_{\e=1}^L \Big(a_{m\e} (\xb_\e + \db_\e) + \theta_{m\e} \xb_\e \Big)+ \alpha_m \zb_m \right] \modl \label{e:latmanip2real}\\
&= \left[\sum_{\e=1}^L a_{m\e} \Big([\tb_\e - \db_\e]\modl + \db_\e\Big)\right.  \nonumber \\
& \qquad ~+ ~\left. \sum_{\e = 1}^L \theta_{m\e} \xb_\e + \alpha_m \zb_m \right] \modl \label{e:latmanip3real}\\
&= \left[\sum_{\e=1}^L a_{m\e} \tb_\e + \sum_{\e = 1}^L \theta_{m\e} \xb_\e + \alpha_m \zb_m \right] \modl \label{e:latmanip4real}\\
&= \left[\vb_m + \sum_{\e = 1}^L \theta_{m\e} \xb_\e + \alpha_m \zb_m \right] \modl \label{e:latmanip5real}
\end{align} where the last two steps are due to (\ref{e:modprop1}). From Lemma \ref{l:dither}, the pair of random variables $(\vb_m, \vbh_m)$ has the same joint distribution as the pair $(\vb_m, \mathbf{\tilde{v}}_m)$ defined by the following: 
\begin{align}
\mathbf{\tilde{v}}_m &= \big[Q_m(\mathbf{v}_m + \mathbf{z}_{eq, m})\big]\modl \\
\mathbf{z}_{eq,m} &= \alpha_m \zb_m + \sum_{\ell=1}^L{\theta_{m\e} \mathbf{\tilde{d}}_\ell}
\end{align} where each $\mathbf{\tilde{d}}_\ell$ is drawn independently according to a uniform distribution over $\mathcal{V}$. See Figure \ref{f:equivchannel} for a block diagram of the equivalent channel. The probability of error $\Pr(\vbh_m \neq \vb_m)$ is thus equal to the probability that the equivalent noise leaves the Voronoi region surrounding the codeword, \mbox{$\Pr\big(\zb_{m,eq} \notin \mathcal{V}_{\ell_{\text{MAX}}(m)}\big)$}.

\begin{figure}[h]
\begin{center}
\psset{unit=0.65mm}
\begin{pspicture}(2,-23)(129,50)

\rput(0,32){
\rput(2,0){$\wb_1$}\psline{->}(6,0)(10,0)
\psframe(10,-5)(20,5) \rput(15,0){$\phi$} 
\psline{->}(20,0)(30,0) \rput(24.5,3.6){$\tb_1$}
}

\rput(0,12){
\rput(2,0){$\wb_2$}\psline{->}(6,0)(10,0)
\psframe(10,-5)(20,5) \rput(15,0){$\phi$} 
\psline{->}(20,0)(30,0) \rput(24.5,3.6){$\tb_2$}
}

\rput(15,-1){$\vdots$}

\rput(0,-18){
\rput(1.5,0){$\wb_L$}\psline{->}(6,0)(10,0)
\psframe(10,-5)(20,5) \rput(15,0){$\phi$} 
\psline{->}(20,0)(30,0) \rput(24.5,3.6){$\tb_L$}
}

\psframe(30,-23)(48,37)
\rput(39,7){\Large{$\mathbf{A}$}}

\rput(42,0){

\rput(0,32){
\psline{->}(6,0)(10,0) 
\pscircle(12.5,0){2.5}  \rput(12.5,0){\psline{-}(0,-1.25)(0,1.25) \psline{-}(-1.25,0)(1.25,0)} 
\psline{->}(12.5,8)(12.5,2.5) \rput(14,10.5){$\zb_{eq,1}$}
\psline{->}(15,0)(20,0)
\rput(-29,0){
\psframe(49,-5)(60,5) \rput(54.5,0){$Q_1$} 
\psline{->}(60,0)(65,0)
\psframe(65,-5)(85,5) \rput(75,0){$\modl$}
\psline{->}(85,0)(95,0) \rput(89.5,4){$\mathbf{\tilde{v}}_1$}
\psframe(95,-5)(107,5) \rput(101,0){$\phi^{-1}$} 
\psline{->}(107,0)(111,0)\rput(114.5,0.3){$\ubh_1$}
}
}

\rput(0,12){
\psline{->}(6,0)(10,0) 
\pscircle(12.5,0){2.5}  \rput(12.5,0){\psline{-}(0,-1.25)(0,1.25) \psline{-}(-1.25,0)(1.25,0)} 
\psline{->}(12.5,8)(12.5,2.5) \rput(14,10.5){$\zb_{eq,2}$}
\psline{->}(15,0)(20,0)
\rput(-29,0){
\psframe(49,-5)(60,5) \rput(54.5,0){$Q_2$} 
\psline{->}(60,0)(65,0)
\psframe(65,-5)(85,5) \rput(75,0){$\modl$}
\psline{->}(85,0)(95,0) \rput(89.5,4){$\mathbf{\tilde{v}}_2$}
\psframe(95,-5)(107,5) \rput(101,0){$\phi^{-1}$} 
\psline{->}(107,0)(111,0)\rput(114.5,0.3){$\ubh_2$}
}
}

\rput(46,-1){$\vdots$}
\rput(72,-1){$\vdots$}

\rput(0,-18){
\psline{->}(6,0)(10,0) 
\pscircle(12.5,0){2.5}  \rput(12.5,0){\psline{-}(0,-1.25)(0,1.25) \psline{-}(-1.25,0)(1.25,0)} 
\psline{->}(12.5,8)(12.5,2.5) \rput(15,10.5){$\zb_{eq,M}$}
\psline{->}(15,0)(20,0)
\rput(-29,0){
\psframe(49,-5)(60,5) \rput(54.8,0){$Q_M$} 
\psline{->}(60,0)(65,0)
\psframe(65,-5)(85,5) \rput(75,0){$\modl$}
\psline{->}(85,0)(95,0) \rput(89.8,4){$\mathbf{\tilde{v}}_M$}
\psframe(95,-5)(107,5) \rput(101,0){$\phi^{-1}$} 
\psline{->}(107,0)(111,0)\rput(115.5,0.3){$\ubh_M$}
}
}
}
\end{pspicture}
\end{center}
\caption{Equivalent channel induced by the modulo-$\L$ transformation. In this ``virtual'' channel model, each encoder maps its message $\wb_\e$ to a lattice point $\tb_\e$. Each relay observes an integer combination $\sum{a_{m\ell} \tb_\e}$ of the lattice points corrupted by effective noise $\zb_{eq,m}$. It then quantizes onto the fine lattice using $Q_m$ and takes $~\modl$ to get an estimate $\mathbf{\tilde{v}}_m$ of the lattice equation $\vb_m =[\sum{a_{m\ell}\tb_{\ell}}]~\modl$. Finally, the relay maps the recovered lattice equation to an estimate $\ubh_m$ of its desired linear equation of the messages $\ub_m = \bigoplus{q_{m\e} \wb_\e}$ where $q_{m\ell} = g^{-1}( [a_{m\ell}]~\modp )$ are the finite field representations of the coefficients. } \label{f:equivchannel}
\end{figure}
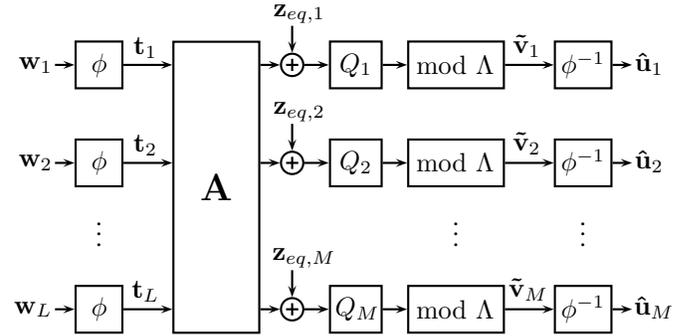

Using Lemma \ref{l:gaussiannoise} from Appendix \ref{s:gaussiannoise}, the density of $\zb_{eq,m}$ can be upper bounded (times a constant) by the density of an i.i.d. zero-mean Gaussian vector $\zb^*_m$ whose variance $\sigma_m^2$ approaches
\begin{align}
N_{eq,m} &= \alpha_m^2 + P \sum_{\e=1}^L \theta_{m\e}^2 \\
&= \alpha_m^2 + P \| \alpha_m \hb_m - \ab_m \|^2 \label{e:complexstart}
\end{align} as $n \rightarrow \infty$. 
We also show in Appendix \ref{s:awgngood} that $\L_1, \L_2, \ldots, \L_L$ are good for AWGN. From Definition \ref{d:awgngood}, this means that $\epsilon_m = \Pr(\zb^*_m \notin \Vm_{\e_{\text{MAX}}(m)})$ goes to zero exponentially in $n$ so long as the volume-to-noise ratio satisfies $\mu(\Lambda_{\e_{\text{MAX}}(m)},\epsilon_m) > 2 \pi e$. If this occurs, then $\Pr(\zb_{eq,m} \notin \Vm_{\e_{\text{MAX}}(m)})$ goes to zero exponentially in $n$ as well. Note that, by the union bound, the average probability of error $\epsilon$ is upper bounded by the sum
\begin{align}
\epsilon \leq \sum_{m=1}^M \Pr(\zb_{eq,m} \notin \Vm_{\e_{\text{MAX}}(m)})  \ . 
\end{align}

To ensure that the probability of error goes to zero for all desired equations\footnote{Note that by Lemma \ref{l:coeffbound} the number of available coefficient vectors $\ab_m$ at each relay is finite if $\| \hb_m \|$ and $P$ are finite. Therefore, it can be shown via a union bound that each relay can decode more than one equation.}, we get that the volume of  $\Vm_{\e_{\text{MAX}}(m)}$ must satisfy
\begin{align}
2 \pi e < \mu(\L_{\e_{\text{MAX}}(m)}, \epsilon_m) =\frac{(\mathsf{Vol}(\Vm_{\e_{\text{MAX}}(m)}))^{2/n}}{\sigma_m^2}
\end{align} for all relays with $a_{m\e} \neq 0$. If we set the volume of each Voronoi region $\Vm_{\e}$ as follows, the constraints are always met:
\begin{align}
\mathsf{Vol}(\Vm_\e) > \left(2 \pi e \max_{m:a_{m\e}\neq 0}{\sigma_m^2}\right)^{n/2}
\end{align} 

Recall that the rate of a nested lattice code is 
\begin{align}
r_\e = \onen \log{\left(\frac{\mathsf{Vol}(\Vm)}{\mathsf{Vol}(\Vm_\e)}\right)}.
\end{align} Using (\ref{e:normsecmoment}), we can solve for the volume of the fundamental Voronoi region of the coarse lattice:
\begin{align}
\mathsf{Vol}(\Vm) = \left(\frac{P}{G(\Lambda)}\right)^{n/2} \label{e:volv}
\end{align} 

It follows that we can achieve any rate less than
\begin{align}
r_\e < \min_{m:a_{m\e} \neq 0}{\frac{1}{2} \log^+{\left(\frac{P}{G(\L) 2 \pi e   \sigma_m^2}\right)}}
\end{align} Choose $\delta > 0$. Since $\Lambda$ is good for quantization, for $n$ large enough, we have that $G(\Lambda) 2 \pi e < (1 + \delta)$. We also know that $\sigma_m^2$ converges to $N_{eq,m}$ so for $n$ large enough we have $\sigma_m^2 < (1 + \delta) N_{eq,m}$. Finally, we get that the rate $r_\ell$ of each nested lattice code is at least
\begin{align}
\min_{m: a_{m \e} \neq 0} {\frac{1}{2} \log^+\left(\frac{P}{\alpha_m^2 + P\|\alpha_m \hb_{m} - \ab_m\|^2}\right) - \log(1 + \delta)} \nonumber
\end{align} Thus, by choosing $\delta$ small enough, we can approach the computation rates as closely as desired. 
\end{proof}


We now put all of these ingredients together to prove Theorem \ref{t:basiccompreal}. See Figures \ref{f:encoder} and \ref{f:decoder} for block diagrams of the encoding and decoding process.

\textit{Proof of Theorem \ref{t:basiccompreal}:} See Figure \ref{f:systemdiagram} for a block diagram. Choose $\epsilon > 0$. Encoder $\e$ maps its finite field message vector $\wb_\ell$ to a lattice point $\mathbf{t}_\ell  \in \Lambda_\ell \cap \mathcal{V}$, using $\phi$ from Lemma \ref{l:fieldtolattice},
\begin{align}
\mathbf{t}_\ell &= \phi(\wb_\ell) \ . 
\end{align} Using Theorem \ref{t:latcompreal}, these lattice points can be transmitted across the channel so that the relays can make estimates $\vbh_m$ of lattice equations $\vb_m$ with coefficient vectors $\ab_m \in \Zbb^L$ such that $\Pr\left( \cup_m \{ \vbh_m \neq  \vb_m \} \right) < \epsilon$ for $n$ large enough so long as
\begin{align}
R_\ell <& \min_{m: a_{ml} \neq 0} \frac{1}{2} \log^+\left(\frac{P}{\alpha_m^2 + P \| \alpha_m \mathbf{h}_m - \mathbf{a}_m \|^2}\right)  \nonumber
\end{align} for some $\alpha_1, \ldots, \alpha_M \in \Rbb$. 
Finally, using $\phi\n$ from Lemma \ref{l:msgeqn}, each relay can produce estimates of the desired linear combination of messages, $\ubh_m = \phi\n(\vbh_m)$, such that \mbox{$\Pr\big(\cup_m \{ \ubh_m \neq \ub_m \}\big) < \epsilon$} where
\begin{align}
\ub_m &=\bigoplus_{\ell = 1}^L{q_{m \ell }\wb_\ell} \\
q_{m\ell} &= g\n( [a_{m\e} ] \modp ) \ . 
\end{align} 
\subsection{Complex-Valued Channel Models} \label{s:complexcompproof}

We now show how to use nested lattice codes over complex-valued channel models.

\begin{theorem}\label{t:latcompcomplex}
For any $\epsilon > 0$ and $n$ large enough, there exist nested lattice codes $\Lambda \subseteq  \L_1 \subseteq \cdots \subseteq \L_L$ with rates $R_1, \ldots, R_L$, such that for all channel vectors $\hb_1, \ldots, \hb_M \in \Rbb^L$ and coefficient vectors $\ab_1, \ldots, \ab_M \in \{\Zbb + j \Zbb \}^L$, each relay can decode lattice equations $\vb_m^R, \vb_m^I$ where 
\begin{align}
\vb_m^R = \left[\sum_{\ell = 1}^L{\Re(a_{m\e})\tb_\e^R - \Im(a_{m\e})\tb_\e^I} \right]\modl \\
\vb_m^I = \left[\sum_{\ell = 1}^L{\Im(a_{m\e})\tb_\e^R + \Re(a_{m\e})\tb_\e^I} \right]\modl
\end{align}
of transmitted lattice points $\tb_\ell^R, \tb_\e^I \in \L_\ell \cap \Vm$ with average probability of error $\epsilon$ so long as
\begin{align}
r_\ell &<\frac{1}{2}\log^+\left(\frac{P}{|\alpha_m|^2 + P \| \alpha_m \mathbf{h}_m - \mathbf{a}_m \|^2}\right) \label{e:latcomp}
\end{align}for some choice of $\alpha_1, \ldots, \alpha_M \in \Cbb$.
\end{theorem}

\begin{proof}
First, we scale our nested lattice ensemble so that the coarse lattice $\Lambda$ has second moment $P/2$. Each encoder is given two {dither vectors}, $\mathbf{d}_\ell^R$ and $\mathbf{d}_\ell^I$, which are independently drawn according to a uniform distribution over $\mathcal{V}$. All dither vectors are made available to each relay. Encoder $\ell$ generates a channel input:
\begin{align}
\mathbf{x}_\ell = \big[\mathbf{t}^R_\ell - \mathbf{d}^R_\ell\big]\modl + j \big[\mathbf{t}^I_\ell - \mathbf{d}^I_\ell\big]\modl \ . \label{e:xell}
\end{align} By Lemma \ref{l:dither}, the real and imaginary parts of $\xb_\e$ are independent and uniform over $\Vm$ so $E[\|\mathbf{x}_\ell\|^2] = nP$, with expectation taken over the dithers.\footnote{In Appendix \ref{s:fixeddithers}, we argue that there exist fixed dithers that meet the power constraint $\| \mathbf{x}\|^2 \leq nP$.}

Let $\e_{\text{MAX}}(m) = \max{\{\e: a_{m \e} \neq 0\}}$ and let $Q_m$ denote the lattice quantizer for $\L_{\e_{\text{MAX}}(m)}$. Each relay computes
\begin{align}
\sb_m^R &= \Re(\alpha_m \mathbf{y}_m) + \sum_{\ell=1}^L{\Re(a_{m \e}) \mathbf{d}^R_\ell- \Im(a_{m \e}) \mathbf{d}^I_\ell} \label{e:decodedither1}\\
\sb_m^I &= \Im(\alpha_m \mathbf{y}_m) + \sum_{\ell=1}^L{\Im(a_{m \e}) \mathbf{d}^R_\ell+ \Re(a_{m \e}) \mathbf{d}^I_\ell}. \label{e:decodedither2}
\end{align} To get estimates of the lattice equations, these vectors are quantized onto $\Lambda_{\e_{\text{MAX}}(m)}$ modulo the coarse lattice $\Lambda$:
\begin{align}
\vbh_m^R&= \left[Q_m(\sb_m^R)\right]\modl\\
\vbh_m^I&= \left[Q_m(\sb_m^I)\right]\modl.
\end{align} Note that by (\ref{e:modprop2}) we have
\begin{align}
 \Big[Q_m\big(\sb_m^R\big)\Big]\modl =  \Big[Q_m\big(\big[\sb_m^R\big]\modl\big)\Big]\modl.
 \end{align} 
 \begin{figure*}[!t]
\begin{align}
&[\sb_m^R]\modl =\left[\sum_{\e=1}^L{\Big(\Re(\alpha_m h_{m\e}) \Re(\xb_\e) - \Im(\alpha_m h_{m\e})\Im(\xb_\e) + \Re(a_{m\e}) \db_\e^R - \Im(a_{m\e}) \db_\e^I}\Big) + \Re(\alpha_m \zb_m) \right] \modl \label{e:latmanip1}\\
&= \left[\sum_{\e=1}^L{\Big( \Re(a_{m\e}) (\Re(\xb_\e) + \db_\e^R) - \Im(a_{m\e}) (\Im(\xb_\e) + \db_\e^I}) + \theta_{m\e}^R \Re(\xb_\e) - \theta_{m\e}^I \Im(\xb_\e) \Big)+ \Re(\alpha_m \zb_m) \right] \modl \label{e:latmanip2}\\
&\overset{\mbox{\footnotesize{(a)}}}{=} \left[\sum_{\e=1}^L{\Big( \Re(a_{m\e}) \tb^R_\e - \Im(a_{m\e}) \tb^I_\e } + \theta_{m\e}^R \Re(\xb_\e) - \theta_{m\e}^I \Im(\xb_\e) \Big) + \Re(\alpha_m \zb_m) \right] \modl \label{e:latmanip3} \\
&\overset{\mbox{\footnotesize{(b)}}}{=} \left[ \vb_m^R + \sum_{\e=1}^L{ \Big( \theta_{m\e}^R \Re(\xb_\e) - \theta_{m\e}^I \Im(\xb_\e) \Big)+ \Re(\alpha_m \zb_m)}\right] \modl \label{e:latmanip4}
\end{align}
\hrulefill
\vspace*{4pt}
\end{figure*}
 
 Define $\theta_{m\e}^R = \Re(\alpha_m h_{m\ell} - a_{m\ell})$ and 
$\theta_{m\e}^I = \Im(\alpha_m h_{m\ell} - a_{m\ell})$.  We now show that $[\sb_m^R]~\modl$ is equivalent to $\vb_m^R$ plus some noise terms in (\ref{e:latmanip1})-(\ref{e:latmanip4}). Using similar manipulations, it can be shown that $[\sb_m^I]~\modl$ is equivalent to $\vb_m^I$ plus some noise terms as well. From Lemma \ref{l:dither}, the pairs of random variables $(\vb_m^R, \vbh_m^R)$ and $(\vb_m^I, \vbh_m^I)$ have the same joint distributions as the pairs $(\vb_m^R, \mathbf{\tilde{v}}_m^R)$ and $(\vb_m^I, \mathbf{\tilde{v}}_m^I)$, respectively, where 
\begin{align}
\mathbf{\tilde{v}}_m^R &= \big[Q_m(\mathbf{v}_m^R + \mathbf{z}^R_{eq, m})\big]\modl \\
\mathbf{\tilde{v}}_m^I &= \big[Q_m(\mathbf{v}_m^I + \mathbf{z}^I_{eq, m})\big]\modl \\
\mathbf{z}^R_{eq,m} &= \Re(\alpha_m \zb_m) + \sum_{\ell=1}^L{\theta_{m\ell}^R\mathbf{\tilde{d}}_\ell^R - \theta_{m\e}^I \mathbf{\tilde{d}}_\ell^I} \\
\mathbf{z}^I_{eq,m} &= \Im(\alpha_m \zb_m) + \sum_{\ell=1}^L{\theta_{m\e}^I \mathbf{\tilde{d}}_\ell^R + \theta_{m\e}^R \mathbf{\tilde{d}}_\ell^I}
\end{align} where each $\mathbf{\tilde{d}}_\ell^R$ and $\mathbf{\tilde{d}}_\ell^I$ is drawn independently according to a uniform distribution over $\mathcal{V}$. 
\begin{figure}[t]
\begin{center}
\psset{unit=0.6mm}
\begin{pspicture}(0,0)(110,60)

\rput(10,0){
\rput(-10,5){$\wb_\e^I$}\psline{->}(-5,5)(0,5)
\psframe(0,0)(10,10) \rput(5,5){$\phi$}
\psline{->}(10,5)(40,5)
\rput(25,16){\scalebox{0.45}{

\psline[linewidth=2pt,linecolor=black!50!white](-26.667,0)(-13.333,20)(13.333,20)(26.667,0)(13.333,-20)(-13.333,-20)(-26.667,0)

\rput(-20,0){
\psline(-3.333,5)(3.333,5)(6.667,0)(3.333,-5)(-3.333,-5)(-6.667,0)(-3.333,5)
\pscircle[fillstyle=solid,fillcolor=black](0,0){1}
}
\rput(-10,15){
\psline(-3.333,5)(3.333,5)(6.667,0)(3.333,-5)(-3.333,-5)(-6.667,0)(-3.333,5)
\pscircle[fillstyle=solid,fillcolor=black](0,0){1}
}
\rput(-10,5){
\psline(-3.333,5)(3.333,5)(6.667,0)(3.333,-5)(-3.333,-5)(-6.667,0)(-3.333,5)
\pscircle[fillstyle=solid,fillcolor=black](0,0){1}
}
\rput(-10,-5){
\psline(-3.333,5)(3.333,5)(6.667,0)(3.333,-5)(-3.333,-5)(-6.667,0)(-3.333,5)
\pscircle[fillstyle=solid,fillcolor=black](0,0){1}
}
\rput(-10,-15){
\psline(-3.333,5)(3.333,5)(6.667,0)(3.333,-5)(-3.333,-5)(-6.667,0)(-3.333,5)
\pscircle[fillstyle=solid,fillcolor=black](0,0){1}
}
\rput(20,0){
\psline(-3.333,5)(3.333,5)(6.667,0)(3.333,-5)(-3.333,-5)(-6.667,0)(-3.333,5)
\pscircle[fillstyle=solid,fillcolor=black](0,0){1}
}
\rput(0,10){
\psline(-3.333,5)(3.333,5)(6.667,0)(3.333,-5)(-3.333,-5)(-6.667,0)(-3.333,5)
\pscircle[fillstyle=solid,fillcolor=black](0,0){1}
}
\rput(0,0){
\psline(-3.333,5)(3.333,5)(6.667,0)(3.333,-5)(-3.333,-5)(-6.667,0)(-3.333,5)
\pscircle[fillstyle=solid,fillcolor=black](0,0){1}
}
\rput(0,-10){
\psline(-3.333,5)(3.333,5)(6.667,0)(3.333,-5)(-3.333,-5)(-6.667,0)(-3.333,5)
\pscircle[fillstyle=solid,fillcolor=black](0,0){1}
}
\rput(10,15){
\psline(-3.333,5)(3.333,5)(6.667,0)(3.333,-5)(-3.333,-5)(-6.667,0)(-3.333,5)
\pscircle[fillstyle=solid,fillcolor=black](0,0){1}
}
\rput(10,5){
\psline(-3.333,5)(3.333,5)(6.667,0)(3.333,-5)(-3.333,-5)(-6.667,0)(-3.333,5)
\pscircle[fillstyle=solid,fillcolor=black](0,0){1}
}
\rput(10,-5){
\psline(-3.333,5)(3.333,5)(6.667,0)(3.333,-5)(-3.333,-5)(-6.667,0)(-3.333,5)
\pscircle[fillstyle=solid,fillcolor=black](0,0){1}
}
\rput(10,-15){
\psline(-3.333,5)(3.333,5)(6.667,0)(3.333,-5)(-3.333,-5)(-6.667,0)(-3.333,5)
\pscircle[fillstyle=solid,fillcolor=black](0,0){1}
}

\pscircle[fillstyle=solid,fillcolor=black](0,20){1}
\pscircle[fillstyle=solid,fillcolor=black](20,10){1}
\pscircle[fillstyle=solid,fillcolor=black](20,-10){1}

}}
\psframe(40,0)(60,10) \rput(50,5){Dither}
}

\rput(10,30){
\rput(-10,5){$\wb_\e^R$}\psline{->}(-5,5)(0,5)
\psframe(0,0)(10,10) \rput(5,5){$\phi$}
\psline{->}(10,5)(40,5)
\rput(25,16){\scalebox{0.45}{

\psline[linewidth=2pt,linecolor=black!50!white](-26.667,0)(-13.333,20)(13.333,20)(26.667,0)(13.333,-20)(-13.333,-20)(-26.667,0)

\rput(-20,0){
\psline(-3.333,5)(3.333,5)(6.667,0)(3.333,-5)(-3.333,-5)(-6.667,0)(-3.333,5)
\pscircle[fillstyle=solid,fillcolor=black](0,0){1}
}
\rput(-10,15){
\psline(-3.333,5)(3.333,5)(6.667,0)(3.333,-5)(-3.333,-5)(-6.667,0)(-3.333,5)
\pscircle[fillstyle=solid,fillcolor=black](0,0){1}
}
\rput(-10,5){
\psline(-3.333,5)(3.333,5)(6.667,0)(3.333,-5)(-3.333,-5)(-6.667,0)(-3.333,5)
\pscircle[fillstyle=solid,fillcolor=black](0,0){1}
}
\rput(-10,-5){
\psline(-3.333,5)(3.333,5)(6.667,0)(3.333,-5)(-3.333,-5)(-6.667,0)(-3.333,5)
\pscircle[fillstyle=solid,fillcolor=black](0,0){1}
}
\rput(-10,-15){
\psline(-3.333,5)(3.333,5)(6.667,0)(3.333,-5)(-3.333,-5)(-6.667,0)(-3.333,5)
\pscircle[fillstyle=solid,fillcolor=black](0,0){1}
}
\rput(20,0){
\psline(-3.333,5)(3.333,5)(6.667,0)(3.333,-5)(-3.333,-5)(-6.667,0)(-3.333,5)
\pscircle[fillstyle=solid,fillcolor=black](0,0){1}
}
\rput(0,10){
\psline(-3.333,5)(3.333,5)(6.667,0)(3.333,-5)(-3.333,-5)(-6.667,0)(-3.333,5)
\pscircle[fillstyle=solid,fillcolor=black](0,0){1}
}
\rput(0,0){
\psline(-3.333,5)(3.333,5)(6.667,0)(3.333,-5)(-3.333,-5)(-6.667,0)(-3.333,5)
\pscircle[fillstyle=solid,fillcolor=black](0,0){1}
}
\rput(0,-10){
\psline(-3.333,5)(3.333,5)(6.667,0)(3.333,-5)(-3.333,-5)(-6.667,0)(-3.333,5)
\pscircle[fillstyle=solid,fillcolor=black](0,0){1}
}
\rput(10,15){
\psline(-3.333,5)(3.333,5)(6.667,0)(3.333,-5)(-3.333,-5)(-6.667,0)(-3.333,5)
\pscircle[fillstyle=solid,fillcolor=black](0,0){1}
}
\rput(10,5){
\psline(-3.333,5)(3.333,5)(6.667,0)(3.333,-5)(-3.333,-5)(-6.667,0)(-3.333,5)
\pscircle[fillstyle=solid,fillcolor=black](0,0){1}
}
\rput(10,-5){
\psline(-3.333,5)(3.333,5)(6.667,0)(3.333,-5)(-3.333,-5)(-6.667,0)(-3.333,5)
\pscircle[fillstyle=solid,fillcolor=black](0,0){1}
}
\rput(10,-15){
\psline(-3.333,5)(3.333,5)(6.667,0)(3.333,-5)(-3.333,-5)(-6.667,0)(-3.333,5)
\pscircle[fillstyle=solid,fillcolor=black](0,0){1}
}

\pscircle[fillstyle=solid,fillcolor=black](0,20){1}
\pscircle[fillstyle=solid,fillcolor=black](20,10){1}
\pscircle[fillstyle=solid,fillcolor=black](20,-10){1}

}}
\psframe(40,0)(60,10) \rput(50,5){Dither}
}

\psline{->}(70,35)(85,35)(95,22)
\psline{->}(70,5)(76.5,5)
\psline{->}(79,14.5)(79,7.5) \rput(79,18){$j$}
\pscircle(79,5){2.5} \psline(78,4)(80,6) \psline(80,4)(78,6)
\psline{->}(81.5,5)(85,5)(95,18)
\pscircle(96.5,20){2.5} \psline(95.25,20)(97.75,20) \psline(96.5,18.75)(96.5,21.25)
\psline{->}(99,20)(105,20)
\rput(109,19.5){$\xb_\e$}
\end{pspicture}
\end{center}
\caption{Block diagram of the complex-valued compute-and-forward encoder at transmitter $\ell$, $\Em_\e$. Messages from a finite field are mapped onto a nested lattice code, dithered, and transmitted across the channel.} \label{f:encoder}
\end{figure}
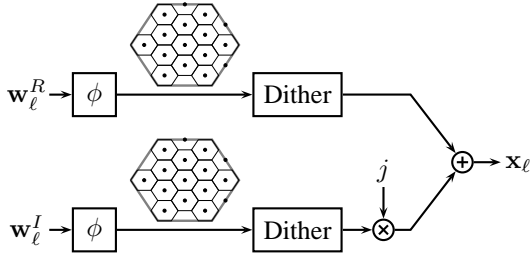
\begin{figure}[t]
\begin{center}
\psset{unit=0.58mm}
\begin{pspicture}(-15,0)(126,60)

\rput(12,0){
\psframe(-6,0)(4,10) \rput(-1,5){$\Im$}
\psline{->}(4,5)(9,5)
\psframe(9,-3)(27,13) \rput(18,8){\mbox{\scriptsize{Remove}}} \rput(18,2){\mbox{\scriptsize{Dithers}}}
\psline{->}(27,5)(32,5)
\psframe(32,0)(73,10) \rput(52.5,5){$[Q_m(~)]\hspace{0.02in}\modl$}
\psline{->}(73,5)(94,5)
\rput(83.5,15.5){\scalebox{0.36}{

\psline[linewidth=2pt,linecolor=black!50!white](-26.667,0)(-13.333,20)(13.333,20)(26.667,0)(13.333,-20)(-13.333,-20)(-26.667,0)

\rput(-20,0){
\psline(-3.333,5)(3.333,5)(6.667,0)(3.333,-5)(-3.333,-5)(-6.667,0)(-3.333,5)
\pscircle[fillstyle=solid,fillcolor=black](0,0){1}
}
\rput(-10,15){
\psline(-3.333,5)(3.333,5)(6.667,0)(3.333,-5)(-3.333,-5)(-6.667,0)(-3.333,5)
\pscircle[fillstyle=solid,fillcolor=black](0,0){1}
}
\rput(-10,5){
\psline(-3.333,5)(3.333,5)(6.667,0)(3.333,-5)(-3.333,-5)(-6.667,0)(-3.333,5)
\pscircle[fillstyle=solid,fillcolor=black](0,0){1}
}
\rput(-10,-5){
\psline(-3.333,5)(3.333,5)(6.667,0)(3.333,-5)(-3.333,-5)(-6.667,0)(-3.333,5)
\pscircle[fillstyle=solid,fillcolor=black](0,0){1}
}
\rput(-10,-15){
\psline(-3.333,5)(3.333,5)(6.667,0)(3.333,-5)(-3.333,-5)(-6.667,0)(-3.333,5)
\pscircle[fillstyle=solid,fillcolor=black](0,0){1}
}
\rput(20,0){
\psline(-3.333,5)(3.333,5)(6.667,0)(3.333,-5)(-3.333,-5)(-6.667,0)(-3.333,5)
\pscircle[fillstyle=solid,fillcolor=black](0,0){1}
}
\rput(0,10){
\psline(-3.333,5)(3.333,5)(6.667,0)(3.333,-5)(-3.333,-5)(-6.667,0)(-3.333,5)
\pscircle[fillstyle=solid,fillcolor=black](0,0){1}
}
\rput(0,0){
\psline(-3.333,5)(3.333,5)(6.667,0)(3.333,-5)(-3.333,-5)(-6.667,0)(-3.333,5)
\pscircle[fillstyle=solid,fillcolor=black](0,0){1}
}
\rput(0,-10){
\psline(-3.333,5)(3.333,5)(6.667,0)(3.333,-5)(-3.333,-5)(-6.667,0)(-3.333,5)
\pscircle[fillstyle=solid,fillcolor=black](0,0){1}
}
\rput(10,15){
\psline(-3.333,5)(3.333,5)(6.667,0)(3.333,-5)(-3.333,-5)(-6.667,0)(-3.333,5)
\pscircle[fillstyle=solid,fillcolor=black](0,0){1}
}
\rput(10,5){
\psline(-3.333,5)(3.333,5)(6.667,0)(3.333,-5)(-3.333,-5)(-6.667,0)(-3.333,5)
\pscircle[fillstyle=solid,fillcolor=black](0,0){1}
}
\rput(10,-5){
\psline(-3.333,5)(3.333,5)(6.667,0)(3.333,-5)(-3.333,-5)(-6.667,0)(-3.333,5)
\pscircle[fillstyle=solid,fillcolor=black](0,0){1}
}
\rput(10,-15){
\psline(-3.333,5)(3.333,5)(6.667,0)(3.333,-5)(-3.333,-5)(-6.667,0)(-3.333,5)
\pscircle[fillstyle=solid,fillcolor=black](0,0){1}
}

\pscircle[fillstyle=solid,fillcolor=black](0,20){1}
\pscircle[fillstyle=solid,fillcolor=black](20,10){1}
\pscircle[fillstyle=solid,fillcolor=black](20,-10){1}

}}
\psframe(94,0)(106,10) \rput(100,5){$\phi\n$}
\psline{->}(106,5)(111,5) \rput(115.5,5.5){$\ubh_m^I$}
}

\rput(12,30){
\psframe(-6,0)(4,10) \rput(-1,5){$\Re$}
\psline{->}(4,5)(9,5)
\psframe(9,-3)(27,13) \rput(18,8){\mbox{\scriptsize{Remove}}} \rput(18,2){\mbox{\scriptsize{Dithers}}}
\psline{->}(27,5)(32,5)
\psframe(32,0)(73,10) \rput(52.5,5){$[Q_m(~)]\hspace{0.02in}\modl$}
\psline{->}(73,5)(94,5)
\rput(83.5,15.5){\scalebox{0.36}{

\psline[linewidth=2pt,linecolor=black!50!white](-26.667,0)(-13.333,20)(13.333,20)(26.667,0)(13.333,-20)(-13.333,-20)(-26.667,0)

\rput(-20,0){
\psline(-3.333,5)(3.333,5)(6.667,0)(3.333,-5)(-3.333,-5)(-6.667,0)(-3.333,5)
\pscircle[fillstyle=solid,fillcolor=black](0,0){1}
}
\rput(-10,15){
\psline(-3.333,5)(3.333,5)(6.667,0)(3.333,-5)(-3.333,-5)(-6.667,0)(-3.333,5)
\pscircle[fillstyle=solid,fillcolor=black](0,0){1}
}
\rput(-10,5){
\psline(-3.333,5)(3.333,5)(6.667,0)(3.333,-5)(-3.333,-5)(-6.667,0)(-3.333,5)
\pscircle[fillstyle=solid,fillcolor=black](0,0){1}
}
\rput(-10,-5){
\psline(-3.333,5)(3.333,5)(6.667,0)(3.333,-5)(-3.333,-5)(-6.667,0)(-3.333,5)
\pscircle[fillstyle=solid,fillcolor=black](0,0){1}
}
\rput(-10,-15){
\psline(-3.333,5)(3.333,5)(6.667,0)(3.333,-5)(-3.333,-5)(-6.667,0)(-3.333,5)
\pscircle[fillstyle=solid,fillcolor=black](0,0){1}
}
\rput(20,0){
\psline(-3.333,5)(3.333,5)(6.667,0)(3.333,-5)(-3.333,-5)(-6.667,0)(-3.333,5)
\pscircle[fillstyle=solid,fillcolor=black](0,0){1}
}
\rput(0,10){
\psline(-3.333,5)(3.333,5)(6.667,0)(3.333,-5)(-3.333,-5)(-6.667,0)(-3.333,5)
\pscircle[fillstyle=solid,fillcolor=black](0,0){1}
}
\rput(0,0){
\psline(-3.333,5)(3.333,5)(6.667,0)(3.333,-5)(-3.333,-5)(-6.667,0)(-3.333,5)
\pscircle[fillstyle=solid,fillcolor=black](0,0){1}
}
\rput(0,-10){
\psline(-3.333,5)(3.333,5)(6.667,0)(3.333,-5)(-3.333,-5)(-6.667,0)(-3.333,5)
\pscircle[fillstyle=solid,fillcolor=black](0,0){1}
}
\rput(10,15){
\psline(-3.333,5)(3.333,5)(6.667,0)(3.333,-5)(-3.333,-5)(-6.667,0)(-3.333,5)
\pscircle[fillstyle=solid,fillcolor=black](0,0){1}
}
\rput(10,5){
\psline(-3.333,5)(3.333,5)(6.667,0)(3.333,-5)(-3.333,-5)(-6.667,0)(-3.333,5)
\pscircle[fillstyle=solid,fillcolor=black](0,0){1}
}
\rput(10,-5){
\psline(-3.333,5)(3.333,5)(6.667,0)(3.333,-5)(-3.333,-5)(-6.667,0)(-3.333,5)
\pscircle[fillstyle=solid,fillcolor=black](0,0){1}
}
\rput(10,-15){
\psline(-3.333,5)(3.333,5)(6.667,0)(3.333,-5)(-3.333,-5)(-6.667,0)(-3.333,5)
\pscircle[fillstyle=solid,fillcolor=black](0,0){1}
}

\pscircle[fillstyle=solid,fillcolor=black](0,20){1}
\pscircle[fillstyle=solid,fillcolor=black](20,10){1}
\pscircle[fillstyle=solid,fillcolor=black](20,-10){1}

}}
\psframe(94,0)(106,10) \rput(100,5){$\phi\n$}
\psline{->}(106,5)(111,5) \rput(115.5,5.5){$\ubh_m^R$}
}
%
\psline{->}(-6.5,33)(-6.5,22.5) \rput(-5.75,36){$\alpha_m$}
\rput(-16,18){$\yb_m$} \psline{->}(-14,20)(-9,20)
\pscircle(-6.5,20){2.5} \psline(-7.5,19)(-5.5,21) \psline(-5.5,19)(-7.5,21)
\psline{->}(-5,21.5)(1,35)(6,35)
\psline{->}(-5,18.5)(1,5)(6,5)
\end{pspicture}
\end{center}
\caption{Block diagram of the complex-valued compute-and-forward decoder for relay $m$, $\Dm_m$. The channel observation is scaled and decomposed into its real and imaginary components. The decoder then removes the dithers, quantizes onto the appropriate fine lattice, and takes the modulus over the coarse lattice. This results in an equation of lattice codewords which is then mapped into an equation of messages over the finite field.} \label{f:decoder}
\end{figure}
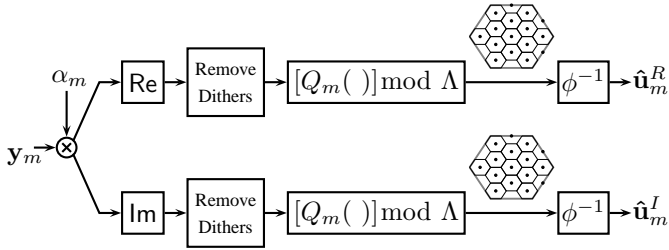
Using Lemma \ref{l:gaussiannoise} from Appendix \ref{s:gaussiannoise}, we have that the densities of both $\zb_{eq,m}^R$ and $\zb_{eq,m}^I$ are upper bounded (times a constant) by the density of an i.i.d. zero-mean Gaussian vector $\zb^*_m$ whose variance $\sigma_m^2$ approaches
\begin{align}
N_{eq,m} &= \frac{|\alpha_m|^2}{2} + \frac{P}{2} \Big( \big(\theta_{m\e}^R\big)^2 + \big(\theta_{m \e}^I\big)^2\Big)\\
&= \frac{|\alpha_m|^2}{2} + \frac{P}{2}\| \alpha_m \hb_m - \ab_m \|^2 
\end{align} as $n \rightarrow \infty$. Note that the effective $\snr$ for both real and imaginary components is
$P/(|\alpha_m|^2 + P \| \alpha_m \hb_m - \ab_m \|^2)$ since the second moment of $\Lambda$ is $P/2$.
This is the same effective $\snr$ encountered in the proof of Theorem \ref{t:latcompreal} and the rest of the proof follows identically from (\ref{e:complexstart}) onwards.
\end{proof}

\textit{Proof of Theorem \ref{t:basiccompcomplex}:}
See Figures \ref{f:encoder} and \ref{f:decoder} for block diagrams of the encoding and decoding processes. Choose $\epsilon > 0$. Encoder $\e$ maps its finite field message vectors $\wb_\ell^R$ and $\wb_\ell^I$ to a lattice points $\mathbf{t}_\ell^R, \mathbf{t}_\ell^I  \in \Lambda_\ell \cap \mathcal{V}$, using $\phi$ from Lemma \ref{l:fieldtolattice},
\begin{align}
\mathbf{t}_\ell^R &= \phi(\wb_\ell^R), \qquad \mathbf{t}_\ell^I = \phi(\wb_\ell^I) \ . 
\end{align} Using Theorem \ref{t:latcompreal}, these lattice points can be transmitted across the channel so that the relays can make estimates $\vbh_m^R$ and $\vbh_m^I$ of lattice equations $\vb_m^R$ and $\vb_m^I$ with coefficient vectors $\ab_m \in \ZC^L$ such that $\Pr\left( \cup_m \big\{\{ \vbh_m^R \neq \vb^R_m \} \cup  \{ \vbh^I_m \neq \vb^I_m \}\big\}  \right) < \epsilon$ for $n$ large enough so long as
\begin{align}
R_\ell <& \min_{m: a_{ml} \neq 0}  \log^+\left(\frac{P}{|\alpha_m|^2 + P \| \alpha_m \mathbf{h}_m - \mathbf{a}_m \|^2}\right)  \nonumber
\end{align} for some $\alpha_1, \ldots, \alpha_M \in \Rbb$. 
Finally, using $\phi\n$ from Lemma \ref{l:msgeqn}, each relay can produce estimates of the desired linear combinations of messages, $\ubh_m^R = \phi\n(\vbh_m^R)$ and $\ubh_m^I = \phi\n(\vbh^I_m)$, such $\Pr\left(\cup_m \big\{\{ \ubh_m^R \neq \ub_m^R \} \cup  \{ \ubh_m^I \neq \ub_m^I \}\big\} \right) < \epsilon$ where
\begin{align}
\ub_m^R &= \bigoplus_{\ell = 1}^L\Big({q_{m \ell }^R \wb_\ell^R \oplus (-q_{m \ell }^I) \wb_\ell^I}\Big) \\
\ub_m^I &= \bigoplus_{\ell = 1}^L\Big({q_{m \ell }^I \wb_\ell^R \oplus q_{m \ell }^R \wb_\ell^I}\Big) \\
q_{m\ell}^R &= g^{-1}\Big(\big[\Re({a}_{m\ell}) \big]\hspace{-0.1in}\mod p \Big) \\
q_{m\ell}^I &=  g^{-1}\Big(\big[\Im({a}_{m\ell}) \big]\hspace{-0.1in}\mod p \Big) \ .
\end{align}

\subsection{Multi-Stage Networks}

The framework developed in this section can easily be applied to AWGN networks with more than one layer of relays. Once the first layer has recovered its equations, it can just treat them as a set of messages for the second layer. The second layer simply decodes equations with coefficients that are close to the channel coefficients. This process repeats until the equations reach a destination. Since these layered equations are all linear, they can be expressed as linear equations over the original messages. 

\section{Recovering Messages} \label{s:recovery}

The primary goal of compute-and-forward is to enable higher achievable rates across an AWGN network. Relays decode linear equations of transmitted messages and pass them towards the destination nodes which, upon receiving enough equations, attempt to solve for their desired messages. In this section, we give sufficient conditions for recovering messages from a given set of equations. 

It will be useful to represent the equations in matrix form. For real-valued channels let $\Qb= \{q_{m\e}\}$ be the matrix of equation coefficients. For complex-valued channels, let $\Qb^R = \{q^R_{m \e}\}$ and $\Qb^I = \{q^I_{m \e}\}$ be the real and imaginary coefficient matrices. Using this representation, we can write out the received equations for real-valued channels in matrix form,
\begin{align*}
\Big[ \ub_1 ~\cdots ~\ub_M \Big]^T = ~ \Qb \Big[ \wb_1~ \cdots~ \wb_L \Big]^T \ . 
\end{align*} Similarly, for complex-valued channels, we can write
\begin{align*}
&\Big[ \ub^R_1 ~\cdots ~\ub^R_M~ \ub^I_1 ~\cdots ~\ub^I_M \Big]^T \\&\qquad = ~ 
\left[
\begin{array}{cc}
  \Qb^R& -\Qb^I    \\
\Qb^I  & \Qb^R     
\end{array}
\right]
 \Big[ \wb^R_1~ \cdots~ \wb^R_L~ \wb^I_1~ \cdots~ \wb^I_L \Big]^T \ . 
\end{align*}
These matrix formulations immediately yield sufficient conditions for recovery.
\begin{theorem} \label{t:recoverallreal} For real-valued channels, a destination, given $M$ linear combinations of messages with coefficient matrix $\Qb \in \Fbb_p^{M \times L}$, can recover all messages if and only if $\Qb$ has rank $L$. 
\end{theorem}

\begin{theorem} \label{t:recoverallcomplex} For complex-valued channels, a destination, given $M$ linear combinations of messages with real and imaginary coefficient matrices $\Qb^R, \Qb^I \in \Fbb_p^{M \times L}$, can recover all messages if and only if both $\Qb^R$ and $\Qb^I$ have rank $L$. 
\end{theorem}

In many cases, a destination may only be interested in a subset of the transmitted messages. Depending on the coefficients, it may be able to reduce the number of required equations. Recall that $\delta_\e$ is the unit vector with $1$ in the $\ellth$ entry and 0 elsewhere. 

\begin{theorem} \label{t:recoveronereal} For real-valued channels, a destination, given $M$ linear combinations of messages with coefficient matrix $\Qb \in \Fbb_p^{M \times L}$, can recover the message $\wb_\ell$ if there exists a vector $\cb \in \Fbb_p^M$ such that $\cb^T \Qb = \delta_\ell^T$. 
\end{theorem}

\begin{theorem} \label{t:recoveronecomplex} For complex-valued channels, a destination, given $M$ linear combinations of messages with real and imaginary coefficient matrices $\Qb^R, \Qb^I \in \Fbb_p^{M \times L}$, can recover the message $\wb_\ell$ if there exists a vector $\cb \in \Fbb_p^{2M}$ such that 
\begin{align}
\cb^T \left[
\begin{array}{cc}
  \Qb^R& -\Qb^I    \\
\Qb^I  & \Qb^R     
\end{array}
\right] = \delta_\ell^T \ . \label{e:recoverone}
\end{align}
\end{theorem}
\begin{proof}
Clearly, the vector $\cb$ can be applied to the received equations $[ \ub^R_1 ~\cdots ~\ub^R_M~ \ub^I_1 ~\cdots ~\ub^I_M ]^T$ to recover $\wb_\ell^R$. Let $\cb^R$ denote the first $M$ elements of $\cb$, $\cb^I$ denote the last $M$ elements, and let $\tilde{\mathbf{c}} = [ -\cb^I ~\cb^R]$. By symmetry, replacing $\cb$ with $\tilde{\mathbf{c}}$ in (\ref{e:recoverone}) will yield the unit vector $\delta^T_{\ell + L}$ instead of $\delta^T_\ell$. Thus, $\tilde{\mathbf{c}}$ can be used to extract $\wb_\ell^I$ from the equations.
\end{proof}

\begin{remark} These conditions can also be stated directly in terms of the coefficient vectors $\ab_1, \ldots, \ab_M$. For real-valued channels, set $\Ab = [\ab_1 ~\cdots ~\ab_M]^T$. Now, we can substitute $\Qb$ with $\Ab$ in Theorems \ref{t:recoverallreal} and \ref{t:recoveronereal} so long as all operations are taken modulo $p$. For complex-valued channels, the same holds true for Theorems \ref{t:recoverallcomplex} and \ref{t:recoveronecomplex} if we replace $\Qb^R$ with $\Re(\Ab)$ and $\Qb^I$ with $\Im(\Ab)$.\end{remark} 

It may be more convenient to evaluate the rank of the coefficients directly on the complex field. This is possible, given some mild assumptions on the equation coefficients.

\begin{theorem} \label{t:rankcomplex}
Assume that, in an AWGN network, the magnitude of each equation coefficient is upper bounded by a constant $a_{\text{MAX}}$. Then, for sufficiently large blocklength $n$ and field size $p$, there exists a set of nested lattice codes such that a destination can recover all $L$ messages from $L$ equations if their coefficient matrix $\Ab = [\ab_1 ~\cdots ~\ab_L]^T$ is full rank over the complex field. 
\end{theorem}
\begin{proof}
$\Ab$ is full rank over the complex field if and only if its real-valued representation $\mathbf{\tilde{A}}$ is full rank over the reals. Recall that a matrix is full rank only if its determinant is non-zero. We will now show that for sufficiently large $p$, if the determinant of $\mathbf{\tilde{A}}$ is non-zero over the reals it is non-zero modulo $p$. The determinant over $\Rbb$ can be written as
\begin{align}
\det(\mathbf{\tilde{A}}) = \sum_{\sigma \in \mathcal{S}}{\mbox{sgn}(\sigma) \prod_{m = 1}^{2L}{\tilde{a}_{m \sigma(m)}}}
\end{align} where $\mathcal{S}$ is the set of all permutations of $\{1,2,\ldots,2L\}$, sgn$(\sigma)$ is the signature of the permutation which is equal to $1$ for even permutations and $-1$ for odd permutations, and $\tilde{a}_{m\e}$ are the entries of $\mathbf{\tilde{A}}$. Using the upper bound on the magnitudes of the $a_{m\e}$ and the fact that $| \mathcal{S} | = (2L)!$, the determinant is lower and upper bounded as follows:
\begin{align}
-(2L)! (a_{\text{MAX}})^{2L} \leq \det(\mathbf{\tilde{A}}) \leq (2L)! (a_{\text{MAX}})^{2L} \ . 
\end{align} The determinant under modulo $p$ arithmetic can be written as
\begin{align} 
\left[\sum_{\sigma \in \mathcal{S}}{\mbox{sgn}(\sigma) \prod_{m = 1}^{2L}{\tilde{a}_{m \sigma(m)}}}\right] \modp \ . 
\end{align}
Since the underlying field size $p \rightarrow \infty$ as $n \rightarrow \infty$, for large enough blocklength $n$, we can use the bounds on $\det(\mathbf{\tilde{A}})$ to show that the determinant modulo $p$ does not wrap   around zero. This immediately implies that it is zero if and only the determinant is zero over the reals. 
\end{proof}

\begin{remark}
Theorem \ref{t:rankcomplex} can also be stated in terms of bounds on the channel coefficients. For instance, if $|h_{m\e}| < h_{\text{MAX}}$, then we can use the bound in Lemma \ref{l:coeffbound}, to show that $|a_{m \e}|$ is bounded as well. More generally, the result holds if the channel coefficients are drawn from a distribution such that $Pr\left(\cup_{m\e}\{|h_{m\e}| > h_{\text{MAX}} \}\right) \rightarrow 0$ as $h_{\text{MAX}} \rightarrow \infty$. In this case, we choose $h_{\text{MAX}}$ such that this probability is very small and can be absorbed into the total probability of error for our scheme. The result follows by taking an appropriate increasing sequence of $h_{\text{MAX}}$. 
\end{remark}

\begin{figure}[h]
\begin{center}
\psset{unit=0.55mm}
\begin{pspicture}(0,-20)(150,50)

\footnotesize

\rput(1,32){$\wb_1$}
\rput(1,12){$\wb_2$}
\rput(0.5,-18){$\wb_M$}
\psline{->}(5,32)(10,32)
\psline{->}(5,12)(10,12)
\psline{->}(5,-18)(10,-18)

\rput(2,0){
 \psframe(8,27)(25,37)
\rput(17,32){$\mathsf{Tx}~1$} \rput(31,35.5){$\xb_1$}
\psline{->}(25,32)(38,32)

 \psframe(8,7)(25,17)
\rput(17,12){$\mathsf{Tx}~2$} \rput(31,15.5){$\xb_2$}
\psline{->}(25,12)(38,12)

\rput(18,-1){$\vdots$}

 \psframe(8,-23)(25,-13)
\rput(17,-18){$\mathsf{Tx}~M$} \rput(31,-14.5){$\xb_M$}
\psline{->}(25,-18)(38,-18)
}


\psframe(40,-23)(55,37)
\rput(47.5,7){\large{$\mathbf{H}$}}

\rput(13,0){
\psline{->}(42,32)(46.5,32)
\pscircle(49,32){2.5} \psline{-}(47.75,32)(50.25,32)
\psline{-}(49,30.75)(49,33.25) \psline{->}(49,39.5)(49,34.5) \rput(49,42.5){$\zb_1$}
\psline{->}(51.5,32)(63,32) \rput(57,35.5){$\yb_1$}

\psline{->}(42,12)(46.5,12)
\pscircle(49,12){2.5} \psline{-}(47.75,12)(50.25,12)
\psline{-}(49,10.75)(49,13.25) \psline{->}(49,19.5)(49,14.5) \rput(49,22.5){$\zb_2$}
\psline{->}(51.5,12)(63,12)  \rput(57,15.5){$\yb_2$}

\psline{->}(42,-18)(46.5,-18)
\pscircle(49,-18){2.5} \psline{-}(47.75,-18)(50.25,-18)
\psline{-}(49,-19.25)(49,-16.75) \psline{->}(49,-10.5)(49,-15.5) \rput(49,-7.5){$\zb_M$}
\psline{->}(51.5,-18)(63,-18) \rput(57,-14.5){$\yb_M$}

\psframe(63,27)(85,37) \rput(74.5,32){$\mathsf{Relay}~1$}

\psframe(63,7)(85,17) \rput(74.5,12){$\mathsf{Relay}~2$}

\rput(69.5,-1){$\vdots$}

\psframe(63,-23)(85,-13) \rput(74.5,-18){$\mathsf{Relay}~M$}
}

\rput(63,0){
\psline{->}(35,32)(46.5,32)
\pscircle(49,32){2.5} \psline{-}(47.75,32)(50.25,32)
\psline{-}(49,30.75)(49,33.25) \psline{->}(49,39.5)(49,34.5) \rput(49.5,42.5){$\zb^{a}_1$}
\psline{->}(51.5,32)(63,32)
 \rput(40,36){$\xb^{a}_1$}
 \rput(56,36){$\yb^{a}_1$}

\psline{->}(35,12)(46.5,12)
\pscircle(49,12){2.5} \psline{-}(47.75,12)(50.25,12)
\psline{-}(49,10.75)(49,13.25) \psline{->}(49,19.5)(49,14.5) \rput(49.5,22.5){$\zb^{a}_2$}
\psline{->}(51.5,12)(63,12) 
\rput(40,16){$\xb^{a}_2$}
\rput(56,16){$\yb^{a}_2$}

\psline{->}(35,-18)(46.5,-18)
\pscircle(49,-18){2.5} \psline{-}(47.75,-18)(50.25,-18)
\psline{-}(49,-19.25)(49,-16.75) \psline{->}(49,-10.5)(49,-15.5) \rput(49.5,-6.5){$\zb^{a}_M$}
\psline{->}(51.5,-18)(63,-18) 
\rput(40.5,-14){$\xb^{a}_M$}
\rput(56.5,-14){$\yb^{a}_M$}
}

\psframe(126,-23)(140,37) \rput(133,7){\mbox{\normalsize{$\mathsf{Rx}$}}}
\psline{->}(140,7)(145,7) \rput(150,18){$\mathbf{\hat{w}}_1$}
\rput(150,12){$\mathbf{\hat{w}}_2$} \rput(150,6){$\vdots$} 
\rput(150,-3){$\mathbf{\hat{w}}_M$}

\end{pspicture}
\end{center}
\caption{A linear relay network where compute-and-forward is beneficial.} \label{f:hadamardrelay}
\end{figure}
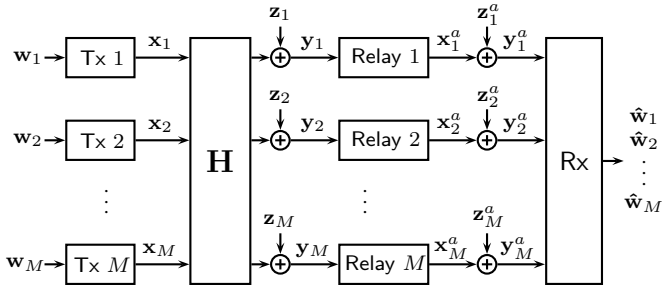

\begin{example} \label{e:hadamard}
Consider the AWGN network in Figure \ref{f:hadamardrelay}. Transmitters $1$ through $M$ send messages $\wb_1, \ldots, \wb_M$ through a channel $\Hb$ to M relays. Each relay has a point-to-point AWGN channel to the receiver which wants to recover all of the messages at the highest possible symmetric rate. Each channel input has power $P$ and all noise terms are i.i.d. circularly symmetric Gaussian with variance $1$. Let $\Hb$ be an $M \times M$ Hadamard matrix. (We assume that $M$ is chosen such that a Hadamard matrix of that size exists. ) Recall that a Hadamard matrix has $\pm 1$ entires such that $\Hb \Hb^T = M \mathbf{I}$.

Using Theorems \ref{t:optcompcomplex} and \ref{t:rankcomplex} and setting the coefficient vectors equal to the channel vectors, $\ab_m = \hb_m$, compute-and-forward can achieve 
\begin{align}
R_{_{\text{COMP}}}= \log^+\left(\frac{1}{M} + P\right)
\end{align} bits per channel use per user since $\Hb$ is full rank. It can be shown that decode-and-forward, amplify-and-forward, and compress-and-forward (with i.i.d. Gaussian codebooks) can achieve 
\begin{align}
R_{_{\text{DF}}}&= \frac{1}{M}\log\left(1 + MP\right)\\
R_{_{\text{AF}}} =R_{_{\text{CF}}}&= \log\left(1 + P\left(\frac{P}{MP + 1}\right)\right)
\end{align} bits per channel use per user. Compute-and-forward is the dominant strategy except at very low power and it rapidly approaches the upper bound $R_{_{\text{UPPER}}}= \log\left(1 + P\right)$ as $P \rightarrow \infty$. As $M$ increases the rates of decode-and-forward, amplify-and-forward, and compress-and-forward go to $0$.
\end{example} 

\section{Successive Cancellation} \label{s:cancel}

Once a relay has recovered an equation of messages, it can subtract its contribution from the channel observation. This results in a residual channel output from which it can extract a different equation, potentially with a higher rate than possible over the original channel. One key difference from standard applications of successive cancellation is that the relay cannot completely cancel out all channel inputs associated with the decoded equation. This is because in the first step, it only decodes an integer combination of the messages, which is often not the same as the linear combination taken by the channel.

We demonstrate an achievable region for decoding two different equations using successive cancellation at each relay. This can be easily generalized to more than two equations. For succinctness, we only state this result for real-valued channel models.

\begin{theorem} \label{t:twocomp}
Let $\hb_1, \ldots, \hb_M \in \Rbb^L$ denote the channel vectors and $R_\e$ denote the message rates. Each relay can first decode an equation with coefficient vector $\ab_m \in \Zbb^L$ and then one with coefficient vectors $\bb_m \in \Zbb^L$ if 
\begin{align*}
&R_\e < \min\left(\min_{m: a_{m \e} \neq 0}\mathcal{R}_1(\hb_m, \ab_m), \min_{m: b_{m \e} \neq 0}\mathcal{R}_2(\hb_m, \ab_m, \bb_m) \right) \\
&\mathcal{R}_1(\hb_m, \ab_m) =  \frac{1}{2}\log^+\left(\frac{P}{\alpha_m^2 + P \| \alpha_m \mathbf{h}_m - \mathbf{a}_m \|^2}\right)\\
&\mathcal{R}_2(\hb_m, \ab_m, \bb_m) \\
&=\begin{cases}
{\displaystyle \frac{1}{2}\log^+\left(\frac{P}{\beta_m^2 + P \sum_{\e \neq i}|\beta_m h_{m \e} - b_{m\e}|^2}\right)},~~\ab_m = \delta_i ,&~ \\
{\displaystyle \frac{1}{2}\log^+\left(\frac{P}{\beta_m^2 + P \| \beta_m \mathbf{h}_m - \tau_m\mathbf{a}_m - \bb_m \|^2}\right)},\mbox{\small{otherwise.}} &~
 \end{cases}
 \end{align*} for some choice of $\alpha_m, \beta_m \in \Rbb$ and $\tau_m \in \Zbb$. 
\end{theorem}

\begin{proof}
All messages are mapped onto lattice points, dithered, and transmitted across the channel as in the proof of Theorem \ref{t:basiccompcomplex}. The first set of equations can be reliably decoded using the procedure from Theorem \ref{t:basiccompcomplex} as well. Now, we condition on the event that each relay has successfully recovered the equation with coefficient vectors $\ab_m$.

Consider the case where the first coefficient vector at relay $m$ is a unit vector $\ab_m = \delta_i$. This means that relay $m$ can successfully decode the message $\wb_i$ from encoder $i$. It can then replicate the encoding process to get $\xb_i$. Now, the relay removes $\xb_i$ from $\yb_m$,
\begin{align}
\yb_m - h_{mi} \xb_i = \sum_{\e \neq i} {h_{m\e} \xb_\e + \zb_m} \ ,
\end{align} and uses this as a channel output for Theorem \ref{t:basiccompreal} to get the equation with coefficient vector $\mathbf{\tilde{b}}_m$ which is equal to $\bb_m$ except that it has $0$ in the $\ith$ position. It then adds $b_{mi} \wb_i$ to the recovered equation to get $\bb_m$. 

If $\ab_m$ is not a unit vector, the decoder has access to the lattice equation
\begin{align}
\vb_m = \left[\sum_{\ell = 1}^L{a_{m\e} \tb_\e} \right]\modl 
\end{align} from which it computes
\begin{align}
\mathbf{\bar{v}}_m &= \left[\vb_m - \sum_{\ell = 1}^L{a_{m\e}\db_\e } \right]\modl \nonumber = \left[\sum_{\ell = 1}^L{a_{m\e}\xb_\e} \right]\modl \\
\mathbf{\tilde{y}}_m &= \big[\beta_m\yb_m -\tau_m \mathbf{\bar{v}}_m\big] \modl \nonumber \\
 &= \left[\sum_{\ell = 1}^L{(\beta_m{h}_{m\e} - \tau_m a_{m\e})\xb_\e}+\zb_m \right]\modl \nonumber  \ .
\end{align} Now we can follow the steps in the proof of Theorem \ref{t:latcompreal}. In (\ref{e:decodedither}), replace $\alpha_m \yb_m$ with $\mathbf{\tilde{y}}_m$. In all steps of the proof, substitute $a_{m\e}$ with $b_{m\e}$, $\alpha_m h_{m\e}$ with $\beta_m h_{m\e} - \tau_m a_{m\e}$, and, if it has not already been replaced, $\alpha_m$ with $\beta_m$. 
\end{proof}

\begin{remark} Given $\ab_m$, $\bb_m$, and $\tau_m$, we can solve for the optimal $\alpha_m$ and $\beta_m$ following the steps of the proof of Theorem \ref{t:optcompreal}.
\end{remark}
\begin{remark} The restriction of $\tau_m$ to the integers stems from the fact that (\ref{e:modprop3}) only holds for integer coefficients.
\end{remark}

\begin{example} There are $L= 4$ transmitters and $M=1$ relay and the channel vector is $\hb_1 = [10~10~8~8]^T$. The relay wants to first decode the equation with coefficient vector $\ab_1 = [1~1~1~1]^T$ and then with coefficient vector $\bb_1=[1~~1~-1~-1]^T$. Using Theorem \ref{t:twocomp}, this is possible if the message rates satisfy
\begin{align}
R_\e < \min{\left(\frac{1}{2}\log^+\left(\frac{1}{4} + \frac{81 P}{1 +4 P}\right), \frac{1}{2}\log^+\left(\frac{1}{328} + P\right)\right)} \nonumber 
\end{align} by using $\tau_1 = 9$ so that $\hb_1 - \tau_1  \ab_1 = \bb_1$. Note that if we applied Theorem \ref{t:basiccompreal} directly to decode $\bb_1$, we would not be able to get a positive rate. 
\end{example}

\begin{remark}
As noted in Remark \ref{r:piecewise}, it may be more efficient to recover an equation piecewise by recovering equations of subsets of messages and taking an appropriate linear combination of these equations. Theorem \ref{t:twocomp} is strictly better for this process than Theorem \ref{t:basiccompreal}.
 \end{remark}

\subsection{Multiple-Access}
Assume there is only one relay and that it wants to recover all transmitted messages. This is the standard Gaussian multiple-access problem whose capacity region is well-known to be the set of all rate tuples $(R_1, \ldots, R_L)$ satisfying
\begin{align}
\sum_{\e \in S}{R_\e} ~~<~~ \frac{1}{2}\log{\left(1 + P \sum_{\e \in S}{|h_{1\e}|^2}\right)}
\end{align} for all subsets $S \subseteq \{1,2,\ldots, L\}$ \cite[Theorem 14.3.5]{coverthomas}. We now show that compute-and-forward includes the multiple-access capacity region as a special case. First, we consider the corner point of the capacity region associated with decoding the messages in ascending order. From Example \ref{ex:onemsg}, the first message can be decoded (while treating the others as noise) if
\begin{align}
R_1 < \frac{1}{2}\log{\left(1 + \frac{|h_{11}|^2P}{1 + P \sum_{i = 2}^L{|h_{1i}|^2}}\right)} \ .
\end{align} Using successive cancellation, the relay removes $\xb_1$ from the channel observation to get $\sum_{\e=2}^L{h_{1\e}\xb_\e} +\zb_1$. It then repeats the above procedure for each message in ascending order to get
\begin{align}
R_\e < \frac{1}{2}\log{\left(1 + \frac{|h_{1\e}|^2P}{1 + P \sum_{i=\ell+1}^L{|h_{1i}|^2}} \right)} \ . 
\end{align} The resulting rate tuple is a corner point of the multiple-access capacity region. By changing the decoding order, any corner point is achievable. Note that any point on the boundary of the capacity region is achievable by time-sharing corner points. 

\begin{remark}
One interesting open problem is to develop \textit{joint decoding} for the compute-and-forward framework. Of course, within the context of multiple-access, this is possible with nested lattice codewords as they have good statistical properties. Extending joint decoding to recovering equations of messages may enlarge the computation rate region. 
\end{remark}

\section{Superposition}\label{s:superposition}

In the previous section, we considered the scenario where each relay decodes several equations, but the transmitters each use a single codebook (as in Theorem \ref{t:basiccompreal}). However, when decoding multiple equations, it is sometimes useful to superimpose multiple codebooks. We investigate this possibility in this section for real-valued channels. As before, the complex case follows naturally.

 We will assume that there are two levels $A$ and $B$ and that each relay wants to a recover an equation from both levels. (If it is not interested in a level, it can just set its desired coefficients to zero.)  

Each encoder has two messages $\wb_{\e A}$ and $\wb_{\e B} $ with rates $R_{\e A}$ and $R_{\e B}$ respectively. Relay $m$ wants to decode equations $\ub_{mA}$ and $\ub_{mB}$ with coefficient vectors $\ab_m$ and $\bb_m$, respectively, for $m = 1,2,\ldots, M$. In the theorem below, we give achievable rates for this scenario by combining superposition and successive cancellation. The basic idea is to superimpose two lattice codes at each receiver scaled by $\gamma_{\e A}$ and $\gamma_{\e B}$ to ensure that the power constraint is met.

\begin{theorem} \label{t:twolevels}
Choose $\gamma_{\e A}, \gamma_{\e B}$ such that $ \gamma_{\e A}^2 + \gamma_{\e B}^2 = 1$. For channel vectors $\hb_1, \ldots, \hb_M \in \Rbb^L$, the relays can first decode any set of linear equations over $\wb_{\e A}$ with coefficient vectors $\ab_1,  \ldots, \ab_M \in \Zbb^L$ and then any set of linear equations over $\wb_{\e B}$ with coefficient vectors $\bb_1, \ldots, \bb_M \in \Zbb^L$ if \begin{align*}
&R_{\e A} < \min_{m: a_{m \e} \neq 0}\frac{1}{2}\log^+\left(\frac{P}{N_{mA}}\right)\\
&R_{\e B} < \min_{m: b_{m \e} \neq 0}\frac{1}{2}\log^+\left(\frac{P}{N_{mB}}\right) 
\end{align*} where
\begin{align*}
&\hb_{m A}= [\gamma_{1A} h_{m1} ~\cdots ~\gamma_{LA}h_{mL}]^T\\
&\hb_{m B}= [\gamma_{1B} h_{m1} ~\cdots ~\gamma_{LB}h_{mL}]^T\\
&N_{mA} = |\alpha_m|^2(1 + P \| \hb_{mB}\|^2) + P \| \alpha_m \mathbf{h}_{m A}- \mathbf{a}_m \|^2 \\
&N_{mB1} = \\
&~~|\beta_m|^2(1 + P \sum_{\e \neq i}{| \gamma_{\e A} h_{m\e}|^2})+ P \| \beta_m \mathbf{h}_{m B}- \mathbf{b}_m \|^2 \\
&N_{mB2} =\\
&~~ |\beta_m|^2 + P \| \beta_m \hb_{mA} - \tau_m \ab_m\|^2 + P \| \beta_m \mathbf{h}_{m B}- \mathbf{b}_m \|^2 \\
&N_{mB} = \begin{cases}
 N_{mB1}
, & \ab_m = \delta_i \mbox{ for some $i$}, \\
N_{mB2}, & \mbox{otherwise.}
 \end{cases}
 \end{align*} for some choice of $\alpha_m, \beta_m \in \Rbb$ and $\tau_m \in \Zbb$. 
\end{theorem}

\begin{proof}
Choose two sets of nested lattices $\Lambda \subset \Lambda_{LA} \subset \cdots \subset \Lambda_{1A}$, $\Lambda \subset \L_{LB} \subset \cdots \subset \L_{1B}$ with appropriate rates where $\L$ is the coarse lattice with second moment $P$. Each encoder maps its messages onto lattice points using $\phi$ from Lemma \ref{l:fieldtolattice} and dithers them with $\db_{\e A}, \db_{\e B}$ drawn independently and uniformly over the fundamental Voronoi region $\Vm$ of $\L$,
\begin{align*}
&\tb_{\e A} = \phi_A(\wb_{\e A})~~~~~~~~~~~~~~~~~~\tb_{\e B} = \phi_B(\wb_{\e B}) \\
&\xb_{\e A}  = [\tb_{\e A} - \db_{\e A}]\modl ~~~~~ \xb_{\e B}  = [\tb_{\e B} - \db_{\e B}]\modl 
\end{align*}It then combines $\xb_{\e A}$ and $\xb_{\e B}$ according to $\gamma_{\e A}$ and $\gamma_{\e B}$ which guarantees the power constraint is met:
\begin{align}
\xb_\e &= \gamma_{\e A} \xb_{\e A} + \gamma_{\e B} \xb_{\e B}\\
 E[ \| \xb_\e \|^2 ] &=  \gamma_{\e A} ^2 nP + \gamma_{\e B}^2 nP = nP
 \end{align}
 At each receiver, we can just treat the channel output as if it came from $2L$ transmitters labelled $1A, \ldots, LA, 1B, \ldots, LB$. We can write the channel to receiver $m$ and the desired coefficient vectors as
 \begin{align}
\mathbf{\tilde{h}}_m =  \left[
\begin{array}{c}
\hb_{mA}\\ 
 \hb_{mB}
\end{array}
\right]
 ~~~~~~ \mathbf{\tilde{a}}_m =  \left[
\begin{array}{c}
\ab_m\\ 
 \mathbf{0}
\end{array}
\right]
 ~~~~~~\mathbf{\tilde{b}}_m =  \left[
\begin{array}{c}
\mathbf{0}\\ 
 \bb_m
\end{array}
\right]. \nonumber
 \end{align} We can now directly apply Theorem \ref{t:twocomp} with $\mathbf{\tilde{h}}_m$, $\mathbf{\tilde{a}}_m$, and $\mathbf{\tilde{b}}_m$ to get the desired result.
\end{proof}
\begin{remark}
As before, given $\ab_m$, $\bb_m$,  $\tau_m$, and $\gamma_{\ell A}$ we can solve for the optimal $\alpha_m$ and $\beta_m$ following the steps of the proof of Theorem \ref{t:optcompreal}.
\end{remark}

\begin{remark}
In order to keep the notation manageable, we have chosen to present the superposition strategy in Theorem \ref{t:twolevels} only for two levels. There are several immediate extensions, including:
\begin{itemize}
\item More than two levels.
\item Allowing a different decoding order at each relay.
\item Equations spanning different levels.
\end{itemize}
\end{remark}

\begin{example} There are $L= 3$ transmitters and $M=1$ relay and the channel vector is $\hb_1 = [1 ~ 1 ~ \sqrt{2}]^T$. Set the scaling coefficients to be $\gamma_{1A}=\gamma_{2A} = 0$, $\gamma_{1B}=\gamma_{2B} = 1$, and $\gamma_{3A}=\gamma_{3B} = 1/\sqrt{2}$. The relay wants to first decode the equation with coefficient vector $\ab_1 = [0~0~1]^T$ from level $A$ and then the equation with coefficient vector $\bb_1=[1~1~1]^T$ from level $B$. Using Theorem \ref{t:twolevels}, this is possible if the message rates satisfy
\begin{align}
R_{3A} &< \frac{1}{2}\log\left(1 + \frac{P}{1 +3 P}\right)\\
R_{\e B}&< \frac{1}{2}\log^+\left(\frac{1}{3} + P\right)~~~~~\e = 1,2,3.
\end{align} 
\end{example}

\begin{remark} It can be shown that nested lattice codes can approach the capacity region of the standard Gaussian broadcast problem. See \cite{zse02} for more details.
\end{remark}

\begin{remark}
For an application of this superposition scheme to a backhaul-limited cellular uplink network, see \cite{nsgs09}.
\end{remark}

\section{Outage Formulation} \label{s:outage}

So far, we have considered fixed channel coefficients. Now, we demonstrate that our scheme can be applied to the slow fading scenario. This further emphasizes the fact that our compute-and-forward scheme does {\em not} require channel state information at the transmitters. Under a slow fading model, the channel matrix $\Hb$ is chosen according to some probability distribution and then remains fixed for all time. As a result, we must accept some probability that the rate used by the transmitters is above the maximum rate permitted for those channel coefficients. For an achievable strategy with rate $R_{\text{SCHEME}}(\Hb)$ for fixed $\Hb$, this \textit{outage probability} is given by
\begin{align}
\rho_{\text{OUT}}(R) &= \Pr\left(R_{\text{SCHEME}}(\Hb) < R\right) \ . 
\end{align}
We can also characterize the performance of a given strategy by its \textit{outage rate},
\begin{align}
R_{\text{OUT}}(\rho) &= \sup\{R:\rho_{\text{OUT}}(R) \leq \rho\}.
\end{align}

\begin{figure}[h]
\centering
\includegraphics[width=3.75in]{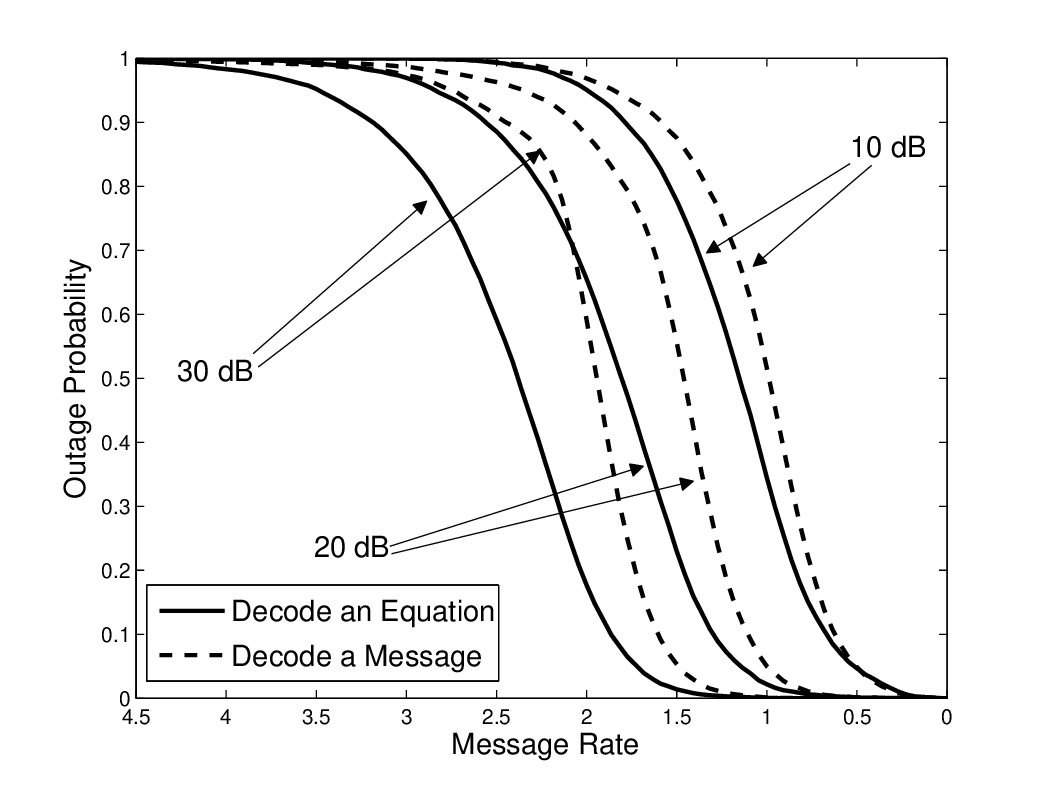}
\caption{Outage probability for a relay that receives $\yb = h_1 \xb_1 + h_2 \xb_2 + h_3 \xb_3 + \zb$ where the $h_\ell$ are i.i.d. according to $\mathcal{N}(0,1)$. The ``Decode a Message'' strategy uses standard random codes and joint typicality decoding to recover at least one of the messages $\wb_1, \wb_2,$ or $\wb_3$. The ``Decode an Equation'' strategy uses compute-and-forward to recover some linear equation $a_1 \wb_1 \oplus a_2 \wb_2 \oplus a_3 \wb_3$. }\label{f:outageprobs}
\end{figure}

\begin{example}
There are three transmitters that communicate to a single relay over a real-valued AWGN multiple-access channel. The channel coefficients $h_\e$ are i.i.d. according to $\mathcal{N}(0,1)$ and are only known to the relay. Each transmitter has a single message $\wb_\e$ of rate $R$. Usually, the relay would only have the choice of decoding one message, two messages, or all three messages with the rates given by the multiple-access rate region.\footnote{Those messages that are not decoded are treated as noise.} The resulting outage probabilities for this strategy are plotted in Figure \ref{f:outageprobs} for $P = 10, 20,$ and $30$dB. We also plot the performance of the compute-and-forward strategy from Theorem \ref{t:basiccompreal}, which permits the relay to decode any linear equation of the messages, $\ub = \bigoplus_\ell a_\ell \wb_\ell$ so long as at least one of the coefficients is not equal to zero. 
\end{example}

The example above demonstrates that decoding an equation is often easier than decoding a message. In order to use compute-and-forward for network communication, we also need that the end-to-end linear transformation of the desired messages is full rank. The next section explores this issue through a case study.

\section{Case Study: Distributed MIMO} \label{s:distmimo}

We will now compare the outage performance of compute-and-forward to the performance of classical relaying strategies over a simple network. Consider the two user distributed MIMO network in Figure \ref{f:distmimo}. There are two sources, two relays, and one destination. The relays see the transmitters through $\Hb$ whose entries are i.i.d. Rayleigh, $h_{m\ell} \sim \mathcal{CN}(0,1)$. We assume that relay $m$ only knows the channel vector $\hb_m$ to itself. Each relay is given a bit pipe with rate $R_0$ bits per channel use to the destination. The destination would like to recover both message $\wb_1$ and $\wb_2$ at the highest possible symmetric outage rate. Recall that for a symmetric rate point to be achievable, both transmitters must be able to communicate their messages with at least that rate.

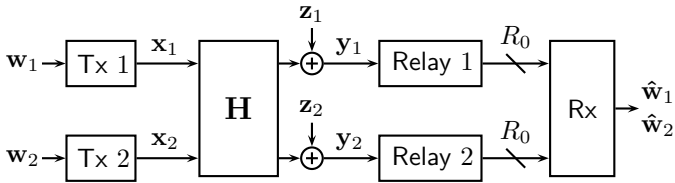
\begin{figure}[h]
\begin{center}
\psset{unit=0.63mm}
\begin{pspicture}(0,8)(135,44)

\rput(1,32){$\wb_1$}
\rput(1,12){$\wb_2$}
\psline{->}(5,32)(10,32)
\psline{->}(5,12)(10,12)

\rput(2,0){
 \psframe(8,27)(23,37)
\rput(16,32){$\mathsf{Tx}~1$} \rput(29,35.5){$\xb_1$}
\psline{->}(23,32)(36,32)

 \psframe(8,7)(23,17)
\rput(16,12){$\mathsf{Tx}~2$} \rput(29,15.5){$\xb_2$}
\psline{->}(23,12)(36,12)
}

\psframe(38,8)(55,37)
\rput(46.5,22.5){\large{$\mathbf{H}$}}

\rput(13,0){
\psline{->}(42,32)(46.5,32)
\pscircle(49,32){2.5} \psline{-}(47.75,32)(50.25,32)
\psline{-}(49,30.75)(49,33.25) \psline{->}(49,39.5)(49,34.5) \rput(49,42.5){$\zb_1$}
\psline{->}(51.5,32)(63,32) \rput(57,35.5){$\yb_1$}

\psline{->}(42,12)(46.5,12)
\pscircle(49,12){2.5} \psline{-}(47.75,12)(50.25,12)
\psline{-}(49,10.75)(49,13.25) \psline{->}(49,19.5)(49,14.5) \rput(49,22.5){$\zb_2$}
\psline{->}(51.5,12)(63,12)  \rput(57,15.5){$\yb_2$}

\psframe(63,27)(85,37) \rput(74.5,32){$\mathsf{Relay}~1$}

\psframe(63,7)(85,17) \rput(74.5,12){$\mathsf{Relay}~2$}

}

\rput(63,0){
\rput(42,37.5){$R_0$}
\psline{-}(40,34)(44,30)
\psline{->}(35,32)(49,32)

\rput(42,17.5){$R_0$}
\psline{-}(40,14)(44,10)
\psline{->}(35,12)(49,12)

}

\psframe(112,8)(126,37) \rput(119,22.5){\mbox{\normalsize{$\mathsf{Rx}$}}}
\psline{->}(126,22.5)(131,22.5) \rput(135,26){$\mathbf{\hat{w}}_1$}
\rput(135,19){$\mathbf{\hat{w}}_2$}

\end{pspicture}
\end{center}
\caption{Two transmitters communicate to a distributed MIMO receiver with two antennas. Each antenna has a rate $R_0$ bit pipe to the receiver.} \label{f:distmimo}
\end{figure}

The basic compute-and-forward strategy has each relay decode the equation with the highest rate and pass that to the destination. If the equations received by the destination are full rank, decoding is successful. However, at low SNR, the probability that the equations are not full rank is quite high as shown in Figure \ref{f:distmimoranks}. One simple solution is to force each relay to choose an equation with $a_{mm} \neq 0$. This results in equations that are far more likely to be solvable at the expense of slightly lower computation rates.\footnote{More work is needed to develop distributed coefficient selection strategies that operate on the optimal tradeoff between computation rate and matrix rank.} The achievable rates for these two strategies are given below and are plotted in Figure \ref{f:distmimorates} for $R_0 = 2$ and outage probability $\rho = 1/4$. 

\begin{align*}
&R_{\text{MAX},m} = \max_{\ab_m}~ \log^+{\left(\left(\| \mathbf{a}_m \|^2 - \frac{P ~ | \mathbf{h}_m^* \mathbf{a}_m |^2}{1 + P\|\mathbf{h}_m\|^2}\right)\n\right)}\\
&R_{\text{NZ},m} = \max_{\substack{\ab_m \\ a_{mm} \neq 0}} \ \log^+{\left(\left(\| \mathbf{a}_m \|^2 - \frac{P ~ | \mathbf{h}_m^* \mathbf{a}_m |^2}{1 + P\|\mathbf{h}_m\|^2}\right)\n\right)}\\
&R_{\text{COMP}}(\Hb) = \begin{cases}
 \min\Big( {\displaystyle \big(\min_m R_{\text{MAX},m}}\big), R_0 \Big)
 &\mathsf{rank}(\Ab) = 2, \\
0 & \mbox{otherwise.}
 \end{cases}\\
&R_{\text{CNZ}}(\Hb) = \begin{cases}
 \min\Big( {\displaystyle \big(\min_m R_{\text{NZ},m}}\big), R_0 \Big)
 &\mathsf{rank}(\Ab) = 2, \\
0 & \mbox{otherwise.}
 \end{cases}
\end{align*}

For decode-and-forward, we require that each relay is responsible for a single message. It attempts to recover this message either by treating the other message as noise or decoding both messages. The rate for this strategy is evaluated below and plotted in Figure \ref{f:distmimorates}. For more details on decode-and-forward  for multiple relays (as well as compress-and-forward and cut-set upper bounds), see \cite{kgg05}.

\begin{align}
R_{\text{ignore},1} &= \log\left(1 + \frac{|h_{11}|^2P}{1 + |h_{12}|^2 P}\right)\\
R_{\text{ignore},2} &= \log\left(1 + \frac{|h_{22}|^2P}{1 + |h_{21}|^2 P}\right)\\
R_{\text{decode},m} &= \min\bigg(\log\left(1 + |h_{m1}|^2P\right), \nonumber \\ &~~~~~~~~~~~ \log\left(1 + |h_{m2}|^2P\right),\nonumber \\ & ~~~~~~~~~~~ \onehalf \log\left(1 + \| \hb_m \|^2 P\right) \bigg) \\
R_{\text{ii}} &= \min(R_{\text{ignore},1}, R_{\text{ignore},2})\\
R_{\text{id}} &= \min(R_{\text{ignore},1}, R_{\text{decode},2})\\
R_{\text{di}} &= \min(R_{\text{decode},1}, R_{\text{ignore},2})\\
R_{\text{dd}} &= \min(R_{\text{decode},1}, R_{\text{decode},2})\\
R_{\text{DF}}(\Hb) &= \min\big(\max(R_{\text{ii}},R_{\text{id}},R_{\text{di}},R_{\text{dd}}),R_0\big)
\end{align}
\begin{figure}[h]
\centering
\includegraphics[width=3.75in]{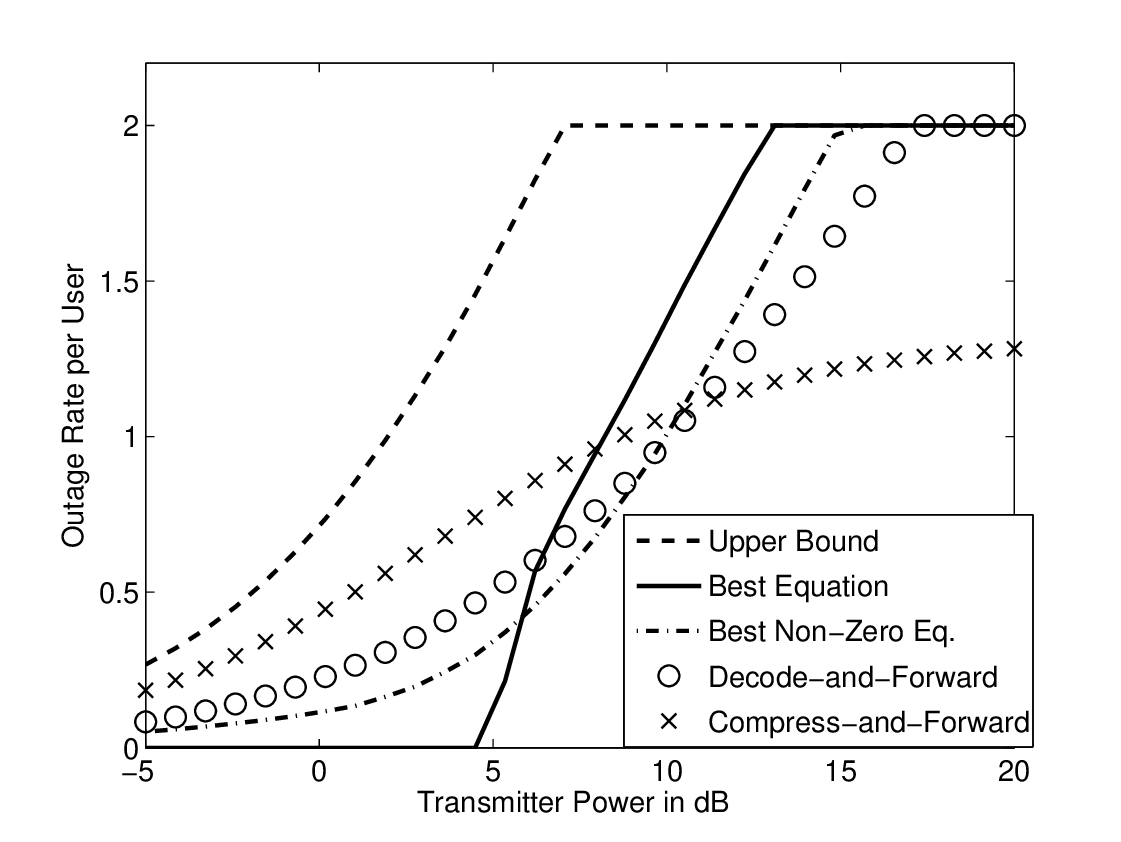}
\caption{Symmetric outage rates for the $2$-user distributed MIMO multiple-access channel with i.i.d. Rayleigh fading only known at the receivers. Here, we set $R_0 = 2$ and outage probability $\rho = 1/4$.}\label{f:distmimorates}
\end{figure}

\begin{figure}[h]
\centering
\includegraphics[width=3.75in]{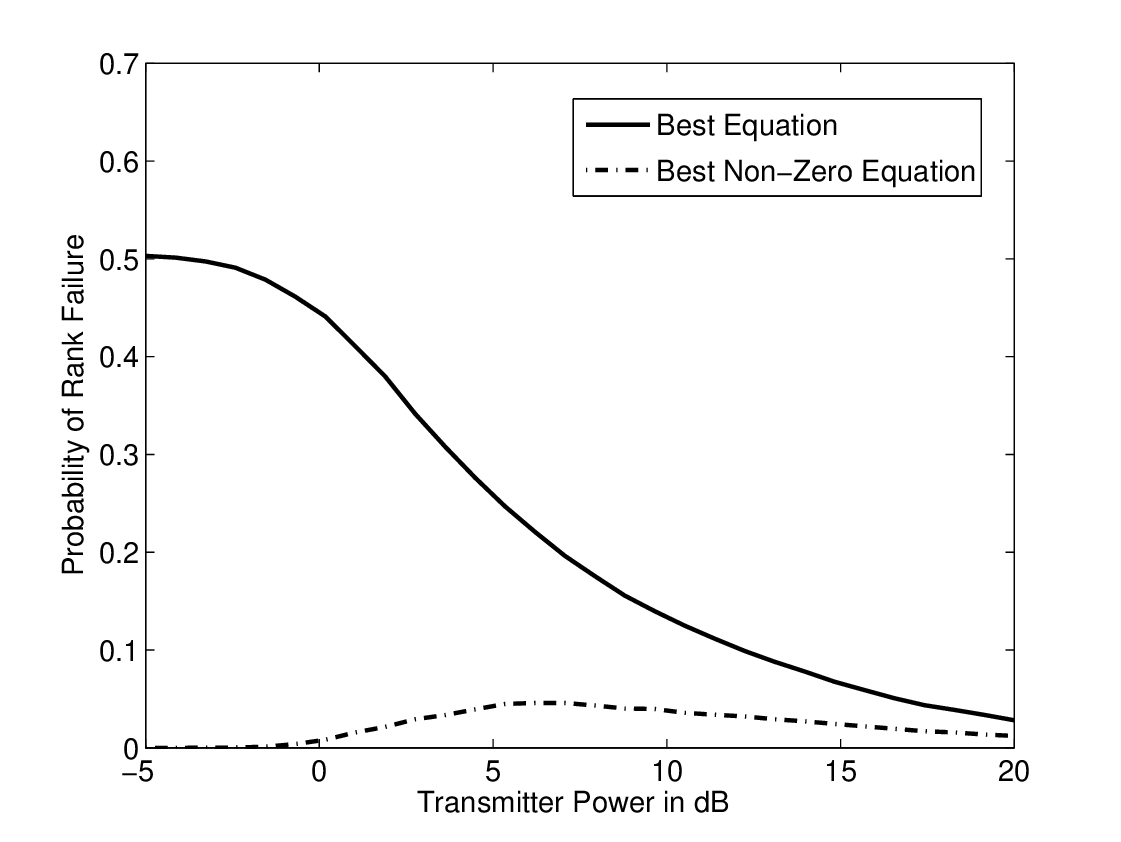}
\caption{Probability of rank failure for the $2$-user distributed MIMO multiple-access channel by having each relay decode the best equation and the best non-zero equation.}\label{f:distmimoranks}
\end{figure}
For our upper bound, we use a cut-set bound that either groups the relays with the sources or with the destination. This yields the following bound on the symmetric rate:
\begin{align}
R_{\text{MIMO}}(\Hb) &= \min\bigg(\log\left(1 + (|h_{11}|^2+|h_{21}|^2)P\right), \nonumber \\ &~~~~~~~~~~~ \log\left(1 + (|h_{12}|^2+|h_{22}|^2)P\right),\nonumber \\ & ~~~~~~~~~~~ \onehalf \log\det\left(\mathbf{I} + \Hb \Hb^* P\right) \bigg) \\
R_{\text{UPPER}}(\Hb) &= \min\big(R_{\text{MIMO}}(\Hb),R_0\big) \ . 
\end{align}

Finally, we consider the performance of compress-and-forward with i.i.d. Gaussian codebooks. The variance of the channel observation at relay $m$ is $1 + \| \hb_m \|^2 P$ and we have to compress this using $R_0$ bits. At the destination, one can equivalently write this as a MIMO channel with channel matrix $\Hb_{\text{CF}}$,
\begin{align}
\snr_{\text{CF},m} &= \frac{P (2^{R_0} -1)}{2^{R_0} + P \| \hb_m \|^2}\\
~~~~~~~~ \Hb_{\text{CF}} & =  \left[
\begin{array}{cc}
\sqrt{\snr_{\text{CF},1}/P}& 0\\ 
0&\sqrt{\snr_{\text{CF},2}/P} 
\end{array}
\right] \Hb \\
R_{\text{CF}}(\Hb) &= R_{\text{MIMO}}(\Hb_{\text{CF}}) \ .
\end{align}
 
From Figure \ref{f:distmimorates}, we can see that compute-and-forward (with the best equation) outperforms all other strategies starting at approximately $8$dB. It also saturates the bit pipes to the destination using 5dB less power per transmitter than required for decode-and-forward. However, the gains are not as dramatic as observed in Example \ref{e:hadamard}. For non-integer coefficients, we can only decode an integer combination and the remainder acts like additional noise. Despite this penalty, compute-and-forward is the best strategy in the moderate transmit power regime. Compress-and-forward is a good strategy at low transmit power since, in this regime, the rate of the bit pipes exceeds the MIMO capacity between the transmitters and the relays. Therefore, the effective noise introduced by vector quantization at the relays does not significantly degrade the effective end-to-end SNR. At high transmit power, this effective noise becomes a significant factor. Decode-and-forward is not as efficient as compute-and-forward at high transmit power as the relays must either treat one of the messages as noise or decode both. However, it outperforms compute-and-forward in the low transmit power regime since it is able to perform joint decoding.\footnote{We do not know how to naturally fit joint decoding into the compute-and-forward framework so we have excluded it (even in the context of multiple-access) to emphasize this fact.}

\begin{remark}Note that the encoding strategy for compute-and-forward does not depend on the choice of equation coefficients at the relay. Therefore, one can obtain the maximum of the best equation rate and the best non-zero equation rate with the same strategy simply by disallowing certain coefficients at the relays past an appropriate $P$.
\end{remark}
\begin{remark}
Since the channel from the transmitters to the relays is essentially a $2$-user interference channel, it may be useful to have each transmitter send out a public and a private message as in the Han-Kobayashi scheme \cite{hk81}. Such a scheme might improve the performance of both the decode-and-forward strategy and the compute-and-forward strategy (by employing superposition as in Section \ref{s:superposition}).
\end{remark}

\section{Upper Bound}

In this section, we give a simple upper bound on the computation rate through a genie-aided argument. This bound does not match our achievable strategy in general and it may be possible to construct tighter outer bounds by taking into account the mismatch between the desired function and the function naturally provided by the channel.

\begin{theorem} \label{t:upper} Assume the channel between the transmitters and the relays is $p(y_1, \ldots, y_M | x_1, \ldots, x_L)$. If the relays, want equations with coefficient vectors $\ab_1, \ldots, \ab_M \in \Zbb^L$, the message rates are upper bounded as follows:
\begin{align*}
R_\e \leq \min_{m: a_{m\e} \neq 0}I(X_\e;Y_m |X_1, \ldots, X_{\e-1},X_{\e + 1}, \ldots, X_L)
\end{align*} For the real-valued Gaussian channel model considered in this paper, with channel vectors $\hb_1, \ldots, \hb_M \in \Rbb^L$, this specializes to 
\begin{align}
R_\e \leq \min_{m: a_{m\e} \neq 0} \frac{1}{2}\log{\left(1 + h_{m \e}^2 P\right)} \ . 
\end{align}
Similarly, for the complex-valued Gaussian channel model considered in this paper, with coefficient vectors $\ab_1, \ldots, \ab_M \in \ZC^L$ channel vectors $\hb_1, \ldots \hb_M \in \Cbb^L$, we have that
\begin{align}
R_\e \leq \min_{m: a_{m\e} \neq 0} \log{\left(1 + |h_{m \e}|^2 P\right)} \ . 
\end{align}
\end{theorem}
\begin{proof}
To each relay $m$ for which $a_{m \e} \neq 0$, we provide all messages except that from encoder $\e$ as genie-aided side-information. Now, we are left with a multicasting problem from encoder $\e$ to all relays with $a_{m \e} \neq 0$. Clearly, the multicast rate is upper bounded by the lowest rate link. For the Gaussian case, it is easy to show that the mutual information expressions are maximized by the Gaussian distribution.
\end{proof}

\section{Conclusions}

In this paper, we have developed a new coding scheme that enables relays to reliably recover equations of the original messages by exploiting the interference structure of the wireless channel. As we have seen, this framework can achieve end-to-end rates across an AWGN network that are not accessible with classical relaying strategies. More generally, the techniques in this paper can be used as building blocks for developing new cooperative communication schemes that exploit both the algebraic and statistical properties of wireless networks. Here, we presented an application to distributed MIMO and we believe there are many other scenarios where it will be useful. For instance, it can reduce energy consumption for gossiping over a sensor network \cite{ndg11} and improve the performance of low-complexity MIMO receiver architectures \cite{zneg11}.

Compute-and-forward also adds to the growing pile of evidence that structured codes are a powerful tool for tackling problems in multi-user information theory. Recently, many new inner bounds have emerged that take advantage of the algebraic structure of multi-user problems. The behavior observed in these strategies is not well-captured by the usual cut-set outer bounds. Therefore, new outer bounds that account for algebraic as well as statistical structure will be needed to better characterize the capacity regions of multi-user networks \cite{eo09}. An interesting direction for future study, inspired by the work of Avestimehr, Diggavi, and Tse on deterministic models \cite{adt11}, is whether compute-and-forward can be used to closely approximate the capacity of an AWGN network.

\section*{Acknowledgment}

The authors would like to thank G. Reeves for pointing out Theorem \ref{t:optcompreal} when this work was in an early stage. They would also like to thank G. Bresler, U. Erez, S. Shamai, and R. Zamir for valuable discussions as well as the anonymous reviewers whose comments improved the presentation of this work.

\appendices

\section{Upper Bound on Noise Densities} \label{s:gaussiannoise}

In this appendix, we demonstrate that the densities of the noise terms in Theorem \ref{t:latcompreal} and \ref{t:latcompcomplex} are upper bounded by the density of an i.i.d. Gaussian vector. The proof follows that of Lemmas 6 and 11 from \cite{ez04}. 
\begin{lemma} \label{l:gaussiannoise}
Let $\zb \sim \mathcal{N}(\mathbf{0},\mathbf{I}^{n\times n})$ and let $\db_\e$ be independently generated according to a uniform distribution over $\Vm$, the fundamental Voronoi region of $\Lambda$. Also, let $\sigma_{\mathcal{B}}^2$ denote the second moment of an n-dimensional ball whose radius is equal to the covering radius $\rcov$ of  $\L$ and let $\zb_\e^*$ be independently generated according to $\mathcal{N}(\mathbf{0},\sigma_{\mathcal{B}}^2 \mathbf{I}^{n \times n})$. Now, let 
\begin{align}
\mathbf{z}_{eq} = \alpha \zb + \sum_{\e = 1}^L{\theta_\e \db_\e}
\end{align} where $\alpha, \theta_\e \in \Rbb$. There exists an i.i.d. Gaussian vector
\begin{align}
\mathbf{z}^* = \alpha \zb + \sum_{\e = 1}^L{\theta_\e \zb^*_\e}
\end{align} with variance $\sigma^2$ satisfying
\begin{align}
\sigma^2 \leq \alpha^2+ \left(\frac{\rcov}{\reffec}\right)^2 P  \sum_{\e=1}^L{\theta_\e^2}
\end{align} such that the density of $\zb_{eq}$ is upper bounded as follows:
\begin{align}
f_{\zb_{eq}}(\zb) &\leq e^{Lc(n)n}f_{\zb^*}(\zb) \\
c(n) &= \ln{ \left(\frac{\rcov}{\reffec}\right)} + \onehalf \ln{2 \pi e G_{\Bm}^{(n)}} + \onen \label{e:densityexp}
\end{align} where $\ln$ is the natural logarithm, $G_{\Bm}^{(n)}$ is the normalized second moment of an n-dimensional ball, and $\reffec$ is the effective radius of $\L$.
\end{lemma}

\begin{proof}
First, we will show that the density of $\zb_{eq}$ is upper bounded as desired. From Lemma 11 in \cite{ez04}, we have that
\begin{align}
f_{\db_{\e}}(\zb)  \leq e^{c(n)n}f_{\zb^*_{\e}}(\zb)  \ . \label{e:ditherbound}
\end{align} Since $\zb, \db_1, \ldots, \db_L$ are independent, we can write the density of $\zb_{eq}$ as an $n$-dimensional convolution of the densities of its components,
\begin{align}
f_{\zb_{eq}}(\zb) = f_{\alpha\zb}(\zb) * f_{\theta_1\db_{1}}(\zb) * \cdots * f_{\theta_L\db_{L}}(\zb) \ . 
\end{align} Similarly, we can write the density of $\zb^*$ as
\begin{align}
f_{\zb^*}(\zb) = f_{\alpha\zb}(\zb) * f_{\theta_1\zb^*_{1}}(\zb) * \cdots * f_{\theta_L\zb^*_{L}}(\zb) \ . 
\end{align} Since probability densities are non-negative, we can use the upper bound in (\ref{e:ditherbound}) to get
\begin{align}
f_{\alpha \zb}(\zb) * f_{\theta_\e \db_\e}(\zb) \leq f_{\alpha \zb}(\zb) * e^{c(n)n} f_{\theta_\e \zb^*_\e}(\zb) \ .
\end{align} Applying this idea $L$ times to $f_{\zb_{eq}}(\zb)$ yields
\begin{align}
f_{\zb_{eq}}(\zb) &\leq e^{Lc(n)n}f_{\zb^*}(\zb) \ . 
\end{align} We must now upper bound the variance of $\zb^*$. By Definition \ref{d:reffec}, $\mathsf{Vol}(\Bm(\reffec)) = \mathsf{Vol}(\Vm)$. Recall that a ball has the smallest second moment for a given volume. Let $\mathbf{b}$ be generated according to the uniform distribution over $\Bm(\rcov)$. It follows that
\begin{align}
P &= \onen E\left[ \| \db_\e \|^2\right] \\
& \geq \onen E\left[\left\| \frac{\reffec}{\rcov} \mathbf{b} \right\|^2\right] = \left(\frac{\reffec}{\rcov}\right)^2 \sigma_{\Bm}^2 \ . 
\end{align} Finally, we get
\begin{align}
\sigma^2 &= \onen E\left[\| \alpha \zb \|^2 \right] + \onen \sum_{\e =1}^L{E\left[\|\theta_\e \zb^*_\e\|^2\right]} \\
&= \alpha^2 + \sigma_{\Bm}^2 \sum_{\e = 1}^L{\theta_\e^2}\\
&\leq \alpha^2 + \left(\frac{\rcov}{\reffec}\right)^2  P \sum_{\e = 1}^L{\theta_\e^2} \ . 
\end{align}
\end{proof}

Since the coarse lattice is good for covering and for quantization, $\frac{\rcov}{\reffec} \rightarrow 1$ and $G_{\Bm}^{(n)} \rightarrow \frac{1}{2\pi e}$ as $n \rightarrow \infty$. Therefore, $c(n) \rightarrow 0$ as $n \rightarrow \infty$. As we will show in the next appendix, the fine lattices are good for AWGN, which means that they can attain a positive error exponent for i.i.d. Gaussian noise whose variance is smaller than their respective second moments. 

\section{Fine Lattices are Good for AWGN} \label{s:awgngood}

We now show that the fine lattices from Section \ref{s:latticeconstruct} can recover from i.i.d. Gaussian noise.

\begin{lemma} \label{l:awgngood}
$\L_1,\L_2,\ldots, \L_L$ are good for AWGN with probability that goes to $1$ as $n \rightarrow \infty$ so long as $\frac{n}{p} \rightarrow 0$. 
\end{lemma}
\begin{proof}
Recall that the coarse lattice $\Lambda$ is good for AWGN. Let $\mathcal{\tilde{C}}_\ell$ be a codebook consisting of length $n$ codewords randomly and independently generated according to a uniform distribution over $\Vm$, the fundamental Voronoi region of $\Lambda$. Let $\mathbf{\tilde{z}}_\e$ denote an i.i.d. Gaussian vector with zero-mean and any variance $\sigma_\e^2$ such that the volume-to-noise ratio $\mu(\Lambda_\ell, \epsilon) = \frac{(\V(\Vm))^{2/n}}{\sigma_\e^2}$ is greater than $2\pi e$.  Consider the following channel from $\mathbf{\tilde{x}}_\ell \in \mathcal{\tilde{C}}_\ell$ to $\mathbf{\tilde{y}}_\e \in \Vm$:
\begin{align}
\mathbf{\tilde{y}}_\e = \left[\mathbf{\tilde{x}}_\e + \mathbf{\tilde{z}}_\e \right]\modl
\end{align} and let $\epsilon_\e$ the probability that $\mathbf{\tilde{x}}_\e$ is incorrectly decoded from $\mathbf{\tilde{y}}_\e$.
As part of the proof of Theorem 5 in \cite{ez04}, it is shown that the random coding error exponent for this channel is equal to the Poltyrev exponent (see Equation (56) in \cite{ez04}). This means that $\epsilon_\e$ decreases exponentially with $n$ for volume-to-noise ratio greater than $2 \pi e$. Appendix C of \cite{ez04} shows that the same performance is possible via Euclidean decoding if $\mathbf{\tilde{x}}_\e$ is drawn according to a uniform distribution over $\{p\n \L\} \cap \Vm$ and $\frac{n}{p} \rightarrow 0$. 

From Lemma \ref{l:latprops}, we know that the marginal distribution of each element of $\Lambda_\e \cap \Vm$ is uniform over $\{p\n \L\} \cap \Vm$. Furthermore, all points in the set $\L_\e \cap \Vm$ are pairwise independent. This is all that is required to apply the union bound and obtain the same performance as i.i.d. inputs over $\{p\n \L\} \cap \Vm$ in terms of the error exponent.

Thus, the probability that $\L_\e$ is good for AWGN (with the Poltyrev error exponent) goes to $1$ as $n \rightarrow \infty$. It follows from the union bound that $\L_1, \ldots, \L_L$ are simultaneously good for AWGN with high probability as $n \rightarrow \infty$.
\end{proof}

\section{Fixed Dithers} \label{s:fixeddithers}

We now show that there exist fixed dithers that are appropriate for our coding scheme. Instead of setting the second moment of $\Lambda$ to $P$, we will set the covering radius $\rcov$ to $\sqrt{nP}$. Recall that the covering radius is chosen such that the resulting ball $\Bm(\rcov)$ includes every point of the fundamental Voronoi region $\Vm$. Therefore, setting $\rcov = \sqrt{nP}$ guarantees that every transmission $\xb_\ell \in \Vm$ satisfies the power constraint. We now show that the rate loss can be made arbitrarily small. 

The effective radius $\reffec$ is chosen such that $\V(\Vm) = \V(\Bm(\reffec))$. Recall that, for even $n$, the volume of an $n$-dimensional ball of radius $1$ is  
\begin{align}
\V(\Bm(1)) = \frac{\pi^{n/2}}{(n/2)!} \ .
\end{align} By Stirling's approximation, for any $\delta > 0$ and $n$ large enough, this is lower bounded by 
\begin{align}
\V(\Bm(1)) \geq \left(\frac{2\pi e}{n(1+\delta)}\right)^{n/2} \ . 
\end{align} Thus, for any $\delta$ and $n$ large enough, the volume of $\V$ satisfies \begin{align*}
\V(\Vm) = \V(\Bm(\reffec)) &=\left(\frac{\reffec}{\rcov}\right)^n \V(\Bm(\rcov)) \\
&\geq \left(\frac{\reffec}{\rcov}\right)^n  \left(\frac{2\pi e ~\rcov^2}{n(1+\delta)}\right)^{n/2} \\
&= \left(\frac{\reffec}{\rcov}\right)^n  \left(\frac{2\pi e ~P}{(1+\delta)}\right)^{n/2}  \ . 
\end{align*} Since $\Lambda$ is also good for covering, we can choose $n$ large enough such that $\rcov^2/\reffec^2 > 1/(1+\delta)$. Finally, we have that  
\begin{align}
\V(\Vm) \geq  \left(\frac{2\pi e ~P}{(1+\delta)^2}\right)^{n/2}  \ . 
\end{align} Substituting this bound into (\ref{e:volv}), we can see that this only reduces the rate by an additional $\log(1 + \delta)$ bits, which can be made arbitrarily small through our choice of $\delta$. 

Note that the probability of error decays exponentially in $n$ averaged over the randomness in the dither vectors and the noise. Therefore, for $n$ large enough, there is at least one good fixed set of dither vectors that attains the desired probability of error $\epsilon$.

\bibliographystyle{ieeetr}

\begin{IEEEbiography}{Bobak Nazer} received the B.S.E.E. degree from Rice University, Houston, TX, in 2003, the M.S. degree from the University of California, Berkeley, CA, in 2005, and the Ph.D degree from the University of California, Berkeley, CA, in 2009, all in electrical engineering. 

He is currently an Assistant Professor in the Department of Electrical and Computer Engineering at Boston University, Boston, MA. From 2009 to 2010, he was a postdoctoral associate in the Department of Electrical and Computer Engineering at the University of Wisconsin, Madison, WI. His research interests are in network information theory and statistical signal processing, with applications to wireless networks and distributed, reliable computation.
 
Dr. Nazer received the Eli Jury award from the EECS Department at UC - Berkeley in 2009 for his dissertation research and a Dean's Catalyst Award from Boston University in 2011. He is a member of Eta Kappa Nu, Tau Beta Pi, and Phi Beta Kappa. 
\end{IEEEbiography}

\begin{IEEEbiography}{Michael Gastpar} received the Dipl. El.-Ing. degree from the Swiss Federal Institute of Technology (ETH), Zurich, in 1997, the M.S. degree from the University of Illinois at Urbana-Champaign, Urbana, in 1999, and the
Doctorat \`es Science degree from the Ecole Polytechnique F\'ed\'erale, Lausanne, Switzerland (EPFL), in 2002, all in electrical engineering. He was also a student in engineering and philosophy at the Universities of Edinburgh and Lausanne.

He is currently a Professor in the School of Computer and Communication Sciences, Ecole Polytechnique F\'ed\'erale, Lausanne, Switzerland, and an Associate Professor in the Department of Electrical Engineering and Computer Sciences, University of California, Berkeley. He also holds a faculty position at Delft University of Technology and was a researcher at Bell Labs, Lucent Technologies, Murray Hill, NJ. His research interests are in network information theory and related coding and signal processing techniques, with applications to sensor networks and neuroscience.

Dr. Gastpar won the 2002 EPFL Best Thesis Award, an NSF CAREER award in 2004, and an Okawa Foundation Research Grant in 2008. He is an Information Theory Society Distinguished Lecturer (2009--2011). He is currently an Associate Editor for Shannon Theory for the IEEE TRANSACTIONS ON INFORMATION THEORY, and he has served as Technical Program Committee Co-Chair for the 2010 International Symposium on Information Theory, Austin, TX. \end{IEEEbiography}

\end{document}